\documentclass[journal]{IEEEtran}

\usepackage{color}
\usepackage{amsthm}
\usepackage{amsmath}
\usepackage{amssymb}
\usepackage{mathrsfs}
\usepackage{wasysym}
\usepackage{cite}
\usepackage{graphicx}
\usepackage{wrapfig}

\newif\ifitsdraft

\usepackage{ifthen}
\newboolean{draftPaper}
\setboolean{draftPaper}{true}
\newcommand{\nDraft}[1]{
	\ifthenelse{\boolean{draftPaper}}{}{#1}
}
\newcommand{\Draft}[2]{
	\ifthenelse{\boolean{draftPaper}}{#1}{#2}
}
\allowdisplaybreaks
\definecolor{darkgreen}{rgb}{0,0.5,0}
\newcommand\red[1]{{\color{red}#1}}

\newtheorem{theorem}{Theorem}[section]
\newtheorem{definition}[theorem]{Definition}
\newtheorem{proposition}[theorem]{Proposition}
\newtheorem{assumption}[theorem]{Assumption}
\newtheorem{example}[theorem]{Example}
\newtheorem{remark}[theorem]{Remark}

\newtheorem{lemma}[theorem]{Lemma}
\renewcommand\thetheorem{\arabic{section}.\arabic{theorem}}
\usepackage{enumitem}
 \usepackage{balance}
\input{rgsMacros.sty}
\definecolor{darkgreen}{rgb}{0,0.5,0}
\begin{document}

\title{\LARGE \textbf{Certifying the LTL Formula p Until q in Hybrid Systems}}

\author{Hyejin Han, Mohamed Maghenem, and Ricardo G. Sanfelice
\thanks{H. Han and R. G. Sanfelice are with the Department of Electrical and Computer Engineering, University of California, Santa Cruz. Email: 
hhan7,ricardo@ucsc.edu. M. Maghenem is with CNRS France, D{\'e}l{\'e}gation Alpes. Email : Mohamed.Maghenem@gipsa-lab.grenoble-inp.fr}
\thanks{This research has been partially supported by NSF Grants no. ECS-1710621, CNS-1544396, and CNS-2039054, by AFOSR Grants no. FA9550-19-1-0053, FA9550-19-1-0169, and FA9550-20-1-0238, and by CITRIS and the Banatao Institute at the University of California.}}

\maketitle

\begin{abstract}
In this paper, we propose sufficient conditions to guarantee that a linear temporal logic (LTL) formula of the form 
p Until q, denoted by $p \,\mathcal{U} q$, 
is satisfied for a hybrid system. 
Roughly speaking, the formula $p \,\mathcal{U} q$ is satisfied means that the solutions, initially satisfying proposition p, keep satisfying this proposition until proposition q is satisfied. To certify such a formula, connections to invariance notions such as conditional invariance (CI) and eventual conditional invariance (ECI), as well as finite-time attractivity (FTA) are established. As a result, sufficient conditions involving the data of the hybrid system and an appropriate choice of Lyapunov-like functions, such as barrier functions, are derived. The considered hybrid system is given in terms of differential and difference inclusions, which capture the continuous and the discrete dynamics present in the same system, respectively. Examples illustrate the results throughout the paper.
\end{abstract}

\IEEEpeerreviewmaketitle

\nDraft{\vspace{-0.1in}}
\section{Introduction}

\subsection{Background and Motivation}

LTL is a language used to express complex temporal properties of dynamical systems in terms of formulas. Each LTL formula is composed of a set of propositions related by temporal and logical operators. 
A required temporal property, also called  \textit{specification}, is guaranteed for a dynamical system if and only if the corresponding LTL formula is true along the solutions to the considered system. Hence, LTL provides a  framework to formulate complex dynamical properties, that go beyond  stability, convergence, or safety \cite{thistle1986control,girard2006verification,wongpiromsarn2012receding,kwon2008ltlc}.
LTL has been used as a tool to analyze formal specifications for dynamical systems. For examples, in \cite{wongpiromsarn2012receding}, 
LTL is employed to express the 
safety-plus-stability specification and to ensure that the system satisfies it by executing a finite state automata. 
In \cite{giorgetti2005bounded}, the desired behavior of the system is expressed using LTL, and model-checking based algorithms are used to certify the given formula or to provide a counterexample. Furthermore, in \cite{tumova2016multi}, motion planning for multi-agent systems with safety is achieved by verifying corresponding LTL formulas in real time.

A widely used approach to certify formulas along solutions to dynamical systems consists in using model-checking approaches~\cite{baier2008principles,
tabuada2009verification,belta2017formal}. In such approaches, the system is modeled as a finite (or infinite) state automaton and existing model-checking algorithms are able to answer about the satisfaction of the formula for the automaton. The disadvantage of these approaches relies on their decidability; namely, whether the satisfaction of the formula for the automaton is equivalent to its satisfaction for the actual system \cite{tabuada2009verification}. 
Furthermore, when the system exhibits hybrid phenomena, this problem remains mostly unsolved. In other works such as in \cite{platzer2010logical}, theorem provers are developed to analyze LTL formulas by simulating a discretized version of the system. As in every numerical tool, the sensitivity to the discretization step and the dimension of the system is a drawback. 

To avoid the limitations of model checking, other works use analytical approaches based on analyzing specific formulas using Lyapunov-like techniques. 
In \cite{bisoffi2020satisfaction}, a class of 
LTL formulas is modeled as a finite state transition system, which combined to the actual continuous-time control system form a hybrid control system. The considered formula is verified by guaranteeing a recurrence property for the closed-loop system. This recurrence property is certified using Lyapunov-like sufficient conditions. 
In \cite{wongpiromsarn2015automata}, an approach combining automata-based tools and barrier certificates is introduced to avoid computing finite-state abstractions of the system. 
This approach provides a collection of barrier functions to certify the considered LTL formulas. In \cite{han2020linear}, sufficient conditions to certify basic formulas for hybrid systems using Lyapunov-like functions are proposed. In particular, the \emph{always p} and \emph{eventually p} formulas are mainly considered and guaranteed by guaranteeing forward invariance and finite-time attractivity, respectively. 

For LTL formulas involving  \emph{until} operators, we distinguish two versions of the \emph{until} operator; namely, the strong until (denoted $\mathcal{U}_s$) and the weak until (denoted $\mathcal{U}_w$); see \cite{emerson1990temporal, eisner2005topological, piterman2018temporal}. That is, given two propositions $p$ and $q$, the satisfaction of the formula $p \,\mathcal{U}_s q$ implies that proposition $p$ is true until $q$ happens to be true, and $q$ must become true eventually. For the weak version, the satisfaction of the formula $p \,\mathcal{U}_w q$ implies that proposition $p$ is true until $q$ happens to be true; however, $q$ is not required to become true as long as $p$ remains true. Until operators are one of the basic operators in LTL. A simple, but concrete, application example where until operators are used concerns the autonomous navigation problem in a constrained environment \cite{Wolff.ea.14.ICRA}. In such applications, a mobile vehicle 
typically navigates its environment without colliding with obstacles and while following a particular sequence of tasks. For instance, we consider the simple situation where a vehicle needs to exit a room only via its exit door, so the vehicle should not hit the boundaries of the room until it reaches the door. Such a temporal behavior can be expressed in terms of an LTL formula involving the until operator; namely, ``the vehicle stays in the interior\,($p$)'' 
\emph{until} ``it reaches door\,($q$)''.

\nDraft{\vspace{-0.15in}}
\subsection{Contributions}

In this paper, sufficient conditions for the satisfaction of the LTL formulas $p \,\mathcal{U}_w q$ and $p \,\mathcal{U}_s q$ along solutions to hybrid systems are proposed. The proposed sufficient conditions are infinitesimal and involve only the data of the hybrid system and appropriate choice of Lyapunov-like functions, without requiring the computation of the solutions nor discretization of the right-hand side. A key intermediate step to deduce the proposed sufficient conditions consists in establishing sufficient and equivalent relationships between the satisfaction of the considered formulas and the following dynamical properties:

\begin{enumerate}[leftmargin=0.2in]
\item  \textit{Conditional Invariance} (CI). This property suggests that the solutions to the hybrid system remain in a set if they start from a (likely different) set \cite{maghenem2020sufficient}. This property is also called safety in \cite{prajna2007framework}. 

\item \textit{Eventual Conditional Invariance} (ECI). This property suggests that the solutions reach a given set in finite time and remain in it provided that they start from a (likely different) given set \cite{ladde1972analysis}. 

\item \textit{Finite-Time Attractivity} (FTA). This property suggests that the solutions reach a set in finite time provided that they start from a (likely different) set \cite{215}. 
\end{enumerate}

For the formula $p \,\mathcal{U}_w q$, we present sufficient conditions using barrier functions tailored to CI for hybrid systems. Furthermore, for the formula 
$p \,\mathcal{U}_s q$, we present sufficient conditions using Lyapunov-like functions tailored to CI plus ECI or FTA for hybrid systems. The contributions made by this paper are as follows:

\begin{itemize}[leftmargin=0.2in]
\item In Section \ref{sec:characterization}, we propose  characterizations of the satisfaction of the formulas $p \,\mathcal{U}_w q$ and $p \,\mathcal{U}_s q$ for hybrid systems. Specifically, given a hybrid system, denoted $\mathcal{H}$, we establish that CI for an auxiliary version of $\mathcal{H}$, denoted $\mathcal{H}_w$, is sufficient for the satisfaction of the formula 
$p \,\mathcal{U}_w q$ for $\mathcal{H}$.
Moreover, we show that the formula 
$p \,\mathcal{U}_s q$ is verified for $\mathcal{H}$ if $p \,\mathcal{U}_w q$ is satisfied and an auxiliary system, denoted $\mathcal{H}_s$, exhibits an ECI property. Finally, we present that the formula $p \,\mathcal{U}_s q$ for $\mathcal{H}$ is satisfied when the satisfaction of $p \,\mathcal{U}_w q$ for $\mathcal{H}$ is guaranteed and $\mathcal{H}_s$ exhibits a FTA property.
	
\item In Section \ref{Section.CIECI}, sufficient conditions to guarantee CI and ECI for hybrid systems are proposed. 
Inspired by \cite[Theorem 3.4]{ladde1972analysis}, we formulate sufficient conditions for ECI for hybrid systems; see Theorem \ref{thm:pre_eventual_ci}. 
The conditions in \cite[Theorem 3.4]{ladde1972analysis} are tightened by restricting the class of hybrid systems under study to capture both continuous and discrete dynamics.
Moreover, we propose sufficient conditions for ECI that are inspired by the work in \cite{182} for hybrid observers; see Theorem \ref{thm:pre_ECI_interval}.
The conditions in Theorem \ref{thm:pre_ECI_interval} guarantee ECI when we know approximately a bound on the length of the flow interval between successive jumps.

\item In Section \ref{SectionFTA}, inspired by \cite[Proposition B.1 and Proposition B.3]{han.arXiv19}, sufficient conditions that guarantee FTA for hybrid systems are introduced.
In particular, using two Lyapunov functions, we propose sufficient conditions that guarantee pre-FTA for hybrid systems.

\item In Section \ref{sec:sufficient_conditions_until}, sufficient Lyapunov-like conditions to guarantee satisfaction of the formulas $p \,\mathcal{U}_w q$ and $p \,\mathcal{U}_s q$ are deduced by exploiting the equivalence relationships established in Section \ref{sec:characterization}. Specifically, sufficient conditions for the satisfaction of $p \,\mathcal{U}_w q$ are proposed by using the conditions guaranteeing CI in Proposition \ref{prop:conditional_invariance}. Sufficient conditions for the satisfaction of $p \,\mathcal{U}_s q$ are proposed using conditions guaranteeing ECI, such as those in Theorem \ref{thm:ECIf}, Theorem \ref{thm:ECIj}, Theorem \ref{thm:pre_eventual_ci}, Theorem \ref{thm:pre_ECI_interval}, and Theorem \ref{thm:eventual_ci}.
 Similarly, sufficient conditions for the satisfaction of $p \,\mathcal{U}_s q$ are proposed using conditions guaranteeing FTA, such as those in Theorem \ref{thm:FTAf}, Theorem \ref{thm:FTAj}, Theorem \ref{thm:pre_FTA} and Theorem \ref{thm:FTA}.

\end{itemize}

A preliminary version of this work is in \cite{215}, where detailed proofs have been omitted and only fewer results are included.
In particular, Theorem \ref{thm:pre_ECI_interval} and results employing a FTA property are new.

The remainder of this paper is organized as follows. Preliminaries on hybrid systems and temporal logic are in Section \ref{sec:preliminaries}. 
Characterization of the formulas $p \,\mathcal{U}_w q$ and $p \,\mathcal{U}_s q$ using CI, ECI, and FTA are in Section \ref{sec:characterization}. Sufficient conditions to guarantee CI and ECI are in Section \ref{Section.CIECI}.
Finally, sufficient conditions to certify the formulas $p \,\mathcal{U}_w q$ and $p \,\mathcal{U}_s q$ are in Section \ref{sec:sufficient_conditions_until}. Academic examples are provided all along the paper to illustrate concepts and results.

\textbf{Notation.} Let $\mathbb{R}_{\geq 0} := [0, \infty) $ and $\mathbb{N} := \left\{0,1,\ldots, \infty \right\}$. For $x$, $y \in \mathbb{R}^n$, $x^\top$ denotes the transpose of $x$, $|x|$ the Euclidean norm of $x$, $|x|_K := \inf_{y \in K} |x-y|$ defines the distance between $x$ and the nonempty set $K$, and $\langle x,y \rangle = x^\top y$ denotes the inner product between $x$ and $y$. For a set $K \subset \mathbb{R}^n$, we use $\mbox{int} (K)$ to denote its interior, $\partial K$ to denote its boundary, $\cl(K)$ to denote its closure, and $U(K)$ to denote any open neighborhood of $K$. For a set $Q \subset \mathbb{R}^n$, $K \backslash Q$ denotes the subset of elements of $K$ that are not in $Q$. By $\mathcal{C}^1$, we denote the set of continuously differentiable functions. By $F : \mathbb{R}^m \rightrightarrows \mathbb{R}^n$, we denote a set-valued map associating each element $x \in \mathbb{R}^m$ to a subset $F(x) \subset \mathbb{R}^n$. For a real number $s \in \mathbb{R}$, $\mbox{ceil}(s)$ denotes the smallest integer upper bound for $s$. \blue{For a scalar function 
$V: \dom V \rightarrow \mathbb{R}$, 
$L_V(r) := \{x \in \dom V : V(x) \leq r\}$, for some $r \in [0, \infty]$, is the $r$-sublevel set of $V$}.

\section{Preliminaries}
\label{sec:preliminaries}
\subsection{Hybrid Systems}

Following the modeling framework proposed in \cite{goebel2012hybrid}, we consider hybrid systems modeled as
\begin{align} \label{eqn:H}
\mathcal{H}: & \left\{ \begin{array}{ll}
 \dot{x}\phantom{{}^+} \!\!\in F(x) & \qquad x \in C  \\
x^+ \!\!\in G(x)  & \qquad x \in D \mbox{,}
\end{array} \right.
\end{align}
with the state variable $x \in \mathbb{R}^n$, the flow set $C \subset \mathbb{R}^n$, the jump set $D \subset \mathbb{R}^n$, and the flow and jump maps, respectively, $F: \mathbb{R}^n \rightrightarrows \mathbb{R}^n$ and $G: \mathbb{R}^n \rightrightarrows \mathbb{R}^n $. 

A hybrid arc $\phi$ is defined on a hybrid time domain denoted $ \dom \phi \subset \mathbb{R}_{\geq 0} \times \mathbb{N}$. The hybrid arc $\phi$ is parameterized by an ordinary time variable $t \in \mathbb{R}_{\geq 0}$ and a discrete jump variable $j \in \mathbb{N}$. A hybrid time domain $\dom \phi$ is such that for each $(T,J) \in \dom \phi$, $\dom \phi \cap \left( [0,T] \times \left\{ 0, 1, \ldots, J \right\} \right) = \cup^{J}_{j=0} \left([t_j,t_{j+1}] \times \left\{j\right\} \right)$ for a sequence $\left\{ t_j \right\}^{J+1}_{j=0}$, such that $t_{j+1} \geq t_j$ and $t_0 = 0$.
Note that the structure of a hybrid time domain $\dom \phi$ is such that, given $(t,j), (t',j') \in \dom \phi$, $t+j \leq t'+j'$ if $t \leq t'$ and $j \leq j'$.

\begin{definition} [Concept of solution to $\mathcal{H}$] 
\label{solution definition} \index{forward solution to a hybrid system} A hybrid arc $\phi : \dom \phi \to \mathbb{R}^n$ is a {\em solution} to $\mathcal{H}$ if
\begin{itemize}
\item[(S0)] $\phi(0,0) \in \cl(C) \cup D$;
\item[(S1)] for all $j \in \mathbb{N}$ such that $I^j:=\left\{t : (t,j) \in\dom \phi \right\}$ has nonempty interior, $t \mapsto \phi(t,j)$ is locally absolutely continuous and
\begin{equation*}
\label{S1}
\begin{array}{ll}               
\phi(t,j)\in C &\quad \mbox{for all } t\in \mbox{int}(I^j)\mbox{,}\\
\dot{\phi}(t,j)\in F(\phi(t,j)) &\quad \mbox{for almost all } t\in I^j\mbox{;}\\
\end{array} 
\end{equation*}                          
\item[(S2)] for all $(t,j)\in \dom \phi$ such that $(t,j+1)\in \dom \phi$,
\begin{equation*} \label{S2}        
\begin{array}{l}         
\phi(t,j)\in D, \qquad
\phi(t,j+1)\in G(\phi(t,j))\mbox{.}
\end{array}
\end{equation*}                 
\end{itemize}     
\end{definition}

\begin{remark}
According to Definition \ref{solution definition}, when a maximal solution $\phi$ to $\mathcal{H}$ starting from a point in $C$ reaches a point in $\cl(C) \backslash C$ after flowing, it cannot flow back to the set $C$ unless it jumps from the point in $\cl(C) \backslash C$. However, if it reaches $\cl(C) \backslash C$ after a jump, then it is allowed to flow from a point in $\cl(C) \backslash C$ within $C$; see Examples \ref{ex:until_ci} and \ref{counterexpECI} below.
\end{remark}

A solution $\phi$ to $\mathcal{H}$ is said to be maximal if there is no solution $\phi'$ to $\mathcal{H}$ such that $\phi(t,j) = \phi'(t,j)$ for all $(t,j) \in \dom \phi$ with $\dom \phi$ a proper subset of $\dom \phi'$. It is said to be nontrivial if $\dom\phi$ contains at least two elements.
A solution $\phi$ is said to be complete if its domain is unbounded.
It is Zeno if it is complete and $\sup_t \dom\phi < \infty$. It is eventually discrete if $T \!=\! \sup_t\dom\phi < \infty$ and $\dom \phi \cap (\{T\} \times \mathbb{N})$ contains at least two elements. It is genuinely Zeno if it is Zeno, but not eventually discrete.
See \cite{goebel2012hybrid} for more details about hybrid dynamical systems.

For convenience, we define the range of a solution $\phi$ to $\mathcal{H}$ as $\rge \phi \!=\! \{\phi(t,j) : (t,j) \!\in\! \dom\phi\}$.
We use $\mathcal{S}_{\mathcal{H}}(x)$ to denote the set of maximal solutions to $\mathcal{H}$ starting from $x \in \cl(C) \cup D$.
Given a set $\mathcal{A} \subset \mathbb{R}^n$, $\mathcal{R}(\mathcal{A})$ denotes the (infinite-horizon) reachable set from $\mathcal{A}$; i.e., $\mathcal{R}(\mathcal{A}) := \{\phi(t,j) : \phi \in \mathcal{S_H}(\mathcal{A}), (t,j) \in \dom\phi \}$.
\begin{definition}[Settling-time function]
Given a closed set $\mathcal{A} \subset \mathbb{R}^n$ and a solution $\phi$ to $\mathcal{H}$ starting from $\cl(C) \cup D$, the settling-time function $\mathcal{T}_\mathcal{A} : \mathcal{S_H}(\cl(C) \cup D) \mapsto \mathbb{R}_{\geq 0}$ is given by
\begin{align} \label{min-tim-fun}
\mathcal{T}_\mathcal{A}(\phi) :=
\left\{
\begin{array}{ll}
    \!\infty & ~\mbox{if}~\mathcal{R}(\phi(0,0)) \cap \mathcal{A} = \emptyset  \\ 
    \displaystyle
    \!\min_{\substack{\phi(t,j) \in \mathcal{A} \\ (t,j) \in \dom\!\phi}}
        \!t+j & ~\mbox{otherwise.}
\end{array}
\right.    
\end{align}
\end{definition}
Given a solution $\phi$ to $\mathcal{H}$ starting from $\mathbb{R}^n \backslash \mathcal{A}$, the function $\mathcal{T}_\mathcal{A}$ provides (when finite) the first hybrid time at which the solution $\phi$ reaches the set $\mathcal{A}$. If the solution never reaches $\mathcal{A}$, the function $\mathcal{T}_\mathcal{A}$ returns infinity.

\subsection{LTL and Until Operators}

\red{An atomic proposition $p$, is a statement on the system state $x$ that is either \texttt{True} 
($\equiv 1$) or \texttt{False} ($\equiv 0$).}
A proposition $p$ is treated as a (single-valued) function of $x$, that is, as the function 
$x \mapsto p(x) \in \{ 0,1\}$. 
The set of all possible atomic propositions 
is denoted by $\mathcal{P}$. 
An LTL formula $f$ is a sentence consisting 
of atomic propositions as well as logical and temporal operators. In the following, we introduce the LTL formulas studied in this paper.

\begin{definition}[$p \,\mathcal{U}_s q$]
\label{def:strong_until}
\blue{Given two atomic propositions $p, q \in \mathcal{P}$, a solution $\phi$ to $\mathcal{H}$ satisfies the formula $p \,\mathcal{U}_s\, q$ if either $q(\phi(0,0)) = 1$; or, 
\begin{itemize}
\item there exists $(t^\star,j^\star) \in \dom\phi$ such that $t^\star + j^\star > 0$, $q(\phi(t^\star,j^\star)) = 1$, and $p(\phi(s,k)) = 1$ for all $(s,k) \in \dom \phi$ such that  
$0 \leq s+k < t^\star + j^\star$. 
\end{itemize}}
\end{definition}

\begin{definition}[$p \,\mathcal{U}_w q$]
\blue{Given two atomic propositions $p, q \in \mathcal{P}$, a solution $\phi$ to $\mathcal{H}$ satisfies the formula $p \,\mathcal{U}_w\, q$ if either
\begin{itemize}
	\item $p(\phi(s,k)) = 1$ for all $(s,k) \in \dom\phi$ such that $s+k \geq 0$; or
	\item $\phi$ satisfies $p \,\mathcal{U}_s\, q$.
\end{itemize}}
\label{def:weak_until}
\end{definition}

\begin{remark}
\blue{The formula $p \,\mathcal{U}_s q$ is  satisfied for $\mathcal{H}$ 
if, for each maximal solution $\phi$ to $\mathcal{H}$ with $p(\phi(0,0)) + q(\phi(0,0)) \!\geq\! 1$,\footnote{Since the functions $p$ and $q$ map to $\{0,1\}$, $p(\phi(0,0)) + q(\phi(0,0)) \!\geq\! 1$ implies that $p(\phi(0,0)) \!=\! 1$ or $q(\phi(0,0)) \!=\! 1$.} $p \,\mathcal{U}_s q$ is satisfied. Similarly, the formula $p \,\mathcal{U}_w q$ is said to be satisfied for $\mathcal{H}$ 
if, for each maximal solution $\phi$ to $\mathcal{H}$ with $p(\phi(0,0)) + q(\phi(0,0)) \!\geq\! 1$, 
$p \,\mathcal{U}_w q$ is satisfied.}
\end{remark}

\subsection{Set Invariance and Attractivity Notions}
\label{sec:notions}
 
In this section, we present the invariance and attractivity notions used in this paper. 
 
\begin{definition}[Forward (pre-)invariance \cite{maghenem2020sufficient}]
\label{def:forward_invariance}
A set $K \subset \mathbb{R}^n$ is said to be \emph{forward pre-invariant} for $\mathcal{H}$ if, for each solution $\phi \in \mathcal{S_H}(K)$, $\rge \phi \subset K$.
The set $K$ is said to be \emph{forward invariant} for $\mathcal{H}$ if it is forward pre-invariant for $\mathcal{H}$ and, for every $x \in K$, every \blue{maximal} solution $\phi \in \mathcal{S_H}(x)$ is complete.
\end{definition}

\begin{definition}[Conditional invariance \cite{magh2018conditional1}]
\label{def:conditional_invariance}
Given sets $K \subset \mathbb{R}^n$ and $\mathcal{X}_o \subset K$,
the set $K$ is said to be \emph{conditionally invariant (CI)} with respect to the set $\mathcal{X}_o$ for $\mathcal{H}$ if, for each solution $\phi \in \mathcal{S}_{\mathcal{H}}(\mathcal{X}_o)$, $\rge \phi \subset K$.
\end{definition}

\begin{remark} \label{remark:ci_fi}
Note that when $\mathcal{X}_o \!=\! K$, conditional invariance of $K$ with respect to $\mathcal{X}_o$ is equivalent to forward pre-invariance of $K$.
\end{remark}

\begin{remark}
Following the notion of safety in \cite{magh2018conditional1}, conditional invariance of $K$ with respect to $\mathcal{X}_o$ is equivalent to safety with respect to $(\mathcal{X}_o, \mathcal{X}_u)$, with $\mathcal{X}_u := \mathbb{R}^n \backslash K$ defining the region of the state space that the solutions to $\mathcal{H}$ must avoid when starting from the set of initial conditions $\mathcal{X}_o$.
\end{remark}

In the following, inspired by the ideas in \cite{ladde1972analysis} for continuous-time systems, we introduce eventual conditional invariance for hybrid systems $\mathcal{H} = (C,F,D,G)$.

\begin{definition} [Eventual conditional invariance] \label{def:eventual_ci}
Given sets $\mathcal{O} \subset \cl(C) \cup D$ and $\mathcal{A} \subset \mathbb{R}^n$, the set $\mathcal{A}$ is said to be \emph{eventually conditionally invariant (ECI)} with respect to $\mathcal{O}$ for $\mathcal{H}$ if, for each solution $\phi \in \mathcal{S_H}(\mathcal{O})$,
there exists a hybrid time $(t^\star,j^\star) \in \dom\phi$ such that $\phi(t,j) \in \mathcal{A}$ for all $(t,j) \!\in\! \dom\phi$ such that $t+j \!\geq\! t^\star+j^\star$.
\end{definition}

\red{\begin{remark}
When $\X_o \subset K$ and $K$ is CI with respect to $\X_o$, it is guaranteed that solutions starting from $\mathcal{X}_o$ remain in $\mathcal{X}_o \cup K$. 
Unlike CI, ECI of $\mathcal{A}$ with respect to $\mathcal{O}$ for $\HS$ does not imply that solutions remain in $\mathcal{O} \cup \mathcal{A}$. Indeed, even if $\mathcal{O} \subset \A$, we may still have solutions that start in $\mathcal{O}$, leave $\mathcal{O}$, and then eventually get to $\A$.
\end{remark}}

Since $\mathcal{H}$ can have maximal solutions that are not complete, we introduce the following notion which, compared to Definition \ref{def:eventual_ci}, requires that only the complete solutions to $\mathcal{H}$ must reach the set $\mathcal{A}$.

\begin{definition} [Pre-eventual conditional invariance] \label{def:pre_eventual_ci}
Given sets $\mathcal{O} \subset \cl(C) \cup D$ and $\mathcal{A} \subset \mathbb{R}^n$,
the set $\mathcal{A}$ is said to be \emph{pre-eventually conditionally invariant (pre-ECI)}\footnote{\blue{Since $\HS$ can have maximal solutions that are not complete, the pre-ECI notion requires the ECI property for complete solutions only.}} with respect to $\mathcal{O}$ for $\mathcal{H}$ if, for each complete solution $\phi \in \mathcal{S_H}(\mathcal{O})$, there exists a hybrid time $(t^\star,j^\star) \!\in\! \dom\phi$ such that $\phi(t,j) \in \mathcal{A}$ for all $(t,j) \!\in\! \dom\phi$ such that $t+j \!\geq\! t^\star + j^\star$.
\end{definition}

\begin{definition} [Finite-time attractivity]
\label{def:FTA}
Given sets $\mathcal{O} \!\subset\! \cl(C) \cup D$ and $\mathcal{A} \!\subset\! \mathbb{R}^n$ such that $\mathcal{A}$ is closed,
the set $\mathcal{A}$ is said to be \emph{finite-time attractive (FTA)} with respect to $\mathcal{O}$ for $\mathcal{H}$ if,
for each solution $\phi \in \mathcal{S}_{\mathcal{H}}(\mathcal{O})$, $\mathcal{T}_\mathcal{A}(\phi) < \infty$. 
\end{definition}

\begin{definition}[Pre-finite-time attractivity]
\label{def:pre-FTA}
Given sets $\mathcal{O} \subset \cl(C) \cup D$ and $\mathcal{A} \subset \mathbb{R}^n$ such that $\mathcal{A}$ is closed,
the set $\mathcal{A}$ is said to be \emph{pre-finite-time attractive (pre-FTA)} with respect to $\mathcal{O}$ for $\mathcal{H}$ if,
for each complete solution $\phi \in \mathcal{S}_{\mathcal{H}}(\mathcal{O})$, $\mathcal{T}_\mathcal{A}(\phi) < \infty$.
\end{definition}

\nDraft{\vspace{-0.2in}}
\subsection{Standing Assumptions}

\blue{
Our results hold under the following mild assumption on the data of the
hybrid system $\mathcal{H}$.
\begin{assumption} \label{assump:SA} 
The map $F$ is outer semicontinuous and locally bounded with nonempty and convex values on $C$ and $G$ has nonempty images on $D$.  
\end{assumption}
The properties of $F$ in Assumption \ref{assump:SA} are used in the literature of differential inclusions as mild requirements for existence of solutions from $\text{int}(C)$ plus adequate structural properties for the flows, see \cite{aubin2012differential, Aubin:1991:VT:120830}. When $F$ is single valued, the properties of $F$ in Assumption \ref{assump:SA} reduce to simply continuity.}

\section{Characterizing $p \,\mathcal{U} q$ using CI, ECI, and FTA} \label{sec:characterization}

We illustrate our approach on a constrained differential inclusion modeling the continuous-time dynamics of $\mathcal{H}$ which is given by
\begin{equation} \label{eqn:ct}
\dot{x} \in F(x) \qquad\qquad x \in C.
\end{equation}
To this end, we introduce the sets
\begin{equation}
	P \!:=\! \{x \in \mathbb{R}^n : p(x) = 1\} \mbox{, }\: Q \!:=\! \{x \in \mathbb{R}^n : q(x) = 1\}\mbox{.}
\label{eqn:K_sets}
\end{equation}
Note that $P$ and $Q$ collect the set of points where the atomic propositions $p$ and $q$ are satisfied, respectively.
For the purposes of this discussion, we impose the following condition on these sets:
\begin{itemize}
	\item The sets $P$ and $Q$ are closed and $P \subset C$.
\end{itemize}

{\setlength{\parskip}{1em}
By Definition \ref{def:weak_until} specialized to the case of flows of $\mathcal{H}$ only, when 
$p \,\mathcal{U}_w q$ is satisfied for \eqref{eqn:ct}, it follows that, for each solution $\phi$ to \eqref{eqn:ct} starting from the set $P \cup Q$, at least one of the following properties holds:}
\begin{itemize}
	\item[1)] the solution $\phi$ remains in the set $P$ for all time.
	\item[2)] the solution $\phi$ starts and remains in the set $P$ up to when it reaches the set $Q$.
   \item[3)] the solution $\phi$ starts from the set $Q$.
\end{itemize}
Note that we consider solutions to \eqref{eqn:ct} starting from the set $P \cup Q$.
Hence, based on items 1)-3), each solution to \eqref{eqn:ct} starting from the set $P \backslash Q$ needs to either remain in $P$ for all time or remain in $P$ until reaching $Q$ (if it happens);
namely, solutions starting from $P \backslash Q$ should satisfy either item 1) or item 2).
Interestingly, the satisfaction of either item 1) or item 2) can be guaranteed via CI of $P \cup Q$ with respect to $P \backslash Q$ for the following auxiliary system: 
\begin{equation}
\left\{
\begin{array}{ll}
	\dot{x} \in F(x) & \quad x \in C \backslash Q \\
	x^+ = x & \quad x \in Q \mbox{.}
\end{array}
\right.
\label{eqn:ct_m}
\end{equation}
System \eqref{eqn:ct_m} is used to characterize the behavior of system  \eqref{eqn:ct} outside the set $Q$.
In fact, when $P \cup Q$ is CI with respect to $P \backslash Q$ for \eqref{eqn:ct_m},
each solution to \eqref{eqn:ct_m} starting from $P \backslash Q$ remains in $P$ for all time or remains in $P$ until it reaches $Q$.
Indeed, the solutions to \eqref{eqn:ct} from $P \backslash Q$ are the solutions to \eqref{eqn:ct_m} (and vice versa) up to when they reach (if they do) $Q$.
Hence, $P \cup Q$ being CI with respect to $P \backslash Q$ for \eqref{eqn:ct_m} implies the satisfaction of $p \,\mathcal{U}_w q$ for \eqref{eqn:ct}.

From Definition \ref{def:strong_until} specialized to the case of flows of $\mathcal{H}$ only, when $p \,\mathcal{U}_s q$ is satisfied for \eqref{eqn:ct}, it follows that, for each solution $\phi$ to \eqref{eqn:ct} starting from $P \cup Q$, at least one of the following properties holds:
\begin{itemize}
	\item[1)] the solution $\phi$ starts and remains in the set $P$ until it reaches the set $Q$ in finite time.
	\item[2)] the solution $\phi$ starts from the set $Q$.
\end{itemize}
As a difference to $p \,\mathcal{U}_w q$, the satisfaction of $p \,\mathcal{U}_s q$ requires, additionally to $p \,\mathcal{U}_w q$ being satisfied, that every maximal solution $\phi$ to \eqref{eqn:ct} starting from $P \backslash Q$ actually reaches $Q$ in finite time.
When the set $P \cup Q$ is CI with respect to $P \backslash Q$ for \eqref{eqn:ct_m},
this property is guaranteed when $Q$ is ECI with respect to the set $P \cup Q$ for the following system
\begin{equation}
\left\{
\begin{array}{ll}
	\dot{x} \in F(x) & \quad x \in (C \backslash Q) \cap P \\
	x^+ = x & \quad x \in Q \mbox{.}
\end{array}
\right.
\label{eqn:ct_m+}
\end{equation}
System \eqref{eqn:ct_m+} is the restriction of \eqref{eqn:ct_m} to $P \cup Q$.
Hence, when, in addition, $Q$ is ECI with respect to $P \cup Q$ for \eqref{eqn:ct_m+}, each solution to \eqref{eqn:ct_m+} starting from $P \backslash Q$ reaches $Q$ in finite time.
Since the solutions to \eqref{eqn:ct_m} are the solutions to \eqref{eqn:ct_m+} (and vice versa when $p \,\mathcal{U}_w q$ is satisfied) up to when they reach $Q$,
and the solutions to \eqref{eqn:ct_m} are the solutions to \eqref{eqn:ct} (and vice versa) up to when they reach $Q$,
each solution to \eqref{eqn:ct} starting from $P \backslash Q$ reaches $Q$ in finite time and remains in $P$ until it reaches $Q$, which implies the satisfaction of $p \,\mathcal{U}_s q$ for \eqref{eqn:ct}.

As an alternative to the properties required above for $p \,\mathcal{U}_w q$ and $p \,\mathcal{U}_s q$,
we can also employ CI of $P \!\cup\! Q$ with respect to $P \!\cup\! Q$ (namely, forward invariance of $P \!\cup\! Q$) for \eqref{eqn:ct_m} instead of using CI of $P \!\cup\! Q$ with respect to $P \backslash Q$ for \eqref{eqn:ct_m}.
Moreover, to guarantee that each solution to \eqref{eqn:ct} starting from $P \backslash Q$ reaches $Q$ in finite time as needed for the satisfaction of $p \,\mathcal{U}_s q$, we can alternatively employ FTA of $Q$ with respect to $P \!\cup\! Q$ for \eqref{eqn:ct_m+}.
Results exploiting these ideas are also presented in this paper.

{\setlength{\parskip}{1em}
Now, we extend the proposed approach to hybrid systems.
Similar to the approach for \eqref{eqn:ct}, when considering the hybrid system $\mathcal{H} \!=\! (C,F,D,G)$, we impose the following condition in some of our results.}
\begin{assumption} \label{assump:SA1}
The sets $P$ and $Q$ are closed and $P \subset C \cup D$.
\end{assumption}
Following \eqref{eqn:ct_m}, we introduce the auxiliary hybrid system $\mathcal{H}_w \!=\! (C_w, F_w, D_w, G_w)$ given by 
\begin{equation} 
\begin{array}{ll}
	\!\!\!F_w(x) \!:=\! F(x) & x \in C_w \!:=\!C \backslash Q \\
	\!\!\!G_w(x) \!:=\!
	\bigg\{\!\!\!
    \begin{array}{cl}
		x & \mbox{if}~ x \!\in\! Q\\
		G(x) & \mbox{if}~ x \!\in\! D \backslash Q
    \end{array}
    \bigg.
    & x \in D_w \!:=\!  D \cup Q \mbox{.}
\end{array}
\label{eqn:H_m}
\end{equation}
The intuition behind the construction of the system $\mathcal{H}_w$ is as follows:
the system $\mathcal{H}_w$ is used to characterize the behavior of the system $\mathcal{H}$ outside the set $Q$. Indeed, the solutions to $\mathcal{H}$ are the solutions to $\mathcal{H}_w$ (and vice versa) up to when they reach (if they do) the set $Q$.
By Definition \ref{def:weak_until}, when $p \,\mathcal{U}_w q$ is satisfied for $\mathcal{H}$, it follows that, the solutions to $\mathcal{H}_w$ starting from the set $P \backslash Q$ stay in $P \cup Q$.
Hence, by guaranteeing CI of $P \cup Q$ with respect to $P \backslash Q$ for $\mathcal{H}_w$,
we establish that every solution to $\mathcal{H}$ starting from $P \cup Q$ satisfies $p \,\mathcal{U}_w q$ for $\mathcal{H}$.

Similar to the approach for \eqref{eqn:ct}, by Definition \ref{def:strong_until}, when $p \,\mathcal{U}_s q$ is satisfied for $\mathcal{H}$, it follows that, 
the solutions to $\mathcal{H}$ starting from $P \cup Q$ satisfy one of the following properties:
\begin{itemize}
	\item the solution $\phi$ starts and remains in $P$ until it reaches $Q$ in finite hybrid time.
	\item the solution $\phi$ starts from $Q$.
\end{itemize}
Therefore, the satisfaction of $p \,\mathcal{U}_s q$ for $\mathcal{H}$ requires, additionally to $p \,\mathcal{U}_w q$ being satisfied, that every maximal solution $\phi$ starting from $P \backslash Q$ reaches $Q$ in finite hybrid time.
When the set $P \!\cup\! Q$ is CI with respect to $P \backslash Q$ for $\mathcal{H}_w$,
this property is guaranteed by $Q$ being ECI with respect to $P \!\cup\! Q$ for the auxiliary hybrid system $\mathcal{H}_s \!=\! (C_s, F_s, D_s, G_s)$ given by 
\begin{equation} 
\begin{array}{ll}
	\!\!\!F_s(x) \!:=\! F(x) & x \!\in\! C_s \!:=\!(C \backslash Q) \!\cap\! P\\
	\!\!\!G_s(x) \!:=\!
	\bigg\{\!\!\!
    \begin{array}{cl}
		\!x & \mbox{if}~ x \!\in\! Q\\
		\!G(x) & \mbox{otherwise}
    \end{array}
    \bigg.
    & x \!\in\! D_s \!:=\!  (D \!\cap\! P) \!\cup\! Q \mbox{.}
\end{array}
\label{eqn:H_m+}
\end{equation}
The hybrid system $\mathcal{H}_s$ is similar to the system \eqref{eqn:ct_m+} in that $\mathcal{H}_s$ is just the restriction of $\mathcal{H}_w$ in \eqref{eqn:H_m} to $P \cup Q$.
It is easy to see that $C_s \!:=\! C_w \cap (P \cup Q)$ and $D_s \!:=\! D_w \cap (P \cup Q)$.
As a result, when $Q$ is ECI with respect to $P \cup Q$ for $\mathcal{H}_s$, each maximal solution to $\mathcal{H}_s$ starting from $P \backslash Q$ reaches $Q$ in finite hybrid time.
Since the solutions to $\mathcal{H}_s$ are the solutions to $\mathcal{H}_w$ (and vice versa when $p \,\mathcal{U}_w q$ is satisfied) up to when they reach $Q$, and the solutions to $\mathcal{H}_w$ are the solutions to $\mathcal{H}$ (and vice versa) up to when they reach $Q$,
each solution to $\mathcal{H}$ starting from $P \backslash Q$ reaches $Q$ in finite time and remains in $P$ until it reaches $Q$, which implies the satisfaction of $p \,\mathcal{U}_s q$ for $\mathcal{H}$.

Alternatively, the satisfaction of $p \,\mathcal{U}_w q$ for $\mathcal{H}$ can be guaranteed by CI of $P \cup Q$ with respect to $P \cup Q$ (namely, forward invariance of $P \cup Q$) for $\mathcal{H}$; and the satisfaction of $p \,\mathcal{U}_s q$ for $\mathcal{H}_w$ can be guaranteed by using FTA of $Q$ with respect to $P \cup Q$ instead of ECI of $Q$ with respect to $P \cup Q$ for $\mathcal{H}_s$.

\begin{example}[Timer]
Consider a hybrid system $\mathcal{H} = (C,F,D,G)$ modeling a constantly evolving timer with the state $x \in \mathbb{R}$ and
\begin{equation*}
\begin{array}{cl}
	F(x) := 1 & \qquad \forall x \in C := [0,1] \mbox{,}\\
	G(x) := 0 & \qquad \forall x \in D := [1, \infty)\mbox{.}
\end{array}
\end{equation*}
Define the atomic propositions $p$ and $q$ as 
\begin{equation*}
\begin{split}
p(x) &:=
\bigg\{
\begin{array}{cl}
    1 & \mbox{if}~ x \in [1/2,1] \\
    0 & \mbox{otherwise,}
\end{array}
\bigg.\quad
q(x) :=
\left\{
\begin{array}{cl}
    1 & \mbox{if}~ 
    x \in [1, \infty) \\
    0 & \mbox{otherwise,}
\end{array}
\right.
\end{split}
\end{equation*}
for each $x \!\in\! \mathbb{R}^n$. The sets $P$ and $Q$ in \eqref{eqn:K_sets} and the system $\mathcal{H}_w$ in \eqref{eqn:H_m} are given by $Q = D$, $P = [1/2,1]$, and 
\begin{equation*}
\begin{array}{cl}
	F_w(x) := 1 & \qquad \forall x \in C_w := [0,1) \mbox{,}\\
	G_w(x) := x & \qquad \forall x \in D_w := D = Q \mbox{.}
\end{array}
\end{equation*}
We notice that each solution to $\mathcal{H}_w$ starting from $P \backslash Q = [1/2,1)$ flows in $P$ and reaches $x \!=\! 1 \in Q$.
Once a solution reaches $x = 1$, it jumps according to the jump map $G_w(x) \!=\! x$ and stays at $\{1\} \in Q$ by jumping since it cannot flow back to $P\backslash Q$. Hence, the solutions to $\mathcal{H}_w$ starting from $P \backslash Q$ never leave the set $P \cup Q$, which implies that the set $P \cup Q$ is CI with respect to $P \backslash Q$ for $\mathcal{H}_w$.
Note that CI of $P \cup Q$ with respect to $P \backslash Q$ does not hold for $\mathcal{H}$ since once a solution to $\mathcal{H}$ reaches $Q$, it jumps outside $P \cup Q$.
Therefore, the formula $f \!=\! p \,\mathcal{U}_w q$ is satisfied for  $\mathcal{H}$ since the solutions to $\mathcal{H}$ starting from $P \backslash Q$ remain in $P$ until reaching the jump set $D \!=\! Q$. 
\hfill $\triangle$
\end{example}

\subsection{Sufficient Conditions for $p \,\mathcal{U}_w q$ using CI}

The following result characterizes the satisfaction of $p \,\mathcal{U}_w q$ using CI for hybrid systems.

\begin{theorem}[$p \,\mathcal{U}_w q$ via CI] \label{thm:weakuntil_ci}
Consider a hybrid system $\mathcal{H} = (C,F,D,G)$. Given atomic propositions $p$ and $q$, let the sets $P$ and $Q$ be given as in \eqref{eqn:K_sets}, and let the system $\mathcal{H}_w$ be as in \eqref{eqn:H_m}.
The formula $p \,\mathcal{U}_w q$ is satisfied for $\mathcal{H}$ if $P \cup Q$ is CI with respect to $P \backslash Q$ for $\mathcal{H}_w$.
\end{theorem}

\begin{proof}
\blue{Suppose that $P \cup Q$ is CI with respect to $P \backslash Q$ for $\mathcal{H}_w$. We show that, for each solution $\phi$ to $\mathcal{H}$ such that $\phi(0,0) \in P \backslash Q$, $\phi$ stays in $P \cup Q$ up to when $Q$ is reached.  
Indeed, let $\psi$ be a maximal solution to $\mathcal{H}_w$ such that 
$\psi(t,j) = \phi(t,j)$ for all $(t,j) \in \dom \phi$ up to when $Q$ is reached; such a solution $\psi$ to $\mathcal{H}_w$ always exists since the systems $\mathcal{H}$ and $\mathcal{H}_w$ share the same data outside the set $Q$. Furthermore, since $P \cup Q$ is CI with respect to $P \backslash Q$ for $\mathcal{H}_w$, we conclude that $\psi(t,j) \in P \cup Q$ for all $(t,j) \in \dom \psi$. Therefore, $\phi(t,j) \in P \cup Q$ for all $(t,j) \in \dom \phi$ up to when it reaches $Q$, which completes the proof.}
\end{proof}

Note that having $p \,\mathcal{U}_w q$ satisfied for $\mathcal{H}$ does not necessary imply that $P \cup Q$ is CI with respect to $P \backslash Q$ for $\mathcal{H}_w$ as shown next.

\begin{example} \label{ex:until_ci}
Consider a hybrid system $\mathcal{H} = (C,F,D,G)$ given by
\begin{equation*}
\begin{array}{lcl}
	\dot{x} \!\!\!\!\!\!&= F(x) := 1 & \qquad x \in C := [0,\infty) \\
	x^+ \!\!\!\!\!\!&= G(x) := 0
	&\qquad x \in D := [1, \infty)
\end{array}
\end{equation*}
and define atomic propositions $p$ and $q$ such that
\begin{equation*}
    p(x) :=
\bigg\{
\begin{array}{ll}
    1 & \mbox{if } x \in [1/2,\infty)\\
    0 & \mbox{otherwise,}
\end{array}
\bigg.\:\:\:
    q(x) :=
\bigg\{
\begin{array}{ll}
    1 & \mbox{if } x \in [-1,0]\\
    0 & \mbox{otherwise.}
\end{array}
\bigg.
\end{equation*}
The sets $P$ and $Q$ in \eqref{eqn:K_sets} and the data of $\mathcal{H}_w$ in \eqref{eqn:H_m} are given by $P = [1/2,\infty), Q = [-1,0]$, and
\begin{equation*}
    \begin{array}{ll}
	F_w(x) \!:=\! F(x)
	&\forall x \!\in\! C_w \!=\! (0,\infty)\\
	G_w(x) \!:=\! \Big\{\!\!
	\begin{array}{ll}
	    0 & \mbox{if } x \!\in\! [1, \infty)\\
	    x & \mbox{if } x \!\in\! [-1,0]
	\end{array}
	\Big.
	&\forall x \!\in\! D_w \!=\! [-1,0] \!\cup\! [1, \infty)\mbox{.}
\end{array}
\end{equation*}
Each solution to $\mathcal{H}$ starting from $P \backslash Q (=\!P)$ either flows in $P$ or reaches $\{0\} \!\in\! Q$ after a jump. Hence, the formula $p \,\mathcal{U}_w q$ is satisfied for $\mathcal{H}$. Now, each solution to $\mathcal{H}_w$ starting from $P \backslash Q (=\!P)$ is also a solution to $\mathcal{H}$ up to when it reaches $Q$, by reaching $\{0\}$ after a jump. Once a solution to $\mathcal{H}_w$ lands on $\{0\}$, both jumps according to $x^+ \!=\! G_w(x) \!=\! x$ and flows according to $\dot{x} \!=\! F_w(x) \!=\! 1$ are allowed by the concept of solution in Definition \ref{solution definition}. In particular, the solution to $\mathcal{H}_w$ flowing from $\{0\}$ is nontrivial and leaves the set $P \cup Q$. Hence, we conclude that the satisfaction of $p \,\mathcal{U}_w q$ for $\mathcal{H}$ does not imply that $P \cup Q$ is CI with respect to $P \backslash Q$ for $\mathcal{H}_w$.
\hfill $\triangle$
\end{example}

The following example illustrates Theorem~\ref{thm:weakuntil_ci}.

\begin{example}[Bouncing ball] \label{ex:bouncing_ball}
Consider a hybrid system $\mathcal{H} = (C,F,D,G)$ modeling a ball bouncing vertically on the ground, with state $x = (x_1, x_2) \in \mathbb{R}^2$ and data 
\begin{equation*}
\begin{array}{cll}
	F(x) \!\!\!\!&:= \Big[
	\begin{array}{c}
		x_2\\
		-\gamma\\
	\end{array}
	\Big]
	&\quad \forall x \!\in\! C \!:=\! \{x \!\in\! \mathbb{R}^2 : x_1 \geq 0\}\\
	G(x) \!\!\!\!&:= \Big[
	\begin{array}{c}
		0\\
		-\lambda x_2\\
	\end{array}
	\Big]
	&\quad \forall x \!\in\! D \!:=\! \{x \!\in\! \mathbb{R}^2 : x_1 = 0, x_2 \leq 0\}\mbox{,}
\end{array}
\end{equation*}
where $x_1$ denotes the height above the surface and $x_2$ is the vertical velocity. The parameter $\gamma > 0$ is the gravity coefficient and $\lambda \in (0,1)$ is the restitution coefficient.
Let $\varepsilon > 0$ and define atomic propositions $p$ and $q$ such that
\begin{equation*}
p(x) :=
\left\{
\begin{array}{ll}
    1 & \qquad\mbox{if }~ x_1 \in [0,\varepsilon] \:\:\mbox{and}\:\: x_2 \leq 0 \phantom{\geq 0}\\
    0 & \qquad\mbox{otherwise,}
\end{array}
\right.
\end{equation*}
\begin{equation*}
q(x) :=
\left\{
\begin{array}{ll}
    1 & \qquad\mbox{if }~ x_1 \geq 0 \:\:\mbox{and}\:\: x_2 \geq 0 \phantom{\in [0, \varepsilon]} \\
    0 & \qquad\mbox{otherwise.}
\end{array}
\right.
\end{equation*} 
The sets $P$ and $Q$ in \eqref{eqn:K_sets} and the data of $\mathcal{H}_w$ in \eqref{eqn:H_m} are given by $P = [0,\varepsilon] \times \mathbb{R}_{\leq 0}$, $Q = \mathbb{R}_{\geq 0} \times \mathbb{R}_{\geq 0}$, and 
\begin{equation*}
\begin{array}{cll}
	\!\!F_w(x) \!\!\!& =\! F(x)
	& \forall x \!\in\! C_w  \\
	\!\!G_w(x) \!\!\!\!\!& =\!
	\left\{\!\!
	\begin{array}{ll}  
		x & \mbox{if }~ x \in \mathbb{R}_{\geq 0} \!\times\! \mathbb{R}_{\geq 0}\\
		G(x) & \mbox{if }~x \in \{ 0 \} \!\times\! \mathbb{R}_{< 0}
    \end{array}
	\right.
	& \forall x \!\in\! D_w\mbox{,}
\end{array}
\end{equation*}
where $ D_w = (\left\{ 0 \right\} \!\times\! \mathbb{R}_{< 0}) \cup (\mathbb{R}_{\geq 0} \!\times\! \mathbb{R}_{\geq 0}) $ and $C_w \!=\! \mathbb{R}_{\geq 0} \!\times\! \mathbb{R}_{< 0}$.
Note that each solution to $\mathcal{H}_w$ from $P \backslash Q$ either flows in $P$ and reaches $Q$ after jumping from $\left\{ 0 \right\} \!\times\! \mathbb{R}_{< 0} \!\subset\! D_w$ or directly jumps from $\left\{ 0 \right\} \!\times\! \mathbb{R}_{< 0}$ towards $Q$.
Once a solution reaches a point $x \!\in\!Q$ after a jump, it jumps according to the jump map $G_w(x) \!=\! x$.
Note that each solution to $\mathcal{H}_w$ from $(0, \infty) \!\times\! \{0\} \!\subset\! Q$ can only flow.
Hence, every solution to $\mathcal{H}_w$ starting from $P \backslash Q$ never leaves the set $P \!\cup\! Q$, implying that the set $P \!\cup\! Q$ is CI with respect to $P \backslash Q$ for $\mathcal{H}_w$. Then, using Theorem \ref{thm:weakuntil_ci}, we conclude that the formula $p \,\mathcal{U}_w q$ is satisfied for $\mathcal{H}$.
\hfill $\triangle$
\end{example}

\subsection{Sufficient Conditions for $p \,\mathcal{U}_s q$ using $p \,\mathcal{U}_w q$ plus ECI}

The following result characterizes the satisfaction of $p \,\mathcal{U}_s q$ using ECI for hybrid systems in addition to the satisfaction of $p \,\mathcal{U}_w q$.

\begin{theorem}[$p \,\mathcal{U}_s q$ via $p \,\mathcal{U}_w q$ + ECI] \label{thm:stronguntil_eci}
Consider a hybrid system $\mathcal{H} = (C,F,D,G)$. Given atomic propositions $p$ and $q$, let the sets $P$ and $Q$ be given as in \eqref{eqn:K_sets} such that Assumption \ref{assump:SA1} holds, and let the data of $\mathcal{H}_s$ be given as in \eqref{eqn:H_m+}.
The formula $p \,\mathcal{U}_s q$ is satisfied for $\mathcal{H}$ if 
\begin{itemize}
    \item[1)] the formula $p \,\mathcal{U}_w q$ is satisfied for $\mathcal{H}$; and
    \item[2)] the set $Q$ is ECI with respect to $P \cup Q$ for $\mathcal{H}_s$.
\end{itemize}
\end{theorem}

\begin{proof}
By definition of $\mathcal{H}_s$, if the formula $p \,\mathcal{U}_w q$ is satisfied for $\mathcal{H}$ by item 1), each solution to $\mathcal{H}_s$ starting from $P \backslash Q$ remains in $P \cup Q$.
Furthermore, when additionally $Q$ is ECI with respect to $P \cup Q$ for $\mathcal{H}_s$, each maximal solution to $\mathcal{H}_s$ starting from $P \backslash Q$ remains in the set $P \cup Q$ and reaches the set $Q$ in finite hybrid time.
The proof is completed if we show that each maximal solution $\phi$ to $\mathcal{H}$ starting from $P \backslash Q$ stays in $P \cup Q$ for all $(t,j) \in \dom \phi$ such that $t+j \leq \mathcal{T}_Q(\phi)$, and $\mathcal{T}_Q(\phi) < \infty$. To this end, let $\phi$ be a maximal solution to $\mathcal{H}$ starting from $P \backslash Q$. By item 1), $\phi$ remains in $P \backslash Q$ up to when it reaches $Q$ (if that ever happens).
Next, since both $\mathcal{H}$ and $\mathcal{H}_s$ share the same data on $P \backslash Q$, there always exists a solution $\psi$ to $\mathcal{H}_s$ such that $\psi(t,j) = \phi(t,j)$ for all $(t,j) \in \dom \phi$ provided that $t+j \leq \mathcal{T}_Q(\phi) = \mathcal{T}_Q(\psi)$.
Furthermore, by item 2), we know that $\mathcal{T}_Q(\psi) = \mathcal{T}_Q(\phi) < \infty$. Then, since we already know that $\psi(t,j) \in P \cup Q$ for all $(t,j) \in \dom \psi$ by item 1), we conclude that $\phi(t,j) = \psi(t,j) \in P \cup Q$ for all $(t,j) \in \dom \phi$ provided that $t+j \leq \mathcal{T}_Q(\phi) = \mathcal{T}_Q(\psi)$; and thus, the proof is completed.
\end{proof}

\red{
\begin{remark}
Recall that ECI of $Q$ with respect to $P \cup Q$ allows the solutions that start in $P \cup Q$ to leave $P \cup Q$, and then eventually get back to $Q$.  This behavior is excluded  under our conditions guaranteeing $p ~ \U_s q$ since we assume that $p ~ \U_w q$ holds.
\end{remark}
}

The following example shows that having 
$p \,\mathcal{U}_s q$ satisfied for $\mathcal{H}$ does not necessary imply that $Q$ is ECI with respect to $P \cup Q$ for $\mathcal{H}_s$.

\begin{example} \label{counterexpECI}
Consider the hybrid system $\mathcal{H}$ in Example \ref{ex:until_ci}
with $p$ and $q$ therein replaced by $\tilde{p}$ and $\tilde{q}$, respectively,
\begin{equation*}
\begin{split}
    \tilde{p}(x) &:=
	\left\{\begin{array}{ll}
    	1 & \qquad\mbox{if }~ x \in [0, 1+\varepsilon]\\
    	0 & \qquad\mbox{otherwise,}
	\end{array}\right.\\
    \tilde{q}(x) &:=
	\left\{\begin{array}{ll}
    	1 & \qquad\mbox{if }~ x \in [-1, 0] \cup [1+\varepsilon, \infty)\\
    	0 & \qquad\mbox{otherwise,}
	\end{array}\right.
\end{split}
\end{equation*}
with $0 < \varepsilon < 1$. Let $P$ and $Q$ be as in \eqref{eqn:K_sets} with $\tilde{p}$ and $\tilde{q}$ instead of $p$ and $q$, respectively. 
The system $\mathcal{H}_s$ in \eqref{eqn:H_m+} is given by
\begin{equation*}
    \begin{array}{ll}
	F_s(x) \!:=\! F(x)
	&\!\!\forall x \!\in\! C_s \!=\! (0, 1+\varepsilon)\\
	G_s(x) \!:=\!
	\Big\{\!\!
	\begin{array}{ll}
		x & \mbox{if } x \notin [1, 1+\varepsilon)\\
	    0 & \mbox{otherwise}	    
	\end{array}
	\Big.
	&\!\!\forall x \!\in\! D_s \!=\! [-1, 0] \cup [1, \infty)\mbox{.}
\end{array}
\end{equation*}
Each solution to $\mathcal{H}$ starting from $P \backslash Q \!=\! (0,1+\varepsilon)$ either flows in $P$ and reaches $[1+\varepsilon, \infty) \subset Q$ or reaches $\{0\} \!\in \! Q$ after a jump from $ (1, 1+ \varepsilon) \subset D$. Hence, the formula $p \,\mathcal{U}_s q$ is satisfied for $\mathcal{H}$. Now, we consider a solution to $\mathcal{H}$ starting from $P \backslash Q$ that reaches $Q$ for the first time by jumping from $[1,1+\varepsilon) \subset D$ to $\{ 0 \}$. Such a solution is also a solution to $\mathcal{H}_s$ up to when it reaches $\{0\} \in Q$ for the first time, from where, both the jump according to $x^+ = G_s(x) = x$ and the flow according to $\dot{x} = F_s(x) = 1$ are allowed by the concept of solution in Definition \ref{solution definition}. In particular, the solution flowing from $\{0\}$ is nontrivial and leaves the set $Q$.
Hence, we conclude that the satisfaction of $p \,\mathcal{U}_s q$ for $\mathcal{H}$ does not necessary imply that $Q$ is ECI with respect to $P \cup Q$ for $\mathcal{H}_s$.
\hfill $\triangle$
\end{example}

The following example illustrates Theorem \ref{thm:stronguntil_eci}.

\begin{example}[Thermostat]
\label{exp:thermostat}
Consider a hybrid system $\mathcal{H} = (C,F,D,G)$ modeling a controlled thermostat system.
The variable $h$ denotes the state of the heater, i.e., $h \!=\! 1$ implies the heater is on and $h \!=\! 0$ implies the heater is off. The variable $z$ is the room temperature, $z_o$ denotes the temperature outside the room, and $z_\triangle$ denotes the capacity of the heater to raise the temperature such that
$z_o < z_{\min} < z_{\max} < z_o + z_\triangle$.
The system $\mathcal{H}$ with the state $x := (h,z) \in \{0,1\} \times \mathbb{R}$ is given by
\begin{equation*}
\begin{array}{lll}
F(x) := \big[ 0 \quad -z + z_o + z_\triangle h \big]^\top & ~\forall x \in C := C_0 \cup C_1 \\
G(x) := \big[ 1-h \quad  z \big]^\top & ~\forall x \in D := D_0 \cup D_1,
\end{array}
\end{equation*}
where $C_0 := \{x \in \mathbb{R}^2 : h = 0, z \geq z_{\min}\}$, $C_1 := \{x \in \mathbb{R}^2 : h = 1, z \leq z_{\max}\}$, $D_0 := \{x \in \mathbb{R}^2 : h = 0, z \leq z_{\min}\}$, and $D_1 := \{x \in \mathbb{R}^2 : h = 1, z \geq z_{\max}\}$.
Define the atomic propositions $p$ and $q$ as 
\begin{equation*}
\begin{split}
p(x) &:=
\left\{
\begin{array}{ll}
    1 & \qquad\mbox{if}~ x \in \{ 1 \} \times (-\infty, z_{\max}]\\
    0 & \qquad\mbox{otherwise,}
\end{array}
\right.\\
q(x) &:=
\left\{
\begin{array}{ll}
    1 & \qquad\mbox{if}~ x \in \{ 0 \} \times [z_{\min}, + \infty) \\
    0 & \qquad\mbox{otherwise,}
\end{array}
\right.
\end{split}
\end{equation*}
for each $x \!\in\! \mathbb{R}^2$.
Then, the sets $P$ and $Q$ in \eqref{eqn:K_sets} are given by $P \!=\! \{ 1 \} \!\times\! (-\infty, z_{\max}]$, $Q \!=\! \{ 0 \} \!\times\! [z_{\min}, + \infty)$. The system $\mathcal{H}_w$ is given as in \eqref{eqn:H_m},
where $C_w \!=\! \{1\} \!\times\! (-\infty, z_{\max}] (= \!P)$ and $D_w \!=\! (\{0\} \!\times\! \mathbb{R}) \cup (\{1\} \!\times\! [z_{\max}, +\infty))$.
Then, the system $\mathcal{H}_s$ is given as in \eqref{eqn:H_m+} with $C_s \!=\! C_w (=\! P)$ and $D_s \!=\! (\{0\} \!\times\! [z_{\min}, +\infty)) \cup \{(1, z_{\max})\}$.
Note that each solution to $\mathcal{H}_s$ from $P \backslash Q (=P)$ flows in $P$ and reaches $Q$ after jumping from $\{(1, z_{\max})\} \in D_s$.
Once a solution reaches a point $x \in Q$, it jumps according to the jump map $G_s(x) = x$ and stays in $Q$ by jumping, which implies $Q$ is ECI with respect to $P \cup Q$ for $\mathcal{H}_s$.
Furthermore, each solution to $\mathcal{H}_w$ starting from $P \backslash Q (= P)$ flows in $P$ and reaches $Q$ for the first time by jumping from
$\{(1, z_{\max})\}$ to $\{(0, z_{\max})\}$.
Once a solution to $\mathcal{H}_w$ lands on $\{(0,z_{\max})\}$, it jumps according to the jump map $G_w(x) \!=\! x$ and stays in $Q$ by jumping.
Hence, each solution to $\mathcal{H}_w$ starting from $P \backslash Q$ does not leave the set $P \cup Q$, which implies that the set $P \cup Q$ is CI with respect to $P \backslash Q$ for $\mathcal{H}_w$; and thus, using Theorem \ref{thm:weakuntil_ci}, we conclude that the formula $p \,\mathcal{U}_w q$ is satisfied for $\mathcal{H}$.
As a result, using Theorem \ref{thm:stronguntil_eci}, we conclude that the formula $p \,\mathcal{U}_s q$ is satisfied for $\mathcal{H}$.
\hfill $\triangle$
\end{example}

\subsection{Equivalence Between $p \,\mathcal{U}_s q$ and $p \,\mathcal{U}_w q$ plus FTA}

The following result characterizes the satisfaction of $p \,\mathcal{U}_s q$ using FTA in addition to the satisfaction of $p \,\mathcal{U}_w q$.

\begin{theorem}
[$p \,\mathcal{U}_s q$ via $p \,\mathcal{U}_w q$ + FTA]
Consider a hybrid system $\mathcal{H} = (C,F,D,G)$. Given atomic propositions $p$ and $q$, let the sets $P$ and $Q$ be as in \eqref{eqn:K_sets} such that the set $Q$ is closed and let the data of $\mathcal{H}_s$ be given in \eqref{eqn:H_m+}.
The formula $p \,\mathcal{U}_s q$ is satisfied for $\mathcal{H}$ if and only if
\begin{itemize}
    \item[1)] the formula $p \,\mathcal{U}_w q$ is satisfied for $\mathcal{H}$; and
	\item[2)] the set $Q$ is FTA with respect to $P \cup Q$ for $\mathcal{H}_s$.
\end{itemize}
\label{thm:stronguntil_fta}
\end{theorem}

\begin{proof}
Suppose that $p \,\mathcal{U}_s q$ is satisfied for $\mathcal{H}$. By definition of $p \,\mathcal{U}_s q$, we conclude that $p \,\mathcal{U}_w q$ is satisfied for $\mathcal{H}$. Next, we show that $Q$ is FTA with respect to $P \cup Q$ for $\mathcal{H}_s$.
To do so, we consider a maximal solution $\phi$ to $\mathcal{H}_s$ starting from $P \cup Q$.
In particular, each maximal solution to $\mathcal{H}_s$ starting from $Q$, the solution stays in $Q$ by construction of $\mathcal{H}_s$.
Hence, we consider a maximal solution $\phi$ to $\mathcal{H}_s$ starting from $P \setminus Q$.
Since $p \,\mathcal{U}_w q$ is satisfied for $\mathcal{H}$, the solution $\phi$ either reaches $Q$ in finite time or remains in $P \backslash Q$. To exclude the latter case, we show that when $\phi$ remains in $P \backslash Q$, then $\phi$ is a maximal solution to $\mathcal{H}$.
Indeed, assume the existence of a solution $\psi$ to $\mathcal{H}$ that is a nontrivial extension of $\phi$; namely,
there exists $I \!\subset\! \mathbb{R}_{\geq 0} \times \mathbb{N}$ such that $I \!\neq\! \emptyset$ and $\dom \psi = \dom \phi \cup I$.
Note that $\psi (\dom \phi) = \phi (\dom \phi) \subset P \backslash Q$.
Also, since $\psi$ must remain in $P \backslash Q$ up to when it reaches $Q$, we can choose $I$ such that $\psi (\dom \phi \cup I) \subset P \backslash Q$.
Hence, $\psi$ is a solution to $\mathcal{H}_s$, which contradicts the fact that $\phi$ is a maximal solution to $\mathcal{H}_s$. Furthermore, since $p \,\mathcal{U}_s q$ is satisfied for $\mathcal{H}$, we conclude that $\phi$, being a maximal solution to $\mathcal{H}$, must reach $Q$ in finite hybrid time.
 
Now, suppose that the formula $p \,\mathcal{U}_w q$ is satisfied for $\mathcal{H}$. This implies that each maximal solution $\phi$ to $\mathcal{H}$ remains in $P \backslash Q$ for all hybrid time; otherwise, $\phi$ remains in $P \backslash Q$ up to when it reaches $Q$ in finite hybrid time. To exclude the first scenario, we note that when $\phi$ remains in $P \backslash Q$ for all hybrid time, it follows that $\phi$ is also a maximal solution to $\mathcal{H}_s$. However, by item 2), the maximal solutions to $\mathcal{H}_s$ must reach $Q$.
\end{proof}

\section{Sufficient Conditions for CI, pre-ECI, and ECI} \label{Section.CIECI}

\subsection{Sufficient Conditions for CI}
First, we recall the sufficient conditions for invariance notions using a barrier function in \cite{199,maghenem2020sufficient} for hybrid systems. Below, the concept of the tangent cone\footnote{This definition of tangent cone is also known as the contingent cone and also as the Bouligand tangent cone.} to a set is used; see \cite[Definition 5.12]{goebel2012hybrid}. The tangent cone at a point $x \in \mathbb{R}^n$ of a set $C \subset \mathbb{R}^n$ is given by
$T_C(x) \!:=\! \left\{ v \in \mathbb{R}^n: \liminf_{h \rightarrow 0^+} \tfrac{|x + h v|_C}{h} = 0 \right\}$.
We also recall the equivalence \cite[Page 122]{aubin2009set}
\begin{equation} \label{eq.conti}
\begin{split}
v \in T_{C}(x) \Leftrightarrow ~\exists \left\{ h_i \right\}_{i \in \mathbb{N}} &\rightarrow 0^+ ~\mbox{and}~\left\{v_i\right\}_{i \in \mathbb{N}} \rightarrow v\\
	&: x + h_i v_i \in C~~\forall i \in \mathbb{N}\mbox{.}
\end{split}
\end{equation}
 
Furthermore, for the given sets $\mathcal{O}, \mathcal{X}_u \subset \mathbb{R}^n$ with $\mathcal{O} \cap \mathcal{X}_u = \emptyset$, we recall from \cite{199} the notion of a barrier function candidate with respect to $(\mathcal{O}, \mathcal{X}_u)$ for $\mathcal{H}$.

\begin{definition} [Barrier function candidate]
Consider $\mathcal{H} = (C,F,D,G)$ and sets $\mathcal{O}, \mathcal{X}_u \subset \mathbb{R}^n$ with $\mathcal{O} \cap \mathcal{X}_u = \emptyset$. A function $B : \mathbb{R}^n \!\rightarrow\! \mathbb{R}$ is said to be a barrier function candidate with respect to $(\mathcal{O}, \mathcal{X}_u)$ for $\mathcal{H}$ if
\begin{equation} \label{eqn:barrier_candidate_ci}
\left\{
\begin{array}{ll}
	B(x) \leq 0 & \qquad\forall x \in \mathcal{O}\\
	B(x) > 0 & \qquad\forall x \in (C \cup D) \cap \mathcal{X}_u \mbox{.}
\end{array}
\right.
\end{equation}
\end{definition}

To make the paper self contained, in the following, we recall a result on safety for hybrid systems \cite[Theorem 3.2]{199} to derive sufficient conditions for CI for hybrid systems.
Given two sets $\mathcal{O}$ and $\mathcal{X}_u$, the conditions given below provide sufficient conditions to verify that $\mathbb{R}^n \backslash \mathcal{X}_u$ is CI with respect to $\mathcal{O}$ for $\mathcal{H}$.
Furthermore, according to Remark \ref{remark:ci_fi}, when $\mathcal{O} = \mathbb{R}^n \backslash \mathcal{X}_u$, CI of $\mathbb{R}^n \backslash \mathcal{X}_u$ with respect to $\mathcal{O}$ reduces to forward pre-invariance of the set $K := \mathcal{O}$ \cite[Theorem 1 and Proposition 2]{maghenem2020sufficient}.

\begin{proposition} [CI and forward invariance]
Consider a hybrid system $\mathcal{H} = (C,F,D,G)$ \blue{such that Assumption  \ref{assump:SA} holds}.
\begin{itemize}[leftmargin=0.2in]
\item[1)] \label{prop:conditional_invariance}
\blue{Given sets $\mathcal{O}$ and $\mathcal{X}_u$ such that $\mathcal{O} \cap \mathcal{X}_u = \emptyset$ and  $\mathcal{O}$ and $\mathbb{R}^n \backslash \mathcal{X}_u$ are subsets of $C \cup D$, the set $\mathbb{R}^n \backslash \mathcal{X}_u$ is \emph{CI} with respect to $\mathcal{O}$ for $\mathcal{H}$} if there exists a $\mathcal{C}^1$ barrier function candidate $B$ with respect to $(\mathcal{O}, \mathcal{X}_u)$ for $\mathcal{H}$ as in \eqref{eqn:barrier_candidate_ci} such that $K := \{x \in C \cup D : B(x) \leq 0\}$ is closed and the following hold:
\end{itemize}
\begin{eqnarray}
\label{eqn:barrier_condition}
    &\left< \nabla B(x), \eta \right> \leq 0 ~\forall x \!\in\! (U(\partial K) \backslash K) \!\cap\! C, \forall \eta \!\in\! F(x) \!\cap\! T_C(x)\mbox{,}\nonumber\\
    &B(\eta) \leq 0 \:\:\:\qquad\forall x \!\in\! D \!\cap\! K, ~\forall\eta \!\in\! G(x)\mbox{,}\quad\qquad\qquad\qquad\\
	&G(D \cap K) \subset C \cup D\mbox{.}\:\:\:\qquad\qquad\qquad\qquad\qquad\qquad\qquad\nonumber
\end{eqnarray}
\begin{itemize}[leftmargin=0.2in]
	\item[2)] \label{prop:forward_invariance}
Given a closed set $K$ such that $K \subset C \cup D$, the set $K$ is \emph{forward pre-invariant} for $\mathcal{H}$ if there exists a $\mathcal{C}^1$ barrier function candidate $B$ with respect to $(K, \mathbb{R}^n \backslash K)$ for $\mathcal{H}$ as in \eqref{eqn:barrier_candidate_ci} such that \eqref{eqn:barrier_condition} holds.
Furthermore, the set $K$ is \emph{forward invariant} for $\mathcal{H}$ if the following additional conditions hold:
\begin{itemize}
    \item[a)] No maximal solution to $\mathcal{H}$ starting from $K$ has a finite time escape within $C \cap K$.
    \item[b)] Every maximal solution from  $(\partial C \cap K) \backslash D$ is nontrivial.
\end{itemize}
\end{itemize}
\end{proposition}

\begin{remark}
One can guarantee that the solutions to $\mathcal{H}$ do not have a finite escape time\footnote{A solution has finite escape time inside a given set if the solution diverges while remaining inside the set within a bounded (hybrid) time domain; see \cite[Chapter 3]{khalil}.} inside the set $K \cap C$ when, for example, the set $K \cap C$ is compact or when the flow map $F$ has global linear growth on $K \cap C$. Furthermore, according to \cite[Proposition 3]{maghenem2020sufficient}, the existence of a nontrivial solution starting from each point in $(K \cap \partial C) \backslash D$ can be proved by verifying the following infinitesimal condition:
$F(x) \cap T_{K \cap C}(x) \!\neq\! \emptyset$ for all $x \!\in\! U(x_o) \cap (K \cap \partial C)$ and for all $x_o \!\in\! (K \cap \partial C) \backslash D$.
\end{remark}

\subsection{Sufficient Conditions for pre-ECI}

In the following, inspired by \cite[Theorem 3.4]{ladde1972analysis}, we propose sufficient conditions for pre-ECI for hybrid systems.

\begin{theorem} [pre-ECI]
Consider a hybrid system 
$\mathcal{H} = (C,F,D,G)$ \blue{such that Assumption  \ref{assump:SA} holds}, 
and sets $\mathcal{O} \subset C \cup D$ and 
$\mathcal{A} \subset \mathbb{R}^n$. 
The set $\mathcal{A}$ is \emph{pre-ECI} 
with respect to the set $\mathcal{O}$ 
for $\mathcal{H}$ if the following properties hold:
\begin{itemize}
\item[1)] There exist a $\mathcal{C}^1$ function $v : \mathbb{R}^n \rightarrow \mathbb{R}$ and a locally Lipschitz function $f_c : \mathbb{R} \rightarrow \mathbb{R}$ such that
	\begin{itemize}
		\item[1a)]
		$
		\begin{aligned}[t]
			&\!\!\left< \nabla v(x), \eta \right> \leq f_c(v(x)) &\forall x \!\in\! C, ~\forall \eta \!\in\! F(x) \cap T_C(x)\mbox{,}\\
		&\!\!v(\eta) \leq v(x) &\forall x \!\in\! D, ~\forall \eta \!\in\! G(x)\mbox{;}
		\end{aligned}
		$
		\item[1b)] there exists a constant $r_1 > 0$ such that
		the solutions to $\dot{y} = f_c(y)$, starting from $v(\mathcal{O})$, converge to $(-\infty, r_1)$ in finite time.\footnote{The solutions to $\dot{y} = f_c(y)$ from $v(\mathcal{O})$ exist at least until they reach the set $(-\infty,r_1)$.}
	\end{itemize}
\item[2)] There exist a $\mathcal{C}^1$ function $w : \mathbb{R}^n \rightarrow \mathbb{R}$ and a nondecreasing function\footnote{A scalar-valued function $f_d$ is said to be nondecreasing if for each $x \leq y$, $f_d(x) \leq f_d(y)$.} $f_d : \mathbb{R} \rightarrow \mathbb{R}$ such that
	\begin{itemize}
		\item[2a)]
		$ \begin{aligned}[t]
			&\!\!\left< \nabla w(x), \eta \right> \leq 0 &\forall x \in C, ~\forall \eta \in F(x) \cap T_C(x) \mbox{,}\\
		&\!\!w(\eta) \leq f_d(w(x)) &\forall x \in D, ~\forall \eta \in G(x) \mbox{;}
		\end{aligned}
		$
		\item[2b)] there exists a constant $r_2 > 0$ such that 
		the solutions to $z^+ = f_d(z)$, starting from $w(\mathcal{O})$, converge to $(-\infty, r_2)$ in finite time.
	\end{itemize}
	\item[3)] One of the following conditions holds:
	\begin{itemize}
		\item[3a)] Each complete solution to $\mathcal{H}$ starting from $\mathcal{O}$ is eventually continuous and, with $r_1$ coming from item 1b),
		\begin{align} \label{eq:S1}
			S_1 := \{ x \in C : v(x) < r_1 \} \subset \mathcal{A}\mbox{.}
		\end{align}
		\item[3b)] Each complete solution to $\mathcal{H}$ starting from $\mathcal{O}$ is eventually discrete and, with $r_2$ coming from item 2b),
		\begin{align} \label{eq:S2}
			S_2 := \{ x \in D : w(x) < r_2 \} \subset \mathcal{A}\mbox{.}
		\end{align}
		\item[3c)] Each complete solution to $\mathcal{H}$ starting from $\mathcal{O}$ is eventually continuous, eventually discrete, or has a hybrid time domain that is unbounded in both the $t$ and the $j$ direction and,
		with $r_1$ and $r_2$ coming from item 1b) and item 2b) respectively, \eqref{eq:S1} and \eqref{eq:S2} hold.
		\item[3d)] With $r_1$ and $r_2$ coming from item 1b) and item 2b) respectively, \eqref{eq:S1} and \eqref{eq:S2} hold, and
		$G(S_2) \cap C \subset S_1$.
	\end{itemize}
\end{itemize}
\label{thm:pre_eventual_ci}
\end{theorem}
\begin{proof}
According to the definition of pre-ECI, we need to show that for each complete solution $\phi$ to $\mathcal{H}$ starting from $\mathcal{O}$, there exists $(t^\star,j^\star) \in \dom \phi$ such that $\phi(t,j) \in \mathcal{A}$ for all $(t,j) \in \dom \phi$ such that $t + j \geq t^\star + j^\star$. 
Consider a complete solution $\phi$ to $\mathcal{H}$ starting from $\phi(0,0) \in \mathcal{O}$. Let $y$ be the maximal solution to $\dot{y} = f_c(y)$ starting from 
 $y(0) = v(\phi(0,0)) \in v(\mathcal{O})$ and let $z$ be the complete solution to the system $z^+ = f_d(z)$ starting from $z(0) = w(\phi(0,0)) \in w(\mathcal{O})$.

First, if the solution $\phi$ initially flows, we use item 1a) to conclude,
via the comparison lemma in Lemma \ref{lemcomp}, Lemma \ref{lemaux1}, and Lemma \ref{lemaux2} that
$v(\phi(t,0)) \leq y(t)$ for all $t \in I^0$,
where
$I^0 := \{ t \in \mathbb{R}_{\geq 0} : (t,0) \in \dom \phi \}$.
To show this, we used the fact that\footnote{When $y$ has a finite-escape time $t_y > 0$, using item 1b), it follows that $\lim_{t \nearrow t_y} y(t) = -\infty$. Furthermore, since $v(\phi(t,0)) \leq y(t)$ for all $t \in \dom y \cap I^0$, then there must exist $t_\phi \in I^0$ with $t_\phi \leq t_y$ such that $\lim_{t \nearrow t_\phi} v(\phi(t,0)) = -\infty$. Hence, $\phi$ should escape to $-\infty$ no later than $y$.} $I^0 \subset \dom y$. Furthermore, if the solution $\phi$ jumps initially, we conclude using item 1a) that
$v(\phi(0,1)) \leq y(0)$.
By extending this reasoning over the domain of $\phi$, we conclude that
$v(\phi(t,j)) \leq y(t)$ for all $(t,j) \in \dom\phi$.

On the other hand, using item 2a), we conclude that if $\phi$ initially jumps, then 
$w(\phi(0,1)) \!\leq\! f_d(w(\phi(0,0))) \!=\! f_d(z(0)) \!=\! z(1)$.
Otherwise, when the solution $\phi$ initially flows, we conclude that 
$w(\phi(t,0)) \!\leq\! w(\phi(0,0)) \!=\! z(0)$ for all $t \!\in\! I^0$.
Moreover, by extending this reasoning over the domain of $\phi$
and using the fact that $f_d$ is nondecreasing,
we conclude that $w(\phi(t,j)) \!\leq\! z(j)$ for all $(t,j) \!\in\! \dom\phi$.
Indeed, it is easy to see that
$w(\phi(t,j)) \!\leq\! w(\phi(t',j))$ for all $(t,j) \!\in\! \dom\phi$ such that $t \!\geq\! t'$ and $(t',j) \in \dom\phi$; namely, the bound of the function $w$ does not increase over the interval of flow of $\phi$.

Since $f_d$ is nondecreasing, the existence of $j_z \!\in\! \mathbb{N}$ such that $z(j_z) \!\in\! (-\infty,r_2)$ (coming from item 2b)) implies that $z(j) \!\in\! (-\infty,r_2)$ for all $j \!\geq\! j_z$.
Similarly, from the existence of $t_y \!\in\! \mathbb{R}_{\geq 0}$ such that $y(t_y) \!\in\! (-\infty,r_1)$ (in item 1b)), it follows that $y(t) \!\in\! (-\infty,r_1)$ for all $t \!\geq\! t_y$.
Moreover, since the solution $\phi$ is complete, $(\{t_y\} \!\times\! \mathbb{N}) \cap \dom \phi \neq \emptyset$ or $(\mathbb{R}_{\geq 0} \!\times\! \{j_z\}) \cap \dom\phi \neq \emptyset$.
Therefore, we conclude that $v(\phi(t, j)) \!<\! r_1$ for all $(t, j) \!\in\! \dom\phi$ such that $t \!\geq\! t_y$ or $w(\phi(t, j)) \!<\! r_2$ for all $(t, j) \!\in\! \dom\phi$ such that $j \!\geq\! j_z$.

To complete the proof, we show that there exists $(t^\star, j^\star) \in \dom \phi$ such that $\phi(t,j) \in \mathcal{A}$ for all $(t,j) \in \dom \phi$ such that $t + j \geq t^\star + j^\star$.
To this end, we present the following cases depending on items 3a)-3d).  
\begin{itemize}[leftmargin=0.15in]
    \item[a)] When the solution $\phi$ is complete and eventually continuous, it follows that, for some $\tilde{t} \in \mathbb{R}_{\geq 0}$, $\phi(t,j) \in C$ for all $(t,j) \in \dom \phi$ such that $t \geq \tilde{t}$. Hence, we have $v(\phi(t,j)) < r_1$ and $\phi(t,j) \in C$ for all $(t,j)$ such that $t \geq t^\star := \max\{t_y,\tilde{t}\}$; and thus, $S_1$ is nonempty.
    Then, if $S_1$ is a subset of $\mathcal{A}$, with $j^\star$ such that $(t^\star, j^\star) \in \dom\phi$, we have $\phi(t,j) \in S_1 \subset \mathcal{A}$ for all $(t,j) \in \dom\phi$ such that $t +j \geq t^\star + j^\star$.
    \item[b)] When the solution $\phi$ is complete and eventually discrete, it follows that, for some $\tilde{j} \in \mathbb{N}$, $\phi(t,j) \in D$ for all $(t,j) \in \dom \phi$ such that $j \geq \tilde{j}$. Hence, we have $w(\phi(t,j)) < r_2$ and $\phi(t,j) \in D$ for all $(t,j)$ such that $j \geq j^\star := \max\{j_z,\tilde{j}\}$; and thus, $S_2$ is nonempty.
    Then, if $S_2$ is a subset of $\mathcal{A}$, with $t^\star$ such that $(t^\star, j^\star) \in \dom\phi$, we have $\phi(t,j) \in S_2 \subset \mathcal{A}$ for all $(t,j) \in \dom \phi$ such that $t + j \geq t^\star + j^\star$.
    \item[c)] When the hybrid time domain of the solution $\phi$ achieves both an unbounded amount of flows and an unbounded number of jumps, we conclude that $(\{ t_y \} \times \mathbb{N}) \cap \dom\phi \neq \emptyset$; hence, $v(\phi(t,j)) < r_1$ for all $t \geq t_y$ such that $(t,j) \in \dom\phi$. Also, $(\mathbb{R}_{\geq 0} \times \{ j_z \}) \cap \dom\phi \neq \emptyset$; hence, $w(\phi(t,j)) < r_2$ for all $j \geq j_z$ such that $(t,j) \in \dom\phi$.
    As a result, $S_1$ and $S_2$ are nonempty.
    Then, if $S_1$ and $S_2$ are subsets of $\mathcal{A}$, with $(t^\star, j^\star) \in \dom\phi$ such that $t^\star \geq t_y$ and $j^\star \geq j_z$, we conclude that $\phi(t,j) \in S_1 \cup S_2 \subset \mathcal{A}$ for all $(t,j) \in \dom\phi$ such that $t+j \geq t^\star + j^\star$.
\end{itemize}
Then, when one among items 3a)-3c) holds, according to the arguments in a)-c), we have that every complete solution $\phi$ starting from $\mathcal{O}$ satisfies that there exists $(t^\star, j^\star) \!\in\! \dom\phi$ such that $\phi(t,j) \in \mathcal{A}$ for all $(t,j) \in \dom\phi$ such that $t +j \geq t^\star + j^\star$.

Next, suppose that item 3d) holds and the solution $\phi$ is genuinely Zeno, which implies that $\phi$ does not satisfy items 3a)-3c).
        In this case, the solution $\phi$ jumps infinitely many times on a bounded interval of ordinary time and it always flows after a finite number of jumps.
        Hence, due to the fact $\phi$ is genuinely Zeno,
        there exists $(\tilde{t},\tilde{j}) \in \dom \phi$ such that $\tilde{j} \geq j_z$ satisfying $w(\phi(\tilde{t},\tilde{j})) < r_2$ and $\phi(\tilde{t},\tilde{j}) \in D$, which, in turn, implies $\phi(\tilde{t},\tilde{j}) \in S_2 \subset \mathcal{A}$.
        Note that $w(\phi(t,j)) < r_2$ for all $(t,j) \in \dom\phi$ such that $j \geq \tilde{j} \geq j_z$.
        Moreover, using the fact that $\phi$ is genuinely Zeno, $\phi$ jumps to a point in $C$; namely, there exists $(\tilde{t},\tilde{j}), (\tilde{t},\tilde{j}+1) \in \dom\phi$ such that $\phi(\tilde{t},\tilde{j}+1) \in C$.
    
        \blue{Now, according to item 3d)}, when $\phi(\tilde{t},\tilde{j}+1) \in C$, $\phi(\tilde{t},\tilde{j}+1) \in S_1$. Namely, the solution $\phi$ jumps to the set $S_1 \subset \mathcal{A}$ at $(\tilde{t}, \tilde{j})$, which implies that $v(\phi(\tilde{t},\tilde{j}+1)) < r_1$.
Now, we show that the solution $\phi$, which jumps to $S_1$ at $(\tilde{t}, \tilde{j})$, satisfies $v(\phi(t,\tilde{j}+1)) < r_1$ for all $t$ such that $(t,\tilde{j}+1) \in \dom\phi$. Proceeding by contradiction, assume the existence of $t' > \tilde{t}$ such that $(t', \tilde{j}+1) \in \dom\phi$ (i.e., $t'$ is in the interval of flow) and $v(\phi(t', \tilde{j}+1)) \geq r_1$. Let $y$ be a solution to $\dot{y} = f_c(y)$ starting from $v(\phi(\tilde{t},\tilde{j}+1)) < r_1$.
	Under item 1b), with a locally Lipschitz function $f_c$, the unique solution $y$ has to remain in $(-\infty, r_1)$ once it reaches $(-\infty, r_1)$; otherwise, it contradicts uniqueness of solutions.
	Moreover, using the comparison lemma in Lemma \ref{lemcomp}, Lemma \ref{lemaux1}, and Lemma \ref{lemaux2}, we have $v(\phi(t,\tilde{j}+1)) \leq y(t)$ for all $(t,\tilde{j}+1) \in \dom\phi$. Hence, we conclude that $v(\phi(t, \tilde{j}+1)) < r_1$ for all $t > \tilde{t}$ such that $(t, \tilde{j}+1) \in \dom\phi$, which contradicts the existence of $t' > \tilde{t}$ such that $(t', \tilde{j}+1) \in \dom\phi$ and $v(\phi(t', \tilde{j}+1)) \geq r_1$. Therefore, we conclude that $v(\phi(t,\tilde{j}+1)) < r_1$ for all $t$ such that $(t,\tilde{j}+1) \in \dom\phi$; namely, $v(\phi(t,\tilde{j}+1)) < r_1$ for all $t$ in the interval of flow. As a result, we conclude that $w(\phi(t,j)) < r_2$ for all $(t,j) \in \dom\phi$ such that $j \geq \tilde{j}$ and $v(\phi(t,j)) < r_1$ for all $(t,j) \in \dom\phi$ such that $t \geq \tilde{t}$; and thus, with $(t^\star, j^\star) \in \dom\phi$ such that $j^\star \!\geq\! \tilde{j}$ and $t^\star \!\geq\! \tilde{t}$, we conclude that $\phi(t,j) \in S_1 \cup S_2 \!\subset\! \mathcal{A}$ for all $(t,j) \!\in\! \dom\phi$ such that $t + j \!\geq\! t^\star + j^\star$.
\qedhere
\end{proof}

\begin{lemma}
Consider a complete solution $\phi$ to $\mathcal{H}$.
The solution $\phi$ is genuinely Zeno\footnote{A solution $\phi$ to $\mathcal{H}$ is said to be genuinely Zeno if it is Zeno but not discrete.} if and only if the solution $\phi$ does not belong to the following categories: eventually continuous, eventually discrete, and with unbounded hybrid time domain in both the $t$ and the $j$ direction.
\end{lemma}
\begin{proof}
($\Rightarrow$)
Suppose that a complete solution $\phi$ to $\mathcal{H}$ is genuinely Zeno.
Note that, by definition of genuinely Zeno, the solution $\phi$ involves an infinite number of jumps in a finite amount of time, and it always flows after a finite number of jumps for a finite, nonzero amount of flow time.
\begin{itemize}[leftmargin=0.1in]
	\item First, we show that the solution $\phi$ is not eventually continuous. Proceeding by contradiction, assume that $\phi$ is eventually continuous. It follows that there exists $(t',j') \in \dom\phi$ such that $\phi(t,j) \in C$ for all $t \geq t'$, which contradicts the fact that the existence of an infinite number jumps after flowing. Hence, we conclude that $\phi$ is not eventually continuous.

	\item Next, we show that the solution $\phi$ is not eventually discrete. Assume that $\phi$ is eventually discrete. It follows that there exists $(t',j') \in \dom\phi$ such that $t' = \sup_t \dom\phi$ and $\phi(t',j) \in D$ for all $j \geq j'$, which contradicts the fact that $\phi$ always flows after a finite number of jumps. Hence, we conclude that $\phi$ is not eventually discrete.
	
	\item Finally, by definition of genuinely Zeno, the solution $\phi$ does not have a hybrid time domain that is unbounded in both the $t$ and the $j$ direction.
\end{itemize}
($\Leftarrow$)
We consider following cases:
\begin{itemize}[leftmargin=0.1in]
	\item First, we suppose that a complete solution $\phi$ to $\mathcal{H}$ is eventually continuous. It follows that for some $t' \in \mathbb{R}_{\geq 0}$, $\phi(t,j) \in C$ for all $(t,j) \in \dom\phi$ such that $t \geq t'$; and thus, we have $j' = \sup_j \dom\phi$ such that $(t',j') \in \dom\phi$. Hence, the solution $\phi$ does not jump after $(t',j')$. However, since Zeno solutions have infinitely many jumps in a finite time interval, we conclude that the solution $\phi$ is not a genuinely Zeno solution.
	\item Next, we suppose that a complete solution $\phi$ to $\mathcal{H}$ is eventually discrete. It follows that, for some $j' \in \mathbb{N}$, $\phi(t,j) \in D$ for all $(t,j) \in \dom\phi$ such that $j \geq j'$; and thus, we have $t' = \sup_t \dom\phi$ such that $(t',j') \in \dom\phi$. Hence, the solution $\phi$ do not flow after $(t',j')$. However, since genuinely Zeno solutions always flow after a finite number of jumps, we conclude that the solution $\phi$ is not a genuinely Zeno solution.
	\item Finally, we suppose that a complete solution $\phi$ to $\mathcal{H}$ has a hybrid time domain that is unbounded in both the $t$ and the $j$ direction; i.e., $\sup_t \dom\phi = \infty$ and $\sup_j \dom\phi = \infty$. However, a genuinely Zeno solution exhibits an infinite number of jumps in finite time interval; i.e., $\sup_t \dom\phi < \infty$. Hence, the solution $\phi$ is not a genuinely Zeno solution. \qedhere
\end{itemize}
\end{proof}

The following example illustrates the importance of item 3d) in Theorem \ref{thm:pre_eventual_ci}.

\begin{example}
Consider a hybrid system $\mathcal{H} \!=\! (C,F,D,G)$ with the state $x \!\in\! \mathbb{R}$ and the data
\begin{equation*}
\begin{array}{ll}
	F(x) \!:=\! 1
	& \forall x \!\in\! C \!:=\! \mathbb{R}_{\geq 0}\\
	G(x) \!:=\!
	\left\{\!\!
	\begin{array}{ll}
		g(x) & \mbox{if } x \!\in\! \{u_n\}_{n=1}^{\infty}\\
		x & \mbox{otherwise}
	\end{array}
	\right.
	& \forall x \!\in\! D \!:=\! \{u_n\}_{n=1}^{\infty} \!\cup\! \{\tfrac{3}{2}\} \mbox{,}
\end{array}
\end{equation*}
where $g$ is a nondecreasing function such that $g(x) \geq x$ for all $x \in \mathbb{R}$ and $\{u_n\}_{n=1}^{\infty}$ is a strictly increasing sequence such that $u_1 \!=\! 0$, $\lim_{n \rightarrow \infty} u_n \!=\! \tfrac{3}{2}$, and $u_{n+1} = g(u_n)$.
Consider the sets $\mathcal{O}$ and $\mathcal{A}$ given by $\mathcal{O} \!=\! \{0\}$ and $\mathcal{A} \!=\! \{x \!\in\! C : x \!\geq\! 2\} \!\cup\! \{x \!\in\! D : x \!\geq\! 1\}$.
Consider the functions $v$ and $w$ defined as $v(x) = -x + 4$ and $w(x) = -x + 2$ for each $x \in \mathbb{R}$.
Also, define the function $f_c$ as $f_c(y) = -1$ for each $y \in \mathbb{R}$.
For all $x \in C$, $\left< \nabla v(x), F(x) \right> = -1 = f_c(v(x))$; and for all $x \in D$, $v(G(x)) \leq v(x)$ since $v(g(x)) \leq v(x)$. Thus, item 1a) in Theorem \ref{thm:pre_eventual_ci} holds. Furthermore, for $r_1 = 2$, the solutions to $\dot{y} = f_c(y)$ starting from $v(\mathcal{O}) = \{0\}$ converge to $(-\infty, r_1)$ in finite time (in fact, in zero time) and $S_1 \!=\! \{x \!\in\! C: x \!>\! 2\} \subset \mathcal{A}$; hence, item 1b) holds.
Now, to show item 2b) holds, define the function $f_d$ as $f_d(z) = g(z)$ if $z \in w(D)$ and $f_d(z) = z$ otherwise.
For all $x \in C$, $\left< \nabla w(x), F(x) \right> = -1 \leq 0$; and
for all $x \in D$, $w(G(x)) \leq f_d(w(x))$ since $w(G(x)) \leq w(x) = -x+2$ and $f_d(w(x)) = g(-x+2) \geq -x+2$ since $g$ is such that $g(x) \geq x$. Hence, item 2a) holds.
Moreover, for $r_2 = 1$, the solutions to $z^+ = f_d(z)$ starting from $w(\mathcal{O}) = \{0\}$ reach $(-\infty, r_2)$ in finite time (in fact, in zero time) and $S_2 \!=\! \{x \!\in\! D: x \!>\! 1\} \subset \mathcal{A}$; thus, item 2b) holds.
Now, consider a maximal solution $\phi$ to $\mathcal{H}$ starting from $\mathcal{O} \!=\! \{0\}$.
The solution $\phi$ jumps at each element in $\{u_n\}_{n=1}^\infty$ and $\phi$ flows after each jump.
Hence, the solution exhibits infinitely many jumps.
Moreover, the maximal time of flow satisfies $\dom_t \phi = \sum_{j=0}^{\infty} (t_{j+1} - t_j) = \sum_{j=0}^{\infty} (u_{j+1} - u_j) < \infty$ since $\lim_{n \rightarrow \infty} u_n \!=\! \tfrac{3}{2}$.
Hence, $\phi$ is genuinely Zeno.
Finally, with $S_1$ and $S_2$ as defined above,
we have that each $x \in G(S_2) \cap C$ satisfies $x \!>\! 1$. This implies that $G(S_2) \cap C \not\subset S_1$; and thus, none of items 3a)-3d) in Theorem \ref{thm:pre_eventual_ci} hold.
In fact, $\mathcal{A}$ is not pre-ECI with respect to $\mathcal{O}$ for $\mathcal{H}$ since there exists a solution that flows out of $S_1$, which implies the solution does not stay in $\mathcal{A}$.
\end{example}

The following example illustrates Theorem \ref{thm:pre_eventual_ci}.
\begin{example} \label{exp:didact}
Consider a hybrid system $\mathcal{H} = (C,F,D,G)$ with the state $x = (x_1, x_2) \in \mathbb{R}^2$ and the data 
\begin{equation*}
\begin{array}{ll}
	\!\!\!\!F(x) \!:=\!\! \Big[\!\!
	\begin{array}{c}
		-x_1 \!-\! x_2 \\
		x_1 \!-\! x_2 \\
	\end{array}
	\!\!\Big]
	& \forall x \!\in\! C \!:=\! \{x \!\in\! \mathbb{R}^2 \!:\! x_1 \!\geq\! 0, x_1 \!\geq\! x_2\} \\
	\!\!\!\!G(x) \!:=\!\! \bigg[\!\!
	\begin{array}{c}
		-x_2/\sqrt{2}\\
		-x_2/\sqrt{2}
	\end{array}
	\!\!\bigg]
	& \forall x \!\in\! D \!:=\! \{x \!\in\! \mathbb{R}^2 \!:\! x_1 \!=\! 0, x_2 \leq 0\}\mbox{.}
\end{array}
\end{equation*}
Consider the sets $\mathcal{O}$ and $\mathcal{A}$ given by $\mathcal{O} \!=\! [0,1] \!\times\! (-\infty, -1]$ and 
$\mathcal{A} \!=\! \mathbb{R}_{\geq 0} \!\times\! [-1/2, +\infty)$. Next, to conclude that the set $\mathcal{A}$ is pre-ECI with respect to the set $\mathcal{O}$ for $\mathcal{H}$, we show that the conditions in Theorem \ref{thm:pre_eventual_ci} are satisfied.
Consider the functions $v(x) \!=\! |x|^2$ and $f_c(y) \!:=\! -2y$.
For all $x \!\in\! C$, $\left< \nabla v(x), F(x) \right> \!=\! -2(x_1^2+x_2^2) \!=\! f_c(v(x))$; and for all $x \!\in\! D$, $v(G(x)) \!=\! x_2^2 \!=\! v(x)$. Thus, item 1a) holds.
Furthermore, we notice that $v(\mathcal{O}) \!=\! [1, +\infty)$ and that, for $r_1 \!=\! 1/2$, \eqref{eq:S1} holds. Finally, for the system $\dot{y} \!=\! f_c(y) \!=\! - 2 y$, it is easy to see that the solutions starting from $v(\mathcal{O}) = [1, +\infty)$ reach the set $(-\infty, 1/2)$ in finite time; hence, item 1b) is satisfied. On the other hand, consider the functions $w(x) \!=\! - x_2$ and $f_d(z) \!:=\! z/2$ if $z \!\in\! w(D)$ and $f_d(z) \!=\! z$ otherwise.
For all $x \!\in\! C$, $\left< \nabla w(x), x_1 \!-\! x_2 \right> \!=\! x_2 \!-\! x_1 \!\leq\! 0$; and for all $x \!\in\! D$, $w(G(x)) \!=\! x_2/\sqrt{2} \!\leq\! f_d(-x_2) \!=\! -x_2/2$ since $x_2 \!\leq\! 0$.
Hence, we conclude that item 2a) holds.
Moreover, item 2b) holds for $r_2 \!=\! 1/2$ and the solutions to $z^+ \!=\! f_d(z)$ starting from $w(\mathcal{O}) \!=\! [1, +\infty)$ reach $(-\infty, 1/2)$.
Finally, for all $x \!\in\! G(S_2) \cap C \!=\! \{x \!\in\! C: x_1 \!<\! \tfrac{1}{2\sqrt{2}}\}$, $v(x) \!<\! 1/2$. Hence, $G(S_2) \cap C \!\subset\! S_1$, which implies that item 3d) holds.
\hfill $\triangle$
\end{example}

\begin{remark}
As illustrated in Example \ref{exp:didact}, once we propose the candidate functions
$v$ and $w$, we find the functions $f_c, f_d$ and the constants $r_1, r_2$ such that the conditions in Theorem \ref{thm:pre_eventual_ci} hold. That is, for a particular expression of the data of the hybrid system and the sets $\mathcal{O}$ and $\mathcal{A}$, 
we can automate the process of generating the functions and parameters satisfying the conditions in Theorem  \ref{thm:pre_eventual_ci} as in 
\cite{prajna2004safety, oehlerking2007fully}.
\end{remark}

Note that, it is possible to conclude pre-ECI of $\mathcal{A}$ with respect to $\mathcal{O}$ using only condition 1) (or only condition 2, respectively) in Theorem \ref{thm:pre_eventual_ci} provided that we have the knowledge that the solutions from $\mathcal{O}$ can reach the set $\mathcal{A}$ only via flowing (or only jumping, respectively), as shown in the following result.
Indeed, in many applications of hybrid systems, the state variable is composed of both continuous and discrete variables; see the thermostat model in Example \ref{exp:thermostat}. Furthermore, when the sets $\mathcal{O}$ and $\mathcal{A}$ are defined only in terms of the continuous state variables (respectively, only in terms of the discrete state variables), it is possible to conclude that the solutions from $\mathcal{O}$ reach the set $\mathcal{A}$ only by flowing (respectively, only by jumping).

\begin{proposition} [pre-ECI via flows]
Consider a hybrid system $\mathcal{H} = (C,F,D,G)$ \blue{such that Assumption  \ref{assump:SA} holds}, and sets $\mathcal{O} \subset C \cup D$ and $\mathcal{A} \subset \mathbb{R}^n$. The set $\mathcal{A}$ is \emph{pre-ECI} with respect to the set $\mathcal{O}$ for $\mathcal{H}$ if the following properties hold:
\begin{itemize}
	\item[1)] There exist a $\mathcal{C}^1$ function $v : \mathbb{R}^n \rightarrow \mathbb{R}$, a locally Lipschitz function $f_c : \mathbb{R} \rightarrow \mathbb{R}$, and a constant $r_1 > 0$ such that condition 1) in Theorem~\ref{thm:pre_eventual_ci} holds and,
	the set $S_1 := \{x \in C: v(x) < r_1\} \subset \mathcal{A}$.
	\item[2)] For each complete solution $\phi \!\in\! \mathcal{S_H}(\mathcal{O})$, \blue{there exists a solution $y$ to $\dot{y} \!=\! f_c(y)$ starting from $v(\phi(0,0))$ that satisfies $y(t) \!\in\! (-\infty, r_1)$ for all $t \geq t^\star$, for some nonnegative $t^\star \leq \sup \{ t : (t,j) \in \dom\phi \}$}.
	\item[3)] \blue{$G(\mathcal{A} \cap D) \cap (D \backslash C) \subset \mathcal{A}$}, or each solution from $\mathcal{O}$ is eventually continuous.
\end{itemize}
\label{prop:pre_ECI_f}
\end{proposition}

\begin{proof}
According to the definition of pre-ECI, we need to show that, for each complete solution $\phi$ to $\mathcal{H}$ starting from $\mathcal{O}$, there exists $(t^\star,j^\star) \in \dom \phi$ such that $\phi(t,j) \in \mathcal{A}$ for all $(t,j) \in \dom \phi$ such that $t + j \geq t^\star + j^\star$. 
Consider a complete solution $\phi$ to $\mathcal{H}$ starting from $\phi(0,0) \in \mathcal{O}$. Let $y$ be a maximal solution to $\dot{y} = f_c(y)$ starting from 
 $y(0) = v(\phi(0,0)) \in v(\mathcal{O})$.
First, if the solution $\phi$ flows initially, using the comparison lemma in Lemma \ref{lemcomp}, Lemma \ref{lemaux1}, and Lemma \ref{lemaux2} under item 1a) in \blue{Theorem \ref{thm:pre_eventual_ci}}, we conclude that $v(\phi(t,0)) \leq y(t)$ for all $t \in I^0$,
where $I^0 := \{ t \in \mathbb{R}_{\geq 0} : (t,0) \in \dom \phi \} \mbox{.}$
Furthermore, if the solution $\phi$ jumps initially, we conclude using item 1a) in \blue{Theorem \ref{thm:pre_eventual_ci}} that $v(\phi(0,1)) \leq y(0)$. By extending the latter reasoning along the domain of $\phi$, we conclude that 
$v(\phi(t,j)) \leq y(t)$ for all $(t,j) \in \dom \phi$. 
On the other hand, by item 2), there exists $t^\star \in \mathbb{R}_{\geq 0}$ such that
$t^\star \leq \sup \{t: (t,j) \in \dom\phi \}$ and
$y(t) \in (-\infty, r_1)$ for all $t \geq t^\star$. This fact implies that 
$(\{ t^\star \} \times \mathbb{N}) \cap \dom\phi \neq \emptyset$; hence,
$v(\phi(t,j)) < r_1$ for all $t \geq t^\star$ such that $(t,j) \in \dom\phi$
and $S_1$ is nonempty.
As a consequence,
using the fact that $S_1$ is a subset of $\mathcal{A}$,
we conclude that, for all $(t,j) \in \dom\phi$ such that $t \geq t^\star$, we have $\phi(t,j) \in \mathcal{A}$ provided that $\phi(t,j) \in C$ and that $\phi$ reaches $\mathcal{A}$ as it flows after $t^\star$.
Next, under item 3), if $\phi$ is eventually continuous, then $\phi$ remains in $\mathcal{A}$ by flowing; however,
once $\phi$ reaches a point $x \in D \cap \mathcal{A}$, it jumps. However, according to item 3), it remains in the set $\mathcal{A}$ after the jump, which completes the proof.
\end{proof}

\begin{proposition} [pre-ECI via jumps]
Consider a hybrid system $\mathcal{H} \!=\! (C,F,D,G)$ \blue{such that Assumption  \ref{assump:SA} holds} and sets $\mathcal{O} \subset C \cup D$ and $\mathcal{A} \!\subset\! \mathbb{R}^n$. The set $\mathcal{A}$ is \emph{pre-ECI} with respect to the set $\mathcal{O}$ for $\mathcal{H}$ if the following properties hold:
\begin{itemize}
	\item[1)] There exist a $\mathcal{C}^1$ function 
	$w : \mathbb{R}^n \rightarrow \mathbb{R}$, a nondecreasing function $f_d : \mathbb{R} \rightarrow \mathbb{R}$, and a constant $r_2 > 0$ such that condition 2) in Theorem~\ref{thm:pre_eventual_ci} holds and,
	the set $\widetilde{S}_2 := \{x \in C \cup D : w(x) < r_2 \} \subset \mathcal{A}$.
	\item[2)] For each complete solution $\phi \in \mathcal{S_H}(\mathcal{O})$, \blue{there exists a solution $z$ to $z^+ = f_d (z)$ starting from $w(\phi(0,0))$ that satisfies $z(j) \!\in\! (-\infty, r_2)$ for all $j \geq j^\star$, for some nonnegative $j^\star \leq \sup \{ j : (t,j) \in \dom\phi \}$}.
\end{itemize}
\label{prop:pre_ECI_j}
\end{proposition}

\begin{proof}
According to the definition of pre-ECI, we need to show that, for each complete solution $\phi$ to $\mathcal{H}$ starting from $\mathcal{O}$, there exists $(t^\star,j^\star) \in \dom \phi$ such that $\phi(t,j) \in \mathcal{A}$ for all $(t,j) \in \dom \phi$ such that $t + j \geq t^\star + j^\star$. 
Consider such a complete solution 
$\phi$ starting from $\phi(0,0) \in \mathcal{O}$. Let $z$ be a maximal solution to $z^+ = f_d(z)$ starting from 
 $z(0) = w(\phi(0,0)) \in w(\mathcal{O})$. First, if the solution $\phi$ jumps initially, then we conclude that
$ w(\phi(0,1)) \leq z(1)\mbox{.}$
Furthermore, if the solution $\phi$ initially flows, we conclude using item 2a) in Theorem \ref{thm:pre_eventual_ci} that
$ w(\phi(t,0)) \leq z(0)$ for all $t \in I^0 \mbox{,} $
where
$ I^0 := \{ t \in \mathbb{R}_{\geq 0} : (t,0) \in \dom \phi \} \mbox{.} $ By extending the latter reasoning along the domain of $\phi$ and using the fact that $f_d$ is nondecreasing, we conclude that $w(\phi(t,j)) \leq z(j)$ for all $(t,j) \in \dom \phi$. On the other hand, by item 2), there exists $j^\star \in \mathbb{N}$ such that $j^\star \leq \sup\{j: (t,j) \in \dom\phi\}$ and $z(j) \in (-\infty, r_2) $ for all $j \geq j^\star\mbox{.}$
This fact implies the existence of $(t^\star,j^\star) \in (\mathbb{R}_{\geq 0} \times \{ j^\star \} ) \cap \dom\phi \neq \emptyset$ such that $w(\phi(t^\star,j^\star)) < r_2$. Hence, $\phi(t^\star,j^\star) \in \widetilde{S}_2$ and $w(\phi(t,j)) < r_2$ for all $(t,j) \in \dom \phi$ with $j \geq j^\star$
and $\widetilde{S}_2$ is nonempty.
As a consequence,
using the fact that $\widetilde{S}_2$ is a subset of $\mathcal{A}$,
we conclude that, for all $(t,j) \in \dom\phi$ such that $j \geq j^\star$, it follows that $\phi(t,j) \in \mathcal{A}$, which completes the proof.
\end{proof}

\begin{remark}
Note that condition 2) in Proposition~\ref{prop:pre_ECI_f} holds for free for complete solutions starting from $\mathcal{O}$; for which, the domain is unbounded in the $t$ direction. Similarly, condition 2) in Proposition~\ref{prop:pre_ECI_j} holds for free for complete solutions starting from $\mathcal{O}$ for which the domain is unbounded in the $j$ direction. Moreover, maximal solutions are complete when the conditions in \cite[Proposition 2.10 or Proposition 6.10]{goebel2012hybrid} hold.
\end{remark}

\begin{remark}
In Theorem \ref{thm:pre_eventual_ci}, one could consider unifying conditions 1) and 2) as follows:
\begin{equation}
\begin{array}{ll}
    \!\!\left< \nabla v(x), \eta \right> \!\leq\! f_c(v(x)) & \forall \eta \!\in\! F(x) \cap T_C(x), \forall x \in C \mbox{,}\\
    \!\!v(\eta) \!\leq\! f_d(v(x)) & \forall \eta \!\in\! G(x), \forall x \in D \mbox{,}
\end{array}
\label{eqn:v_hyrid_evolution}
\end{equation}
where the functions $f_c$ and $f_d$ are defined in Theorem \ref{thm:pre_eventual_ci}. Furthermore, one may attempt to conclude pre-ECI of $\mathcal{A}$ with respect to $\mathcal{O}$ by showing that the set $(-\infty, r_1]$ is pre-ECI with respect to $v(\mathcal{O})$ for the reduced system given by 
\begin{equation}
    \dot{y} = f_c(y) \quad y \in v(C)\mbox{,} \quad y^+ = f_d(y) \quad y \in v(D)\mbox{.}
\label{eqn:abstracted_H}
\end{equation}
Such a comparison-based reasoning is very useful to analyze purely continuous-time or purely discrete-time systems. In general, a key step for such a reasoning to hold consists in showing that \eqref{eqn:v_hyrid_evolution} and \eqref{eqn:abstracted_H} imply that
$v(\phi(t,j)) \leq y(t,j)$ for all $(t,j) \in \dom \phi$.
However, this bound does not necessarily hold under \eqref{eqn:v_hyrid_evolution} and \eqref{eqn:abstracted_H} due to the possible mismatch in the jump times between the solutions $\phi$ to 
$\mathcal{H}$ and $y$ to \eqref{eqn:abstracted_H}.
However, the said bound holds if we replace the inequalities in \eqref{eqn:v_hyrid_evolution} by equalities, which is restrictive. As a consequence, the comparison arguments, in general, do not extend directly to the context of hybrid systems.  
\label{remark:pre_eventual_ci}
\end{remark}

\subsection{Sufficient Conditions for 
pre-ECI using Approximate Flow Lengths}

In this section, along the lines of Remark \ref{remark:pre_eventual_ci},
we relax some of the requirements in Theorem 4.4 when the jump times of the solutions are known approximately.
For this purpose, inspired by the work in \cite{182} for hybrid observers, we propose a new set of sufficient conditions for pre-ECI.
Given sets $\mathcal{O}, \mathcal{A} \!\subset\! \mathbb{R}^n$, we assume the existence of a set $K \!\subset\! \mathcal{A}$ that is forward pre-invariant for $\mathcal{H}$. Furthermore, we assume that we know approximately the length of the flow interval, between each successive jumps, for all the solutions starting from $\mathcal{O}$ until they reach the set $K$. Below, we use $\dom_{t} \,\phi$ (respectively, $\dom_{j} \,\phi$) to denote the projection of $\dom \phi$ on $\mathbb{R}_{\geq 0}$ (respectively, on $\mathbb{N}$), and we denote $T(\phi) \!=\! \sup\dom_{t} \,\phi$ and $J(\phi) \!=\! \sup\dom_{j} \,\phi$. By $t_j(\phi)$, when $j \in \mathbb{N}_{>0}$, we denote the time stamp associated with the jump $j$ uniquely characterized by
$(t_j(\phi), j-1) \in \dom\phi$ and  $(t_j(\phi), j) \in \dom\phi$. \blue{Moreover, when $j = 0$, $t_0 = 0$ is the initial flow time.}

\begin{definition}[Approximate flow lengths] \label{def:flow_length}
A closed set $\mathcal{I}_{\mathcal{O}, K} \subset \mathbb{R}_{\geq 0}$ is said to be the set of approximate flow lengths for the solutions to $\mathcal{H}$ starting from the set $\mathcal{O}$ and remaining in $\mathbb{R}^n \backslash K$ if, for each such a solution, denoted $\phi$, we have 
\begin{subequations}
\label{eqn:flow_lengths}
\begin{align}
	&\label{eqn:length_1}0 \leq t - t_j(\phi) \leq \sup\mathcal{I}_{\mathcal{O},K}  \quad\forall (t,j) \in \dom\phi\mbox{,}\\
	&\label{eqn:length_2}
	t_{j+1}(\phi) - t_j(\phi) \in \mathcal{I}_{\mathcal{O},K}  \quad\forall j \in \mathbb{N}_{> 0} ~\mbox{if } ~ J(\phi) = +\infty\mbox{,}\\
	&\qquad\qquad\:\: \mbox{or} ~ \blue{\forall j \in \{0,1, \ldots, J(\phi) - 1\} ~\mbox{if} ~ J(\phi) < +\infty}\mbox{.}\nonumber
\end{align}
\end{subequations}
\end{definition}
The set $\mathcal{I}_{\mathcal{O},K} \subset \mathbb{R}_{\geq 0}$ contains the possible lengths of the flow intervals between successive jumps for the solutions starting from $\mathcal{O}$ and remaining in $\mathbb{R}^n \backslash K$.
The role of \eqref{eqn:length_1} is to bound the length of the intervals of flow which are not covered by \eqref{eqn:length_2}, namely possibly the first $[0, t_1(\phi)]$ and the last $\dom_{t}\,\phi \cap [t_{J(\phi)}(\phi), +\infty]$ (when they are defined).
The existence of a set $\mathcal{I}_{\mathcal{O},K} \subset \mathbb{R}_{\geq 0}$ is not a problem, even when $K = \emptyset$, since it can always be given by $\mathcal{I}_{\mathcal{O},K} = \mathbb{R}_{\geq 0}$. However, it has advantage when $\mathcal{I}_{\mathcal{O},K} \subset \mathbb{R}_{\geq 0}$ is selected as tight as possible, namely to have as much information about the duration of flow between successive jumps as possible so that we reduce the number of possible solutions.

\begin{theorem}
[pre-ECI under approximate flow lengths]
\label{thm:pre_ECI_interval}
Given a hybrid system $\mathcal{H} = (C,F,D,G)$ \blue{such that Assumption  \ref{assump:SA} holds}, and sets $\mathcal{O} \subset C \cup D$, $\mathcal{A} \subset \mathbb{R}^n$, and $K \subset \mathcal{A}$ such that $K$ is forward pre-invariant for $\mathcal{H}$,
let the set $\mathcal{I}_{\mathcal{O},K}$ be as in Definition \eqref{def:flow_length} and let $\tau_M := \sup \mathcal{I}_{\mathcal{O},K}$. Then, the set $\mathcal{A}$ is \emph{pre-ECI} with respect to the set $\mathcal{O}$ for $\mathcal{H}$ if
\begin{itemize}[leftmargin=.2in]
	\item[1)] There exist a $\mathcal{C}^1$ function $v : \mathbb{R}^n \rightarrow \mathbb{R}$, a locally Lipschitz function $f_c : \mathbb{R} \rightarrow \mathbb{R}$, and a nondecreasing function $f_d : \mathbb{R} \rightarrow \mathbb{R}$ such that
	\begin{subequations}
	\begin{align}
		\!\!\!\!\!\!\!\!\!\!\!\left< \nabla v(x), \eta \right> \leq f_c(v(x)) & ~~\forall x \!\in\! C \backslash K, \forall \eta \!\in\! F(x) \!\cap\! T_C(x)\mbox{,}\label{eqn:known_interval_v1}\\
		\!\!\!\!\!\!\!\!\!\!\!v(\eta) \leq f_d(v(x)) & ~~\forall x \in D \backslash K,\forall \eta \in G(x)\mbox{;}
 \label{eqn:known_interval_v2}
	\end{align}
	\end{subequations}
	\item[2)] There exists a constant $r > 0$ such that 
	\begin{equation} \label{eqS}
	S := \{ x \in C \cup D : v(x) < r \} \subset \mathcal{A}\mbox{;}
	\end{equation}
	
	\item[3)] Every \blue{maximal} solution to the reduced hybrid system $\mathcal{H}_{r}$ starting from $v(\mathcal{O}) \times \{ 0 \}$ converges to $(-\infty, r) \times \mathbb{R}_{\geq 0}$ in finite hybrid time, where  
\begin{equation}
\!\!\!\!\!\!\!
\mathcal{H}_r :
\left\{ \!\!\!
\begin{array}{cll}
	\left[
	\begin{array}{c}	
 		\dot{y}\\
 		\dot{\tau}
 	\end{array}
 	\right]
 	\!\!\!\!\! & =
 	\left[
 	\begin{array}{c}
 		f_c(y)\\
 		1
	\end{array}
	\right]
	&  \blue{(y, \tau) \in  \mathbb{R} \times  [0, \tau_M]} \mbox{,}\\
	
	\left[
	\begin{array}{c}
		y^+  \\
		\tau^+ 
	\end{array}
	\right]
	\!\!\! & =
	\left[\begin{array}{c}
		f_d(y)
		\\
		0
	\end{array} \right]
	& (y, \tau) \in \mathbb{R} \times
\mathcal{I}_{\mathcal{O},K},
\end{array}	
\right.
\label{eqn:reduced_H}
\end{equation}
with $f_c$ and $f_d$ coming from item 1).
\end{itemize}
\end{theorem}
\begin{proof}
According to the definition of pre-ECI, we need to show that for each complete solution $\phi$ to $\mathcal{H}$ starting from $\mathcal{O}$,
there exists a hybrid time $(t^\star,j^\star) \!\in\! \dom\phi$ such that
$\phi(t,j) \in \mathcal{A}$ for all $(t,j) \!\in\! \dom\phi$ such that $t + j \!\geq\! t^\star + j^\star$. Without loss of generality, consider a complete solution $\phi$ to $\mathcal{H}$ starting from $\phi(0,0) \!\in\! \mathcal{O}$ and remaining in the complement of $K$. Let $(y, \tau)$ be a maximal solution pair to the system in \eqref{eqn:reduced_H} such that $y(0,0) \!=\! v(\phi(0,0)) \!\in\! \mathbb{R}$ and $\tau(0,0) \!=\! 0$. By definition of the set $\mathcal{I}_{\mathcal{O},K}$, we conclude the existence of a solution $(y, \tau)$ to $\mathcal{H}_r$ such that
$\dom (y, \tau) \!=\! \dom y \!=\!  \dom\phi$.

Now, we pick any $j \!\in\! \dom_{j}\,\phi$ and we let $I^j \!:=\! \{t \!\in\! \mathbb{R}_{\geq 0}: (t,j) \!\in\! \dom\phi\}$. Using Lemmas \ref{lemcomp}, \ref{lemaux1}, and \ref{lemaux2} under \eqref{eqn:known_interval_v1}, we conclude that
$v(\phi(t,j)) \!\leq\! y(t,j)$  for all $t \!\in\! I^j$.

Furthermore, for any $j \!\in\! \dom_{j}\,\phi$ such that $(t,j-1), (t, j) \!\in\! \dom\phi$, using \eqref{eqn:known_interval_v2},
we conclude that
$v(\phi(t,j)) \leq y(t,j)$  for all $(t,j) \!\in\! \dom\phi\mbox{.}$
Since the solution $y$ starting from $v(\phi(0,0))$ converges to $(-\infty, r)$ in finite time, we conclude that,
there exists $(t^\star,j^\star) \!\in\! \dom\phi$, such that
$v(\phi(t,j)) \!\in\! (-\infty, r)$ for all $(t,j) \!\in\! \dom\phi: t+j \!\geq\! t^\star+j^\star$. 
This implies, by item 2) and completeness of $\phi$, the existence of $(t^\star,j^\star) \!\in\! \dom\phi$ such that
$ \phi(t,j) \!\in\! \mathcal{A}$ for all $(t,j) \!\in\! \dom\phi: t+j \!\geq\! t^\star+j^\star$,
which completes the proof.
\end{proof}

\begin{example}
[Bouncing ball]
Consider the hybrid system $\mathcal{H} \!=\! (C,F,D,G)$ in Example \ref{ex:bouncing_ball} with $\gamma \!=\! 1$ and $\lambda \!=\! 0.5$. 
Let the sets $\mathcal{O} \!:=\! \{0\} \!\times\! [2,3]$ and 
$\mathcal{A} \!:=\! [0,1] \!\times\! [-1,1]$. 
Furthermore, consider $K \!:=\! \{ x \!\in\! C \cup D : 2x_1 + x_2^2 \leq 1/2 \} \!\subset\! \mathcal{A}$.
Using Proposition \ref{prop:conditional_invariance} with the barrier function candidate 
$B(x) \!:=\! 2x_1 + x_2^2 - 1/2$,
we conclude that $K$ is forward pre-invariant for $\mathcal{H}$. 
Next, for $v(x) \!:=\! 2 x_1 + x_2^2$, we conclude that for $f_c(y) \!=\! 0$, \eqref{eqn:known_interval_v1} holds. Furthermore, for $f_d(y) \!=\! y/4$, \eqref{eqn:known_interval_v2} holds. 
Now, for $r \!=\! 1/2$, it is easy to see that \eqref{eqS} holds. Finally, to conclude pre-ECI of $\mathcal{A}$ with respect to $\mathcal{O}$ using Theorem \ref{thm:pre_ECI_interval}, it is enough to show that $\mathcal{I}_{\mathcal{O},K}$ is bounded. Indeed, using Proposition \ref{prop:conditional_invariance} with the barrier function candidate given by
$B_1(x) \!:=\! 2x_1 \!+\! x_2^2 \!-\! 9$,
we conclude that the zero sublevel set of $B_1$, which contains $\mathcal{O}$, is forward pre-invariant. Hence, $\mathcal{R}(\mathcal{O})$ is bounded. Furthermore, from every initial condition in $\cl(\mathcal{R}(\mathcal{O}))$, the unique maximal solution reaches the set $D$ in finite time. Hence, the interval of flow of the solutions starting from $\mathcal{O}$ is uniformly bounded.   
\hfill $\triangle$
\label{ex:bb_length}
\end{example}

\subsection{Sufficient Conditions for ECI}

The proofs of the following results follow from the proofs of Proposition \ref{prop:pre_ECI_f} and Proposition \ref{prop:pre_ECI_j}, respectively.

\begin{theorem}[ECI via flows]\label{thm:ECIf}
Given a hybrid system $\mathcal{H} = (C,F,D,G)$ \blue{such that Assumption  \ref{assump:SA} holds}, and sets $\mathcal{O} \subset C \cup D$ and $\mathcal{A} \subset \mathbb{R}^n$,
suppose that there exist a $\mathcal{C}^1$ function $v : \mathbb{R}^n \rightarrow \mathbb{R}$, a locally Lipschitz function $f_c : \mathbb{R} \rightarrow \mathbb{R}$, and a constant $r_1 > 0$ such that items 1) - 3) in Proposition~\ref{prop:pre_ECI_f} hold.
Then, the set $\mathcal{A}$ is \emph{ECI} with respect to $\mathcal{O}$ for $\mathcal{H}$ if the following additional condition holds:
\begin{itemize}[leftmargin=.2in]
	\item For each solution $\phi \in \mathcal{S_H}(\mathcal{O})$, there exists a solution $y$ to $\dot{y} = f_c(y)$ starting from $v(\phi(0,0))$ satisfying $y(t) \in (-\infty, r_1)$ for all $t \geq t^\star$ and for some nonnegative $t^\star \leq \sup \{ t : (t,j) \in \dom\phi \}$.
\end{itemize}
\end{theorem}

\begin{theorem}[ECI via jumps]\label{thm:ECIj}
Given a hybrid system $\mathcal{H} = (C,F,D,G)$ \blue{such that Assumption  \ref{assump:SA} holds},  and sets $\mathcal{O} \subset C \cup D$ and $\mathcal{A} \subset \mathbb{R}^n$,
suppose that there exist a $\mathcal{C}^1$ function $w : \mathbb{R}^n \rightarrow \mathbb{R}$, a nondecreasing function $f_d : \mathbb{R} \rightarrow \mathbb{R}$, and a constant $r_2 > 0$ such that items 1) and 2) in Proposition~\ref{prop:pre_ECI_j} hold.
Then, the set $\mathcal{A}$ is \emph{ECI} with respect to $\mathcal{O}$ for $\mathcal{H}$ if the following additional condition holds:
\begin{itemize}[leftmargin=.2in]
	\item For each solution $\phi \in \mathcal{S_H}(\mathcal{O})$, there exists a solution $z$ to $z^+ = f_d(z)$ starting from $w(\phi(0,0))$ satisfying $z(j) \!\in\! (-\infty, r_2)$ for all $j \geq j^\star$ and for some nonnegative $j^\star \leq \sup \{ j : (t,j) \!\in\! \dom\phi \}$.
\end{itemize}
\end{theorem}

The set $\mathcal{A}$ is ECI with respect to the set $\mathcal{O}$ for $\mathcal{H}$ if $\mathcal{A}$ is pre-ECI and the following condition holds additionally.

\begin{theorem}[From pre to non-pre ECI] \label{thm:eventual_ci}
Given a hybrid system $\mathcal{H} = (C,F,D,G)$ \blue{such that Assumption  \ref{assump:SA} holds}, and sets $\mathcal{O} \subset C \cup D$ and $\mathcal{A} \subset \mathbb{R}^n$ such that $\mathcal{A}$ is pre-ECI with respect to $\mathcal{O}$, the set $\mathcal{A}$ is \emph{ECI} with respect to the set $\mathcal{O}$ for $\mathcal{H}$ if the following property holds:
\begin{enumerate}[label={$(\star)$},leftmargin=*]
	\item \label{item:star} 
	There exists a set $S \subset C \cup D \cup \mathcal{A}$ such that $\mathcal{O} \cup \mathcal{A} \subset S$ and $S$ is forward invariant for 
	$\mathcal{H}_w = (C_w, F_w, D_w, G_w)$ in \eqref{eqn:H_m} with $Q$ therein replaced by $\mathcal{A}$.
\end{enumerate}
\end{theorem}

\begin{proof}
When \ref{item:star} holds, since the set $\mathcal{A}$ is pre-ECI with respect to the set $\mathcal{O}$ for $\mathcal{H}$, to complete the proof, it remains only to show that 
the \blue{maximal} solutions to $\mathcal{H}$ starting from $\mathcal{O} \backslash \mathcal{A}$ always reach the set $\mathcal{A}$. Proceeding by contradiction, assume the existence of a maximal solution $\phi$ to $\mathcal{H}$ starting from $\mathcal{O} \backslash \mathcal{A}$ that never reaches the set $\mathcal{A}$. We notice that each solution to $\mathcal{H}$ starting from  $\mathcal{O} \backslash \mathcal{A}$ is a solution to $\mathcal{H}_w$ provided that it does not reach the set $\mathcal{A}$. Hence, since the set $S$ is forward invariant for $\mathcal{H}_w$, we conclude that the solution $\phi$ is complete.
This property contradicts the fact that $\mathcal{A}$ is pre-ECI with respect to the set $\mathcal{O}$ for $\mathcal{H}$, which completes the proof.
\end{proof}

\begin{example}
Consider the hybrid system in Example \ref{exp:didact}. It is easy to see that the set $S := \mathcal{O} \cup \mathcal{A}$ is forward invariant for $\mathcal{H}_w$. Indeed, all the solutions to $\mathcal{H}_w$ starting from $\mathcal{O}$ flow in $\mathcal{O}$ until they reach $\mathcal{A}$.
Since from $\mathcal{A}$, every solution is discrete, complete, and remains in $\mathcal{A}$, $S$ is forward invariant.
\hfill $\triangle$
\label{ex:didact_eci}
\end{example}

\section{Sufficient Conditions for pre-FTA and FTA} \label{SectionFTA}

In this section, inspired by \cite[Proposition B.1 and Proposition B.3]{han.arXiv19} where sufficient conditions for FTA are proposed, we derive sufficient conditions for pre-FTA using Lyapunov functions. 

\subsection{Sufficient Conditions for pre-FTA} \label{SectionpreFTA}

In the following result, we use two Lyapunov functions, the first one strictly decreases along the flows and does not increase along the jumps and the second one decreases strictly along the jumps and does not increase along the flows. Different from \cite[Proposition B.1 and Proposition B.3]{han.arXiv19}, where only a single Lyapunov function is used, item 1) in \cite[Proposition B.1]{han.arXiv19} and item 1) in \cite[Proposition B.3]{han.arXiv19} are not imposed in the following result.
Those extra conditions used in \cite[Proposition B.1 and Proposition B.3]{han.arXiv19} are for non-pre FTA. Since pre-FTA is a property of complete solutions, those conditions are satisfied for free. 
Similarly, see also the results on finite-time stability for hybrid systems in \cite{188}.

\begin{theorem} [pre-FTA] 
Given a hybrid system $\mathcal{H}=(C,F,D,G)$ \blue{such that Assumption  \ref{assump:SA} holds}, and sets $\mathcal{A}, \mathcal{O} \subset \mathbb{R}^n$, suppose that the set $\mathcal{A}$ is closed and 
there exists an open set $\mathcal{N}$ that defines an open neighborhood of $\mathcal{A}$ such that $G(\mathcal{N}) \subset \mathcal{N} \subset \mathbb{R}^n$ and suppose that there exist functions $V: \mathcal{N} \rightarrow \mathbb{R}_{\geq 0}$ and $W: \mathcal{N} \rightarrow \mathbb{R}_{\geq 0}$, that are $\mathcal{C}^1$ on $\mathcal{N} \backslash \mathcal{A}$, such that $\mathcal{O} \subset L_V(r) \cap (C \cup D)$ and $\mathcal{O} \subset L_W(r) \cap (C \cup D)$, where $L_V(r) = \{x \in \mathbb{R}^n : V(x) \leq r\}$ and $L_W(r) = \{x \in \mathbb{R}^n : W(x) \leq r\}$, $r \in [0, \infty]$, are the sublevel sets of $V$ and $W$ contained in $\mathcal{N}$, respectively.
Then, the set $\mathcal{A}$ is \emph{pre-FTA} with respect to $\mathcal{O}$ for $\mathcal{H}$ if the following holds:
\begin{itemize}
	\item[1)] There exist constants $c_1 > 0, c_2 \in [0,1)$ such that
the function $V$ is positive definite with respect to $\mathcal{A}$, and
\begin{subequations}
\label{eqn:pre_FTA_f}
\begin{align}
	&\!\!\!\!\!\!\!\!\left< \nabla V(x), \eta \right> \leq -c_1 V^{c_2}(x)
	\label{eqn:pre_FTA_flow_a}\\
	&\qquad\qquad\qquad\quad \forall x \!\in\! (C \cap \mathcal{N}) \backslash \mathcal{A}, \forall \eta \!\in\! F(x) \!\cap\! T_C(x)\mbox{,}\nonumber\\
	&\!\!\!\!\!\!\!\!V(\eta) - V(x) \leq 0 \quad \forall x \!\in\! (D \cap \mathcal{N}) \backslash \mathcal{A}, \forall \eta \!\in\! G(x)\mbox{,}
	\label{eqn:pre_FTA_flow_b}
\end{align}
\end{subequations}
	\item[2)] There exists a constant $c > 0$ such that the function $W$ is positive definite with respect to $\mathcal{A}$, and
\begin{subequations}
\label{eqn:pre_FTA_j}
\begin{align}
	&\!\!\!\!\!\!\!\!\!\!\!\!\left< \nabla W(x), \eta \right> \!\leq\! 0 ~\forall x \!\in\! (C \!\cap\! \mathcal{N}) \backslash \mathcal{A}, \forall \eta \!\in\! F(x) \!\cap\! T_C(x) \mbox{,}
	\label{eqn:pre_FTA_jump_a}\\
	&\!\!\!\!\!\!\!\!\!\!\!\!W(\eta) \!-\! W(x) \!\leq\! -\min\{c, W(x)\}\label{eqn:pre_FTA_jump_b}\\
	&\qquad\qquad\qquad\qquad\forall x \!\in\! (D \cap \mathcal{N}) \backslash \mathcal{A}, \forall \eta \!\in\! G(x) \nonumber \mbox{.}
\end{align}
\end{subequations}
\end{itemize}
\label{thm:pre_FTA}
\end{theorem}

\begin{proof}
According to the definition of pre-FTA, we need to show that, for each complete solution $\phi$ to $\mathcal{H}$ starting from the set $\mathcal{O}$, $\mathcal{T}_\mathcal{A}(\mathcal{O}) < \infty$; namely, there exists $(t^\star, j^\star) \in \dom\phi$ such that  $\phi(t^\star, j^\star) \in \mathcal{A}$ such that $t^\star + j^\star = \mathcal{T}_\mathcal{A}(\mathcal{O}) < \infty$. Let $\phi$ be a complete solution to $\mathcal{H}$ starting from $\mathcal{O}$.
Due to completeness of $\phi$, we have that $\dom\phi$ is unbounded either in $t$ or in $j$; namely,
$ \sup \{t: (t,j) \in \dom\phi\} + \sup \{j: (t,j) \in \dom\phi \} = \infty\mbox{.}$

\begin{itemize}[leftmargin=.2in]
\item Assume that $\sup \{t: (t,j) \in \dom\phi\} = \infty$. In this case, we pick $(t,j) \in \dom\phi$ and a sequence $0 = t_0 \leq t_1 \leq \ldots \leq t_{j+1} = t$ satisfying
$\dom\phi \cap ([0,t] \times \{0,1,\ldots,j\}) = \bigcup_{i=0}^j ([t_i, t_{i+1}] \times \{i\})\mbox{.}$
Now, suppose that, for each $i \in \{0,1,\ldots,j\}$ and almost all $s \in [t_i, t_{i+1}]$,
$\phi(s,i) \in (C \cap L_V(r)) \backslash \mathcal{A}\mbox{.}$
We will later show that this is the case by picking $t = t_{j+1}$ and $j$ appropriately.
Note that $L_V(r) \subset \mathcal{N}$ by assumption.
Then, \eqref{eqn:pre_FTA_flow_a} implies that, for each $i \in \{0,1, \ldots, j\}$ and for almost all $s \in [t_i, t_{i+1}]$,
$
	\tfrac{d}{ds} V(\phi(s,i)) \leq -c_1 V^{c_2} (\phi(s,i))\mbox{.}
$
That is,
$ V^{-c_2}(\phi(s,i)) dV(\phi(s,i)) \leq -c_1 ds \mbox{.}$
Then, integrating over $[t_i, t_{i+1}]$ both sides of this inequality yields
\begin{equation}
	\!\!\!\tfrac{1}{1-c_2} \big( V^{1-c_2} (\phi(t_{i+1},i)) - V^{1-c_2} (\phi(t_i,i)) \big) \leq -c_1 (t_{i+1} - t_i)\mbox{.}
\label{eqn:pre_c1}
\end{equation}
Similarly, for each $i \in \{1,\ldots,j\}$,
$\phi(t_i,i-1) \in (D \cap L_V(r)) \backslash \mathcal{A}\mbox{.}$
As stated above, we will later show that this is the case by picking $t \!=\! t_{j+1}$ and $j$ appropriately.
Then, \eqref{eqn:pre_FTA_flow_b} implies that
\begin{equation}
	V(\phi(t_i,i)) - V(\phi(t_i,i-1)) \leq 0\mbox{.}
\label{eqn:pre_d1}
\end{equation}
Denoting $\xi := \phi(0,0)$, the two inequalities in \eqref{eqn:pre_c1} and \eqref{eqn:pre_d1} imply that, for each $(t,j) \in \dom\phi$,
\[
	\tfrac{1}{1-c_2} \big( V^{1-c_2} (\phi(t,j)) - V^{1-c_2} (\xi) \big) \leq -c_1 t \mbox{.}
\]
It follows that, using the fact that $c_2 \in [0,1)$,
$V^{1-c_2} (\phi(t,j)) \leq V^{1-c_2} (\xi) - c_1 (1-c_2) t\mbox{.}$
Indeed, the quantity $V^{1-c_2}(\xi) - c_1 (1-c_2) t$ converges to zero in finite time that is upper bounded by $\tfrac{V^{1-c_2} (\xi)}{c_1(1-c_2)}$. Since $V$ is positive definite with respect to $\mathcal{A}$, It follows that there exists $(t^\star, j^\star) \in \dom\phi$ such that $\phi(t^\star, j^\star) \in \mathcal{A}$ with
$
t^\star + j^\star \leq \mathcal{T}_1^\star(\xi) + \mathcal{J}_1^\star(\phi),
$
where
$\mathcal{T}_1^\star(\xi) = \tfrac{V^{1-c_2} (\xi)}{c_1(1-c_2)}\mbox{, }$

$\mathcal{J}_1^\star(\phi) = \sup \{j: (t,j) \in \dom\phi, t < \mathcal{T}_1^\star(\xi)\}\mbox{.}$
Next, we show that $\phi$ stays in $L_V(r)$ until it reaches the set $\mathcal{A}$.
Proceeding by contradiction, if $\phi$ does not stay in $L_V(r)$ until $\phi$ reaches the set $\mathcal{A}$, then there exists a first hybrid time $(t',j') \in \dom\phi$ such that
\[V(\phi(t',j')) \!>\! r\mbox{,} ~\phi(t,j) \!\notin\! \mathcal{A} ~\forall (t,j) \!\in\! \dom\phi \!:\! t + j \!\leq\! t' +j'\mbox{.}\]
Note that $\phi(0,0) \in L_V(r) \subset \mathcal{N}$.
Then, using \eqref{eqn:pre_FTA_f}, the fact that $G(\mathcal{N}) \subset \mathcal{N}$, the continuity of $V$, and the concept of solutions to hybrid inclusions, we conclude that $V(\phi(t',j')) \leq V(\phi(0,0)) \leq r$ since $\phi(t,j) \in (C \cup D) \backslash \mathcal{A}$ for every $(t,j) \in \dom\phi$ such that $t+j \leq t'+j'$. Hence, the contradiction follows.

\item Assume that $\sup \{j: (t,j) \in \dom\phi\} = \infty$. In this case, we
suppose that, for each $i \in \{0,1,\ldots,j\}$ and almost all $s \in [t_i, t_{i+1}]$,
$\phi(s,i) \in (C \cap L_W(r)) \backslash \mathcal{A}\mbox{.}$
We will later show that this is the case by picking $t = t_{j+1}$ and $j$ appropriately.
Note that $L_W(r) \subset \mathcal{N}$ by assumption.
Then, \eqref{eqn:pre_FTA_jump_a} implies that, for each $i \in \{0,1, \ldots, j\}$ and for almost all $s \in [t_i, t_{i+1}]$,
$\tfrac{d}{ds} W(\phi(s,i)) \leq 0\mbox{.}$
Then, integrating over $[t_i, t_{i+1}]$ both sides of this inequality yields
\begin{equation}
	W(\phi(t_{i+1},i)) - W(\phi(t_i,i)) \leq 0\mbox{.}
\label{eqn:pre_c2}
\end{equation}
Similarly, for each $i \in \{1,\ldots,j\}$,
$
	\phi(t_i,i-1) \in (D \cap L_W(r)) \backslash \mathcal{A}\mbox{.}
$
Then, using \eqref{eqn:pre_FTA_jump_b},
\begin{equation}
	\!\!\!\!W\!(\phi(t_i,i)) - W\!(\phi(t_i,i\!-\!1)) \!\leq\!\! -\min \{c,\! W\!(\phi(t_i,i\!-\!1))\}\mbox{.}
\label{eqn:pre_d2}
\end{equation}
Denoting $\xi := \phi(0,0)$, the two inequalities in \eqref{eqn:pre_c2} and \eqref{eqn:pre_d2} imply that, for each $(t,j) \in \dom\phi$,
$W(\phi(t,j)) - W(\xi) \leq -\sum_{i=1}^j \min \{c, W(\phi(t_i,i-1))\} \mbox{.}$ Indeed, the quantity $W(\xi) - \sum_{i=1}^j \min \{c, W(\phi(t_i, i-1))\}$ converges to zero in finite time that is upper bounded by $\mbox{ceil} \big(\tfrac{W(\xi)}{c}\big)$.
 Since $W$ is positive definite with respect to $\mathcal{A}$, It follows that there exists $(t^\star, j^\star) \in \dom\phi$ such that $\phi(t^\star, j^\star) \in \mathcal{A}$ satisfying
$
t^\star + j^\star \leq \mathcal{T}_2^\star(\phi) + \mathcal{J}_2^\star(\xi),
$
where $\mathcal{J}_2^\star(\xi) = \mbox{ceil} \big(\tfrac{W(\xi)}{c}\big)\mbox{,}$
$(\mathcal{T}_2^\star(\phi), \mathcal{J}_2^\star(\xi)), (\mathcal{T}_2^\star(\phi), \mathcal{J}_2^\star(\xi)-1) \in \dom\phi\mbox{.}$ Finally, using the same argument as in the previous item under \eqref{eqn:pre_FTA_j}, we conclude that $\phi$ stays in $L_W(r)$ until it reaches the set $\mathcal{A}$.
\qedhere
\end{itemize}
\end{proof}

\begin{remark}
It is important to note that, when only item 1) (respectively, item 2)) in Theorem \ref{thm:pre_FTA} holds, we cannot guarantee that discrete (respectively, continuous) complete solutions to $\mathcal{H}$ reach 
the set $\mathcal{A}$. 
\end{remark}

The following example illustrates Theorem \ref{thm:pre_FTA}.

\begin{example}[Bouncing ball]
\label{ex:fta_bb1}
Consider the hybrid system $\mathcal{H} = (C,F,D,G)$ in Example \ref{ex:bouncing_ball}. Let $\mathcal{A} = \mathbb{R} \times \mathbb{R}_{\leq 0}$, $\mathcal{N} = \mathbb{R}^2$, and $V(x) = W(x) = |x_2|$ for all $x \in \mathbb{R}^2$ where $V$ and $W$ are $\mathcal{C}^1$ on the open set $\mathbb{R}^2 \backslash \mathcal{A}$. Let $\mathcal{O} \subset L_V(r) \cap (C \cup D)$ where $L_V(r) = \{x \in \mathbb{R}^2 : V(x) \leq r\}$, $r \in [0, \infty]$, is the sublevel set of $V$ contained in $\mathcal{N}$.
It follows that for all $x \in (C \cap \mathcal{N}) \backslash \mathcal{A}$, $\left< \nabla V(x), F(x) \right> = -\gamma < 0$ and $\left< \nabla V(x), F(x) \right> \leq c_1 V^{c_2} (x)$ holds with $c_1 = \gamma$ and $c_2 = 0$; namely, \eqref{eqn:pre_FTA_flow_a} and \eqref{eqn:pre_FTA_jump_a} hold.
Since $(D \cap \mathcal{N}) \backslash \mathcal{A} = \emptyset$, conditions 1) and 2) in Theorem \ref{thm:pre_FTA} hold; and thus, we conclude that $\mathcal{A}$ is pre-FTA with respect to $\mathcal{O}$ for $\mathcal{H}$.
\hfill $\triangle$
\end{example}

In the following result, which is similar to \cite[Proposition B.1]{han.arXiv19},
only a single Lyapunov function is used to guarantee that only complete solutions must reach $\mathcal{A}$ via flows.

\begin{proposition}
[pre-FTA via flows] \label{prop:preFTA1}
Given a hybrid system $\mathcal{H}=(C,F,D,G)$ \blue{such that Assumption  \ref{assump:SA} holds}, and sets $\mathcal{A}, \mathcal{O} \subset \mathbb{R}^n$, suppose that the set $\mathcal{A}$ is closed and 
there exists an open set $\mathcal{N}$ that defines an open neighborhood of $\mathcal{A}$ such that $G(\mathcal{N}) \subset \mathcal{N} \subset \mathbb{R}^n$ and suppose that there exists a function $V: \mathcal{N} \rightarrow \mathbb{R}_{\geq 0}$ that is $\mathcal{C}^1$ on $\mathcal{N} \backslash \mathcal{A}$ and $\mathcal{O} \subset L_V(r) \cap (C \cup D)$ where $L_V(r) := \{x \in \mathbb{R}^n : V(x) \leq r\}$, $r \in [0, \infty]$, is the sublevel set of $V$ contained in $\mathcal{N}$.
Then, the set $\mathcal{A}$ is \emph{pre-FTA} with respect to $\mathcal{O}$ for $\mathcal{H}$ if
\begin{itemize}
	\item[1)] Condition 1) in Theorem~\ref{thm:pre_FTA} holds; and
	\item[2)] For every $x \in \mathcal{N} \cap (C \cup D) \backslash \mathcal{A}$, each complete solution 
	$\phi \in \mathcal{S_H}(x)$ satisfies
$ \tfrac{V^{1-c_2} (x)}{c_1 (1-c_2)} \leq \sup \{ t : (t,j) \in \dom\phi \} \mbox{.}$
\end{itemize}
\end{proposition}

\begin{proof}
Let $\phi$ be a complete solution to $\mathcal{H}$ starting from $\xi := \phi(0,0) \in \mathcal{O}$. Following the first item in the proof of Theorem~\ref{thm:pre_FTA} under item 1), we conclude the existence $(t^\star, j^\star) \in \dom\phi$ such that $\phi(t^\star, j^\star) \in \mathcal{A}$ provided that 
$ \sup \{t : (t,j) \in \dom\phi\} \geq \tfrac{V^{1-c_2} (\xi)}{c_1 (1-c_2)} \mbox{,}$
which, under item 2), completes the proof.
\end{proof}

\begin{example}[Thermostat]
\label{ex:fta_thermostat}
Consider the hybrid system $\mathcal{H} = (C,F,D,G)$ in Example \ref{exp:thermostat}. Let $\mathcal{A} := \{ (h,z)  \in \{0,1\} \times \mathbb{R} : z \in [z_{min},z_{max}]\}$ and $\mathcal{N} \subset \mathbb{R}^2$ be an open neighborhood of $\{0,1\} \times \mathbb{R}$.
Consider the function
$V(x) := (z-z_{max})(z-z_{min})$ if $x \notin \mathcal{A}$ and $V(x) := 0$ otherwise.
Let $\mathcal{O} \subset L_V(r) \cap (C \cup D)$, where $L_V(r) = \{x \in \{0,1\} \times \mathbb{R}: V(x) \leq r\}$, $r \in [0,\infty]$, is the sublevel set of $V$ contained in $\mathcal{N}$.
Note that, for each $x \!\in\! (C \cap \mathcal{N}) \backslash \mathcal{A} \!=\! (\{0\} \times (z_{\max}, \infty)) \cup (\{1\} \times (-\infty, z_{\min}))$, we have $\left< \nabla V(x), F(x) \right> = (2z - z_{\min} - z_{\max})(-z + z_o + z_\triangle h) < -(z_{\max}-z_{\min})^2$; hence, \eqref{eqn:pre_FTA_flow_a} holds with $c_1 = (z_{\max}-z_{\min})^2$ and $c_2 = 0$.
Furthermore, for all $x \in (D \cap \mathcal{N}) \backslash \mathcal{A} = (\{0\} \times (-\infty, z_{\min})) \cup (\{1\} \times (z_{\max}, \infty))$, $V(G(x)) - V(x) = 0$; hence, \eqref{eqn:pre_FTA_flow_b} holds.
Next, item 2) in Theorem \ref{thm:pre_FTA} holds for complete solutions as they are nonzeno, which implies that their domain is unbounded in the $t$ direction. Thus, via Theorem \ref{thm:pre_FTA}, we conclude that $\mathcal{A}$ is pre-FTA with respect to $\mathcal{O}$ for $\mathcal{H}$.
\hfill $\triangle$
\end{example}

In the following result, which is similar to \cite[Proposition B.3]{han.arXiv19},
only a single Lyapunov function is used to guarantee that only complete solutions must reach $\mathcal{A}$ via jumps.

\begin{proposition}
[pre-FTA via jumps] \label{prop:preFTA2}
Given a hybrid system $\mathcal{H}=(C,F,D,G)$ \blue{such that Assumption  \ref{assump:SA} holds}, and sets $\mathcal{A}, \mathcal{O} \subset \mathbb{R}^n$, suppose that the set $\mathcal{A}$ is closed and 
there exists an open set $\mathcal{N}$ that defines an open neighborhood of $\mathcal{A}$ such that $G(\mathcal{N}) \subset \mathcal{N} \subset \mathbb{R}^n$ and suppose that there exists a function $W : \mathcal{N} \rightarrow \mathbb{R}_{\geq 0}$ that is $\mathcal{C}^1$ on $\mathcal{N} \backslash \mathcal{A}$ and $\mathcal{O} \subset L_W(r) \cap (C \cup D)$
where $L_W(r) = \{x \in \mathbb{R}^n : W(x) \leq r\}$, $r \in [0, \infty]$, is the sublevel set of $W$ contained in $\mathcal{N}$. Then, the set $\mathcal{A}$ is \emph{pre-FTA} with respect to $\mathcal{O}$ for $\mathcal{H}$ if
\begin{itemize}
	\item[1)] Condition 2) Theorem~\ref{thm:pre_FTA} holds; and
	\item[2)] For every $x \in \mathcal{N} \cap (C \cup D) \backslash \mathcal{A}$, each complete $\phi \in \mathcal{S}_{\mathcal{H}}(x)$ satisfies
$
    	\mbox{ceil} \big( \tfrac{W(x)}{c} \big) \leq \sup \{ j : (t,j) \in \dom\phi \} \mbox{.}
$
\end{itemize}
\end{proposition}

\begin{proof}
Let $\phi$ be a complete solution to $\mathcal{H}$ starting from $\xi := \phi(0,0) \in \mathcal{O}$. Following the second item in the proof of Theorem~\ref{thm:pre_FTA} under item 1), we conclude the existence $(t^\star, j^\star) \in \dom\phi$ such that $\phi(t^\star, j^\star) \in \mathcal{A}$ provided that 
$
\sup \{j : (t,j) \in \dom\phi\} \geq 
\mbox{ceil} \big( \tfrac{W(\xi)}{c} \big) \mbox{,}
$
which, under item 2), completes the proof.
\end{proof}

\begin{example}[Boucing ball]
\label{ex:fta_bb2}
Consider the hybrid system $\mathcal{H} = (C,F,D,G)$ in Example \ref{ex:bouncing_ball}
with $\gamma > 0$ and $\lambda \in (0,1)$.
Let $\mathcal{A} = \{ x \in C \cup D : 2 \gamma x_1 + (x_2-1) (x_2+1) \leq 0 \}$ and $\mathcal{N} = \mathbb{R}^2$. 
We observe that $\mathcal{A}$ is the sublevel set where the total energy of the ball is less or equal than $1/2$. Consider the function 
\begin{align} \label{EqW}
W(x) := 
\left\{
\begin{array}{ll}
2 \gamma x_1 + x_2^2 - 1 
& \mbox{if } x \in \mathbb{R}^2 \backslash \mathcal{A}
\\
0 & \mbox{otherwise.}
\end{array}
\right.
\end{align}
Let $\mathcal{O} \!\subset\! L_W(r) \cap (C \cup D)$, where $L_W(r) = \{x \in \mathbb{R}^2: W(x) \leq r\}$, $r \in [0, \infty]$, is the sublevel set of $W$ contained in $\mathcal{N}$.
For each $x \in (C \cap \mathcal{N}) \backslash \mathcal{A} = \{x \in \mathbb{R}^2: 2 \gamma x_1 + x_2^2 - 1 > 0, x_1 \leq 0\}$, $\left< \nabla W(x), F(x) \right> = 0$; hence, \eqref{eqn:pre_FTA_jump_a} holds.
Furthermore, for each $x \in (D \cap \mathcal{N}) \backslash \mathcal{A} = \{0\} \times (-\infty, -1]$, $W(G(x)) - W(x) \leq \lambda^2 -1 < -W(x)$; hence, \eqref{eqn:pre_FTA_jump_b} holds with $c = 1-\lambda^2$. Finally, item 2) holds since every maximal solution is complete and achieves unbounded amount of jumps.
Hence, we conclude that $\mathcal{A}$ is pre-FTA with respect to $\mathcal{O}$ for $\mathcal{H}$.
\hfill $\triangle$
\label{ex:bouncing_pre_fta}
\end{example}

\subsection{Sufficient Conditions for FTA} \label{SectionpreFTAA}

\begin{theorem}[FTA via flows]\label{thm:FTAf}
Consider a hybrid system $\mathcal{H}=(C,F,D,G)$ and two sets $\mathcal{A}, \mathcal{O} \subset \mathbb{R}^n$. Suppose that the set $\mathcal{A}$ is closed and there exists an open set $\mathcal{N}$ that defines an open neighborhood of $\mathcal{A}$ such that $G(\mathcal{N}) \subset \mathcal{N} \subset \mathbb{R}^n$. Suppose that there exists a function $V: \mathcal{N} \rightarrow \mathbb{R}_{\geq 0}$ that is $\mathcal{C}^1$ on $\mathcal{N} \backslash \mathcal{A}$ and $\mathcal{O} \subset L_V(r) \cap (C \cup D)$ with $L_V(r) := \{x \in \mathbb{R}^n : V(x) \leq r\}$, $r \in [0, \infty]$, is a sublevel set of $V$ contained in $\mathcal{N}$.
Suppose that the conditions in Proposition~\ref{prop:preFTA1} hold.
Then, the set $\mathcal{A}$ is \emph{FTA} with respect to $\mathcal{O}$ for $\mathcal{H}$ if the following additional condition holds:
\begin{itemize}
	\item for every $x \in \mathcal{N} \cap (C \cup D) \backslash \mathcal{A}$, each $\phi \in \mathcal{S_H}(x)$ satisfies
	\begin{equation*}
		\tfrac{V^{1-c_2} (x)}{c_1 (1-c_2)} \leq \sup \{ t : (t,j) \in \dom\phi \} \mbox{.}
	\end{equation*}
\end{itemize}
\end{theorem}

\begin{proof}
The proof is the same as in Proposition \ref{prop:preFTA1}.
\end{proof}

\begin{theorem}[FTA via jumps]\label{thm:FTAj}
Consider a hybrid system $\mathcal{H}=(C,F,D,G)$ \blue{such that Assumption  \ref{assump:SA} holds},  and two sets $\mathcal{A}, \mathcal{O} \subset \mathbb{R}^n$. Suppose that the set $\mathcal{A}$ is closed and 
there exists an open set $\mathcal{N}$ that defines an open neighborhood of $\mathcal{A}$ such that $G(\mathcal{N}) \subset \mathcal{N} \subset \mathbb{R}^n$.
Suppose that there exists a function $W : \mathcal{N} \rightarrow \mathbb{R}_{\geq 0}$ that is $\mathcal{C}^1$ on $\mathcal{N} \backslash \mathcal{A}$ and $\mathcal{O} \subset L_W(r) \cap (C \cup D)$
where $L_W(r) = \{x \in \mathbb{R}^n : W(x) \leq r\}$, $r \in [0, \infty]$, is a sublevel set of $W$ contained in $\mathcal{N}$.
Suppose that the conditions in Proposition~\ref{prop:preFTA2} hold.
Then, the set $\mathcal{A}$ is \emph{FTA} with respect to $\mathcal{O}$ for $\mathcal{H}$ if the following additional condition holds:
\begin{itemize}
	\item for every $x \in \mathcal{N} \cap (C \cup D) \backslash \mathcal{A}$, each $\phi \in \mathcal{S}_{\mathcal{H}}(x)$ satisfies
	\begin{equation*}
    	\mbox{ceil} \big( \tfrac{W(x)}{c} \big) \leq \sup \{ j : (t,j) \in \dom\phi \} \mbox{.}
	\end{equation*}
\end{itemize}
\end{theorem}

\begin{proof}
The proof is the same as in Proposition \ref{prop:preFTA2}.
\end{proof}

\begin{theorem} 
[From pre to non-pre FTA]
\label{thm:FTA}
Consider a hybrid system $\mathcal{H} = (C,F,D,G)$ \blue{such that Assumption  \ref{assump:SA} holds}, and two sets 
$\mathcal{O} \subset C \cup D$ and $\mathcal{A} \subset \mathbb{R}^n$.
The set $\mathcal{A}$ is \emph{FTA} with respect to $\mathcal{O}$ for $\mathcal{H}$ if the set $\mathcal{A}$ is pre-FTA with respect to $\mathcal{O}$ for $\mathcal{H}$ and \ref{item:star} holds. 
\end{theorem}

\begin{proof}
The proof uses the exact same steps as in proof of Theorem \ref{thm:eventual_ci}. Hence, it is omitted.
\end{proof}

\begin{remark}
Note that both ECI and FTA suggest that the solutions starting from a given set reach a target set in finite time (ECI suggests, additionally, staying in the target after reaching it). However, the sufficient conditions in Theorems \ref{thm:pre_eventual_ci} and \ref{thm:pre_FTA} are quite different. Indeed, in the first case, we use comparison techniques, via an appropriate choice of scalar functions, to build a scalar version of the original system. As consequence, we show that ECI is satisfied for the original system if the same property holds for the scalar one.  However, in the second case, we look for appropriate scalar functions (Lyapunov candidates with respect to the target set), for which the strict decrease along the solutions to the original system implies reaching the target set in finite time.
\end{remark}

\section{Sufficient Conditions for $p \,\mathcal{U}_w q$ and $p \,\mathcal{U}_s q$}
\label{sec:sufficient_conditions_until}

In this section, we combine the results in Section \ref{sec:characterization} and the Lyapunov-like techniques developed in Section \ref{Section.CIECI} to propose sufficient infinitesimal conditions certifying the formulas $p \,\mathcal{U}_w q$ and $p \,\mathcal{U}_s q$.

\subsection{Certifying 
$p \,\mathcal{U}_w q$ using Sufficient Conditions for CI}

First, we present sufficient conditions that guarantee the satisfaction of the formula $p \,\mathcal{U}_w q$ by employing the sufficient conditions for conditional invariance in Proposition~\ref{prop:conditional_invariance}.

\begin{theorem}[$p \,\mathcal{U}_w q$ via CI] \label{thm:weak_until}
Consider a hybrid system $\mathcal{H} \!=\! (C,F,D,G)$. Given atomic propositions $p$ and $q$, let the sets $P$ and $Q$ be as in \eqref{eqn:K_sets} \blue{such that Assumptions \ref{assump:SA1} and \ref{assump:SA} hold}. Then, the formula $p \,\mathcal{U}_w q$ is satisfied for $\mathcal{H}$
if there exists a $\mathcal{C}^1$ barrier function candidate $B$ with respect to  
\blue{$ (\mathcal{O}, \mathcal{X}_u) := (P \backslash Q, \mathbb{R}^n \backslash (P \cup Q))$} for $\mathcal{H}$ as in \eqref{eqn:barrier_candidate_ci}
such that 
$$ K := \{x \in C \cup D \cup Q : B(x) \leq 0\} $$ 
is closed and the following hold:
\begin{itemize}
\item[1)] $\left< \nabla B(x), \eta \right> \leq 0$ for all $x \in (C \backslash Q) \cap (U(\partial K) \backslash K)$ and all $\eta \in F(x) \cap T_{C \backslash Q}(x)$.
\item[2)] $B(\eta) \leq 0$ for all $x \in K \cap (D \backslash Q)$ and all $\eta \in G(x)$.
\item[3)] $G(x) \subset C \cup D \cup Q$ for all $x \in K \cap (D \backslash Q)$.
\end{itemize}
\end{theorem}

\begin{proof}
Let the system $\mathcal{H}_w \!=\! (C_w, F_w, D_w, G_w)$ be as in \eqref{eqn:H_m}. Since $K \!=\! \{x \in C \cup D \cup Q : B(x) \leq 0\}$ and the barrier function candidate $B$ satisfies $B(x) \leq 0$ for all $x \in P \backslash Q $ and $B(x) > 0$ for all $x \in (C \cup D) \backslash (P \cup Q) = (C \cup D \cup Q) \backslash (P \cup Q)$, we conclude that $B$ is a barrier candidate with respect to $(P \backslash Q, \mathbb{R}^n \backslash (P \cup Q))$ for $\mathcal{H}_w$ in \eqref{eqn:H_m}.
Furthermore, item 1) implies that $\left< \nabla B(x), \eta \right> \leq 0$ for all $x \in (U(\partial K) \backslash K) \cap C_w$ and all $\eta \in F(x) \cap T_{C_w}(x)$.
Item 2) implies that $B(\eta) \leq 0$ for all $x \in K \cap (D \backslash Q)$ and all $\eta \in G_w(x)$. Furthermore, when $x \in K \cap Q$, $G_w(x) \!=\! x$ and $B(x) \!\leq\! 0$. Hence, $B(\eta) \!\leq\! 0$ for all $x \in K \!\cap\! D_w$ and all $\eta \in G_w(x)$.
Item 3) implies that $G_w(K \!\cap\! (D \backslash Q)) \!\subset\! C_w \cup D_w$. Furthermore, $G_w(K \!\cap\! Q) \!\subset\! K \cap Q \!\subset\! C_w \cup D_w$. Hence, $G_w(K \cap D_w) \!\subset\! C_w \cup D_w$. Thus, using item 1) in Proposition \ref{prop:conditional_invariance} with $\mathcal{O}$ and $\mathcal{X}_u$ therein replace by $P \backslash Q$ and $\mathbb{R}^n \backslash (P \cup Q)$, respectively, we conclude that $P \cup Q$ is CI with respect to $P \backslash Q$ for $\mathcal{H}_w$. Hence, using Theorem \ref{thm:weakuntil_ci}, we conclude that $p \,\mathcal{U}_w q$ is satisfied for $\mathcal{H}$. 
\end{proof}

\begin{example}[Bouncing ball]\label{expBB1}
Consider the bouncing ball example in Example \ref{ex:bouncing_ball} to confirm the conclusions therein using Theorem \ref{thm:weak_until}. Consider the barrier function candidate $B(x) \!:=\! x_1 - \varepsilon$ with $\varepsilon \!>\! 0$.
Indeed, $B$ is a barrier function candidate with respect to $(P \backslash Q, \mathbb{R}^n \backslash (P \cup Q) )$ for $\mathcal{H}$ since $B(x) \!\leq\! 0$ for all $x \!\in\! P \backslash Q \!=\! [0, \varepsilon] \!\times\! \mathbb{R}_{<0}$ and $B(x) \!>\! 0$ for all $x \!\in\! (C \cup D) \backslash (P \cup Q) \!=\! (\varepsilon, \infty) \!\times\! \mathbb{R}_{<0}$.
Moreover, for all $x \!\in\! C \backslash Q \!=\! \mathbb{R}_{\geq 0} \!\times\! \mathbb{R}_{\leq 0}$, we have $\langle \nabla B(x), F(x) \rangle \!=\! x_2 \!\leq\! 0$; thus, item 1) holds. Moreover, for all $x \!\in\! K \cap D \!=\! D$, $B(G(x)) \!=\! B(x) \!\leq\! 0$; thus, item 2) holds. Finally, for all $x \!\in\! D$, $G(x) \!\in\! \{ 0 \} \!\times\! \mathbb{R}_{\geq 0} \!\subset\! C$; hence, item 3) holds. Thus, using Theorem \ref{thm:weak_until}, we conclude that $p \,\mathcal{U}_w q$ is satisfied for $\mathcal{H}$.
\hfill $\triangle$
\end{example}

In the following, item 1) in Theorems \ref{thm:strong_until}, \ref{thm:strong_ECI_f}, \ref{thm:strong_ECI_j}, and \ref{thm:strong_ECI_interval} can be guaranteed via Theorem \ref{thm:weak_until}.

\begin{example}[Thermostat] \label{exp:thermostat12}
\blue{We reconsider the thermostat hybrid model in Example \ref{exp:thermostat} in order to show that the formula $p \,\mathcal{U}_w q$ is satisfied for $\mathcal{H}$ using Theorem \ref{thm:weak_until}. Indeed, consider the barrier candidate $B(x) := (2h - 1) (z - z_{\max})$. It is easy to see that $B$ is a barrier candidate with respect to $( P \backslash Q, \mathbb{R}^n \backslash (P \cup Q) )$ for $\mathcal{H}$. Furthermore, for all $x \in C \backslash Q = (\{ 1 \} \times \mathbb{R}) \cup (\{ 0 \} \times (-\infty, z_{\max}) )$, we have $\langle \nabla B(x), F(x) \rangle = (2h - 1) (-z + z_o + z_\triangle h) \leq 0$; hence, item 1) holds. Furthermore, for all $x \in K \cap D = [1 \quad z_{\max}]^\top$, $B(G(x)) = B([0 \quad z_{\max}]^\top) \leq 0$; hence, item 2) holds.  Finally, for all $x \in D$, $G(x) \in C$; hence, item 3) holds. As a consequence, using Theorem \ref{thm:weak_until}, we conclude that the formula $p \,\mathcal{U}_w q$ is satisfied for $\mathcal{H}$.
\hfill $\triangle$}
\end{example}

\subsection{Certifying $p \,\mathcal{U}_s q$ using Sufficient Conditions for ECI via Flows and Jumps}

Next, we present sufficient conditions to guarantee the satisfaction of the formula $p \,\mathcal{U}_s q$ by using sufficient conditions for ECI in Theorem \ref{thm:pre_eventual_ci}. The following lemma is used in the proof of some of the following results.

\begin{lemma}
Consider a hybrid system $\mathcal{H} \!=\! (C,F,D,G)$ and propositions $p$ and $q$ such that $p \,\mathcal{U}_w q$ is satisfied for $\mathcal{H}$.
\begin{itemize}
	\item[1)] Every maximal solution to $\mathcal{H}_s$ from $P \backslash Q$ that never reaches $Q$ is also a maximal solution to $\mathcal{H}$.
	\item[2)] Every maximal solution to $\mathcal{H}$ from $P \backslash Q$ that never reaches $Q$ is also a maximal solution to $\mathcal{H}_s$.
\end{itemize}
\label{lemma:H_Hs}
\end{lemma}

\begin{proof}
To show item 1), we let $\phi$ be a maximal solution to $\mathcal{H}_s$ from $P \backslash Q$ that never reaches $Q$.
Proceeding by 
contradiction, we suppose that there exists a solution $\psi$ to $\mathcal{H}$ that is a nontrivial extension of $\phi$; namely, there exists $I \subset \mathbb{R}_{\geq 0} \times \mathbb{N}$ such that $I \neq \emptyset$, $\dom\phi \cup I = \dom\psi$, and $\dom\psi \backslash \dom\phi$ is nonempty. \blue{Note that $\phi$ is a maximal solution to $\mathcal{H}_s$ from $P \backslash Q$ that never reaches the closed set $Q$. Hence, $\phi$ is a solution to $\mathcal{H}$.  Using the $p \, U_w q$ property, we conclude that  $\phi(\text{dom} \phi) \subset P \backslash Q$. Finally, since by definition, $\psi$ is an extension of $\phi$, the relation
$\psi(\text{dom} \phi) = \phi(\text{dom} \phi) \subset P \backslash  Q$ follows.}
Now, since $\psi$ satisfies $p \,\mathcal{U}_w q$, we conclude that $\psi$, starting from $P \backslash Q$, either remains in $P \backslash Q$ for all hybrid time or remains in $P \backslash Q$ up to when it reaches $Q$. \blue{Now, we choose $I$ such that  $\psi$ remains in $P \cup Q$ on $\dom \phi \cup I$, and may reach $Q$ for the first time at the end of $\dom \phi \cup I$. However, in this case, we conclude that $\psi$ is solution to $\mathcal{H}_s$, which contradicts the fact that $\phi$ is a maximal solution to $\mathcal{H}_s$.}

We also proceed by contradiction to show item 2).
Assume the existence of a maximal solution $\phi$ to $\mathcal{H}$ from $P \backslash Q$ that never reaches $Q$.
Since $\phi$ never reaches $Q$ and satisfies $p \, \mathcal{U}_w q$, $\phi$ has to remain in $P \backslash Q$.
Then, suppose the existence of a solution $\psi$ to $\mathcal{H}_s$ that is a nontrivial extension of $\phi$; namely, there exists $I \subset \mathbb{R}_{\geq 0} \times \mathbb{N}$ such that $I \!\neq\! \emptyset$, $\dom\phi \cup I = \dom\psi$, and $\dom\psi \backslash \dom\phi$ is nonempty.
In particular, pick $I$ such that $\psi$ remains in $P \backslash Q$ on $\dom\phi \cup I$ or remains in $P \cup Q$ on $\dom \phi \cup I$, but reaches $Q$ for the first time at the end of $\dom\phi \cup I$.
However, the solution $\psi$ is also a solution $\mathcal{H}$ up to when $\psi$ reaches $Q$, contradicting maximality of $\phi$.
\end{proof}

\begin{theorem} [$p \,\mathcal{U}_s q$ using ECI]
\label{thm:strong_until}
Consider a hybrid system $\mathcal{H} \!=\! (C,F,D,G)$.
Let the system $\mathcal{H}_s = (C_s, F_s, D_s, G_s)$ be as in \eqref{eqn:H_m+}.
Given atomic propositions $p$ and $q$, let the sets $P$ and $Q$ be as in \eqref{eqn:K_sets} such that \blue{such that Assumptions \ref{assump:SA1} and \ref{assump:SA} hold}. Then, the formula 
$p \,\mathcal{U}_s q$ is satisfied for $\mathcal{H}$ if the following hold:
\begin{itemize}
    \item[1)] The formula $p \,\mathcal{U}_w q$ is satisfied for $\mathcal{H}$.
    \item[2)] There exist a $\mathcal{C}^1$ function $v : \mathbb{R}^n \rightarrow \mathbb{R}$, a locally Lipschitz function $f_c : \mathbb{R} \rightarrow \mathbb{R}$, and a constant $r_1 > 0$ such that the following hold:
     \begin{itemize}
        \item[2.1)] $\left< \nabla v(x), \eta \right> \leq f_c(v(x))$ for all $x \in \cl(C_s)$ and for all $\eta \in F(x) \cap T_{\cl(C_s)}(x)$;
        \item[2.2)] $v(\eta) \leq v(x)$ for all $x \in D \cap P$ and for all $\eta \in G(x)$;
        \item[2.3)]
        The solutions to $\dot{y} \!=\! f_c(y)$, starting from $v(P \backslash Q)$, converge to $(-\infty, r_1)$ in finite time.
	 \end{itemize}
	\item[3)] There exist a $\mathcal{C}^1$ function $w : \mathbb{R}^n \rightarrow \mathbb{R}$, a nondecreasing function $f_d : \mathbb{R} \rightarrow \mathbb{R}$, and a constant $r_2 > 0$ such that the following hold:
	\begin{itemize}
	    \item[3.1)] $\left< \nabla w(x), \eta \right> \leq 0$ for all $x \in \cl(C_s)$ and for all $\eta \in F(x) \cap T_{\cl(C_s)}(x)$;
	    \item[3.2)] $w(\eta) \!\leq\! f_d(w(x))$ for all $x \!\in\! D \cap P$ and for all $\eta \!\in\! G(x)$;
	    \item[3.3)]
	    The solutions to $z^+ \!=\! f_d(z)$, starting from $w(P \backslash Q)$, converge to $(-\infty, r_2)$ in finite time.
	\end{itemize}
	\item[4)] One of the following conditions holds:
	\begin{itemize}
		\item[4a)] Each complete solution to $\mathcal{H}$ starting from $P \backslash Q$ is eventually continuous and, with $r_1$ coming from item 2),
		\begin{equation}
			S_1 := \{x \in \cl(C_s) : v(x) < r_1\} \subset Q\mbox{.}
		\label{eqn:s1_strong}
		\end{equation}
    	\item[4b)] Each complete solution to $\mathcal{H}$ starting from $P \backslash Q$ is eventually discrete and, with $r_2$ coming from item 3),
    	\begin{equation}
			S_2 := \{x \in D_s : w(x) < r_2\} \subset Q\mbox{.}
		\label{eqn:s2_strong}
		\end{equation}
    	\item[4c)] Each complete solution to $\mathcal{H}$ starting from $P \backslash Q$ is eventually continuous, eventually discrete, or has a hybrid time domain that is unbounded in both the $t$ and the $j$ direction and, with $r_1$ and $r_2$ coming from item 2) and item 3) respectively, \eqref{eqn:s1_strong} and \eqref{eqn:s2_strong} hold.
    	\item[4d)] With $r_1$ and $r_2$ coming from item 2) and item 3) respectively, \eqref{eqn:s1_strong} and \eqref{eqn:s2_strong} hold,
    	and $G(S_2) \cap \cl(C_s) \subset S_1$.
	\end{itemize}
    \item[5)] No maximal solution to $\mathcal{H}$ has a finite escape time in $(P\backslash Q) \cap C$.
    \item[6)] Every maximal solution to $\mathcal{H}$ from $((P\backslash Q) \cap \partial C) \backslash D$ is nontrivial.
\end{itemize}
\end{theorem}
\begin{proof}
By item 1), every maximal solution to $\mathcal{H}$ starting from $P \cup Q$ satisfies $p \,\mathcal{U}_w q$. Hence, it remains to show that every maximal solution to $\mathcal{H}$ starting from $P \backslash Q$ reaches $Q$.
To this end, note that each maximal solution $\phi$ to $\mathcal{H}$ from $P \backslash Q$ must satisfy one of the following conditions:
\begin{itemize}
	\item[a)] $\phi$ reaches $Q$ in finite hybrid time;
	\item[b)] $\phi$ is not complete and does not reach $Q$ in finite hybrid time; or
	\item[c)] $\phi$ is complete and does not reach $Q$ in finite hybrid time.
\end{itemize}

In the rest of the proof, we show that $\phi$ can only satisfy case a).
To do so, we first show that case b) is not possible, due to items 5) and 6), using contradiction.
That is, suppose $\phi$ is not complete and never reaches $Q$; in particular, $\dom\phi$ is bounded.
Let $(T,J) \!=\! \sup\dom\phi$. 
Due to the fact that $\phi$ never reaches $Q$ and since $\phi$ satisfies $p \,\mathcal{U}_w q$, we conclude that $\phi$ remains in $P \backslash Q$. 
Moreover, under item 5),
the maximal solution $\phi$ does not have a finite escape time inside $(P \backslash Q) \cap C$, which implies that $(T,J) \!\in\! \dom\phi$.
Now, by the definition of solutions to $\mathcal{H}$, $\phi(T,J) \in \cl(C) \cup D$.
First, let $\phi(T,J) \in D$.
In this case, $\phi$ can be extended via a jump.
Next, let $\phi(T,J) \in \cl(C) \backslash D$.
In this case, when $\phi(T,J) \in \textrm{int}(C) \backslash D$, we use Assumption \ref{assump:SA} to conclude that $\phi$ can be extended via flow; and for the case when $\phi(T,J) \in \partial C \backslash D$, we use item 6) to conclude that $\phi$ can be extended via flow. 
Therefore, if $(T,J) \in \dom\phi$, then $\phi$ can be extended via flow or a jump. This contradicts maximality of $\phi$; and thus, case b) is not possible.

Next, we show that case c) is not possible due to items 2)-4) using contradiction. Suppose that items 2), 3), and 4a) hold.
Suppose that there exists a complete solution $\phi$ to $\mathcal{H}$ that does not reach $Q$ in finite hybrid time.
By Lemma \ref{lemma:H_Hs},
$\phi$ is also a maximal solution to $\mathcal{H}_s$.
However, using the arguments in a) in the proof of Theorem \ref{thm:pre_eventual_ci}, there must exist $(t^\star, j^\star) \!\in\! \dom\phi$ such that $\phi(t,j) \in S_1 \!\subset\! Q$ for all $(t,j) \!\in\! \dom\phi$ and $t+j \!\geq\! t^\star+j^\star$.
This implies that $\phi$ must reach $Q$ in finite hybrid time via flow.
Next, suppose that items 2), 3), and 4b) hold.
Proceeding as when 4a) holds, we use Lemma \ref{lemma:H_Hs} to conclude that $\phi$ is a maximal solution to $\mathcal{H}_s$.
Furthermore, using the arguments in b) in the proof of Theorem \ref{thm:pre_eventual_ci},
we conclude the existence of $(t^\star, j^\star) \!\in\! \dom\phi$ such that $\phi(t,j) \!\in\! S_2 \!\subset\! Q$ for all $(t,j) \!\in\! \dom\phi$ and $t+j \!\geq\! t^\star+j^\star$.
This implies that $\phi$ must reach $Q$ in finite hybrid time by jumps.
Similarly, suppose that items 2) and 3) hold and either item 4c) or item 4d) holds.
Using Lemma \ref{lemma:H_Hs}
and the arguments in the proof of Theorem \ref{thm:pre_eventual_ci}, we conclude that there exists $(t^\star, j^\star) \!\in\! \dom\phi$ such that $\phi(t,j) \!\in\! S_1 \cup S_2 \!\subset\! Q$ for all $(t,j) \!\in\! \dom\phi$ and $t+j \!\geq\! t^\star+j^\star$.
This implies that $\phi$ must reach $Q$ in finite hybrid time via flow or jumps.
Therefore, we conclude that case c) is not possible.
\end{proof}

\begin{remark}
We note that sufficient conditions for the satisfaction of $p \,\mathcal{U}_s q$ for $\mathcal{H}$ do not require solutions to $\mathcal{H}$ to stay in the target set $Q$ after reaching it.
Hence, when sufficient conditions that guarantee ECI are employed to derive sufficient conditions for strong until, item 4) in Theorem \ref{thm:strong_until} can be relaxed since item 4) is for guaranteeing solutions to $\mathcal{H}$ to stay in the target set $Q$ after reaching it.
In particular, $G(S_2) \cap \cl(C_s) \subset S_1$, as in item 4b) in Theorem \ref{thm:strong_until}, is not really needed since it requires solutions to $\mathcal{H}$ to stay in $Q$ after reaching $S_2 \subset Q$. 
\label{remark:eci_fta}
\end{remark}

The following example illustrates Theorem \ref{thm:strong_until}.

\begin{example}
[Thermostat]
\label{ex:thermostat_eci}
Consider the hybrid system $\mathcal{H} = (C,F,D,G)$ with the state $x \!:=\! (h,z) \!\in\! \{0,1\} \!\times\! \mathbb{R}$ modeling a controlled thermostat system in Example \ref{exp:thermostat}.
Following the formulation therein, a specification of interest is
that the room temperature decreases once it exceeds $z_{\max}$, which is related to the satisfaction of $p \,\mathcal{U}_s q$ for $\mathcal{H}$
with the propositions $p$ and $q$ given as in Example \ref{exp:thermostat}.
The sets $P$ and $Q$ in \eqref{eqn:K_sets} are given by $P \!=\! \{1\} \!\times\! (-\infty, z_{\max}]$ with $z_{\max} > 0$ and $Q \!=\! \{0\} \!\times\! [z_{\min}, \infty)$.
To show the satisfaction of $p \,\mathcal{U}_s q$ for $\mathcal{H}$, we apply Theorem \ref{thm:strong_until}.
First, consider the barrier function candidate $B(x) \!:=\! (-1)^h(z_{\max} \!-\! z)$.
Indeed, $B$ is a barrier function candidate with respect to $(P \backslash Q, (\{0,1\}\times \mathbb{R}) \backslash (P \cup Q))$ for $\mathcal{H}$ since
for all $x \!\in\! P \backslash Q \!=\! P$, $B(x) \!=\! z-z_{\max} \!\leq\! 0$; and for all $x \!\in\! (C \cup D) \backslash (P \cup Q) \!=\! (\{0\} \!\times\! (-\infty,z_{\min})) \!\cup\! (\{1\} \!\times\! (z_{\max}, \infty))$, $B(x) \!>\! 0$.
Moreover, for all $x \!\in\! C \backslash Q \!=\! \{1\} \!\times\! (-\infty, z_{\max}]$, $B(x) \!=\! z \!-\! z_{\max} \!\leq\! 0$; and thus, $(C \backslash Q) \cap (U(\partial K) \backslash K)\!=\! \emptyset$. Furthermore, for all $x \!\in\! K \cap (D \backslash Q) \!=\! \{(1, z_{\max})\}$, $G(x) \!=\! \{(0, z_{\max})\}  \!\subset\! C \cup D \cup Q$ and $B(G(x)) \!=\! 0$; hence, items 1) - 3) in Theorem \ref{thm:weak_until} hold. It follows that the formula $p \,\mathcal{U}_w q$ is satisfied for $\mathcal{H}$; and thus, item 1) holds.

Next, consider the functions $v(x) \!=\! -z \!+\! z_o \!+\! z_{\triangle}$ and $f_c(y) \!=\! -y$.
Recall that $z_{\max} \!<\! z_o \!+\! z_{\triangle}$.
For all $x \!\in\! \cl(C_s) \!=\! \{1\} \!\times\! (-\infty, z_{\max}]$, $\left< \nabla v(x), F(x) \right> \!=\! z \!-\! z_o \!-\! z_{\triangle} \!\leq\! f_c(v(x)) \!=\! z \!-\! z_o \!-\! z_{\triangle}$; hence, item 2.1) holds. Moreover, for all $x \!\in\! D \cap P \!=\! \{(1, z_{\max})\}$,
$v(G(x)) \!-\! v(x) \!=\! 0$; hence, item 2.2) holds.
Furthermore, the solutions to $\dot{y} \!=\! f_c(y)$ from $v(P \backslash Q) \!=\! [z_o \!+\! z_{\triangle} \!-\! z_{\max}, \infty)$ reach $(-\infty, z_o \!+\! z_{\triangle} - z_{\max})$; and thus, item 2.3) holds for $r_1 \!=\! z_o \!+\! z_{\triangle}-z_{\max}$.
Moreover, $S_1 \!=\! \{x \in \cl(C_s): v(x) < z_o \!+\! z_{\triangle}-z_{\max}\}$ is empty.

On the other hand, for the functions $w(x) \!=\! -z \!+\! z_{\max}$ and $f_d(z) \!=\! z/2$,
for all $x \!\in\! \cl(C_s)$, $\left< \nabla w(x), F(x) \right> \!=\! z \!-\! z_o \!-\! z_{\triangle} < 0$ since $z_{\max} \!<\! z_o \!+\! z_{\triangle}$; and for all $x \!\in\! D \cap P \!=\! \{(1, z_{\max})\}$, $w(G(x)) \!=\! 0 \!=\! f_c(w(x))$. Hence, we conclude that items 3.1) and 3.2) hold. Moreover, the solutions to $z^+ \!=\! f_d(z)$ starting from $w(P \backslash Q) \!=\! (0, \infty)$ reach $(-\infty, z_{\max} \!-\! z_{\min})$; and thus, item 3.3) holds for $r_2 \!=\! z_{\max} \!-\! z_{\min}$.
Moreover, every complete solution to $\mathcal{H}_s$ is eventually discrete due to the jump map $G_s$ and $S_2 \!=\! \{0\} \!\times\! (z_{\min}, \infty) \!\subset\! Q$.
Hence, with $S_1$ and $S_2$ satisfying \eqref{eqn:s1_strong} and \eqref{eqn:s2_strong}, item 4c) holds.
Furthermore, since the flow map $F$ has global linear growth on $(P \backslash Q) \cap C$, the solutions to $\mathcal{H}$ do not have a finite escape time inside $(P \backslash Q) \cap C$; hence, item 5) holds.
Finally, since $(P \backslash Q) \cap \partial C \!=\! \emptyset$, item 6) holds. Thus, we conclude that $p \,\mathcal{U}_s q$ is satisfied for $\mathcal{H}$.
\hfill $\triangle$
\end{example}

\nDraft{\vspace{-0.2in}}
\subsection{Certifying 
$p \,\mathcal{U}_s q$ using Sufficient Conditions for ECI via Flows}

The following result follows from Proposition \ref{prop:pre_ECI_f}.

\begin{theorem} [$p \,\mathcal{U}_s q$ using ECI via flows]
\label{thm:strong_ECI_f}
Consider a hybrid system $\mathcal{H} \!=\! (C,F,D,G)$.
Given atomic propositions $p$ and $q$, let the sets $P$ and $Q$ be as in \eqref{eqn:K_sets} \blue{such that Assumptions \ref{assump:SA1} and \ref{assump:SA} hold}.
Let $C_s$ and $D_s$ be as in \eqref{eqn:H_m+}.
Then, the formula $p \,\mathcal{U}_s q$ is satisfied for $\mathcal{H}$ if the following hold:
\begin{itemize}
    \item[1)] The formula $p \,\mathcal{U}_w q$ is satisfied for $\mathcal{H}$.
    
    \item[2)] There exist a $\mathcal{C}^1$ function $v : \mathbb{R}^n \rightarrow \mathbb{R}$, a locally Lipschitz function $f_c : \mathbb{R} \rightarrow \mathbb{R}$, and a constant $r_1 > 0$ such that the following hold:
	\begin{itemize}
        \item[2.1)] $\left< \nabla v(x), \eta \right> \leq f_c(v(x))$ for all $x \in \cl(C_s)$ and for all $\eta \in F(x) \cap T_{\cl(C_s)}(x)$;
        \item[2.2)] $v(\eta) \leq v(x)$ for all $x \in D \cap P$ and for all $\eta \in G(x)$;
        \item[2.3)]
		The solutions $\dot{y} = f_c(y)$
        starting from $v(P \backslash Q)$ converge to $(-\infty, r_1)$ in finite time, and $S_1 := \{ x \in \cl(C_s) : v(x) < r_1 \} \subset Q$.
	\end{itemize}
	\item[3)] For each solution $\phi \in \mathcal{S_H}(P \backslash Q)$, there exists a solution $y$ to $\dot{y} = f_c(y)$ starting from $v(\phi(0,0))$ such that there exists $t^\star \geq 0$ satisfying:
\begin{align} \label{eqtstar}
\!\!\!\!\!\!\!\!t^\star \!\leq\! \sup \{ t : (t,j) \!\in\! \dom\phi \}\mbox{,} ~y(t) \!\in\! (-\infty, r_1) ~\forall t \geq t^\star\mbox{.}
\end{align}
\end{itemize}
\end{theorem}
\begin{proof}
Consider the system $\mathcal{H}_s$ introduced in \eqref{eqn:H_m+}.
Using Theorem \ref{thm:stronguntil_eci}, we show that $Q$ is ECI with respect to $P \backslash Q$ for $\mathcal{H}_s$ to conclude that $p \,\mathcal{U}_s q$ is satisfied for $\mathcal{H}$.
To this end, we show that $Q$ is ECI with respect to $P \backslash Q$ for $\mathcal{H}_s$ via Theorem \ref{thm:ECIf}.
First, we show that items 1) and 3) in Proposition \ref{prop:pre_ECI_f}, required in Theorem \ref{thm:ECIf}, hold for $\mathcal{H}_s$.
Notice that under item 2), item 1) in Proposition \ref{prop:pre_ECI_f} holds for $\mathcal{H}_s$.
Moreover, item 3) in Proposition \ref{prop:pre_ECI_f} is verified for $\mathcal{H}_s$ \blue{since the jumps from $Q$ remain in $Q$ due to the definition of the jump map $G_w$}, which is $G_w(x) \!=\! x$ for all $x \!\in\! Q$.
Finally, to show that $Q$ is ECI with respect to $P \backslash Q$ for $\mathcal{H}_s$ via Theorem \ref{thm:ECIf}, we show that, for each maximal solution $\phi$ to $\mathcal{H}_s$ starting from $P \backslash Q$, there exists a solution $y$ to $\dot{y} = f_c(y)$ starting from $v(\phi(0,0))$ satisfying $y(t) \!\in\! (-\infty, r_1]$ for all $t \geq t^\star$, for some nonnegative $t^\star \leq \sup \{ t : (t,j) \in \dom\phi \}$.
To show this, we first use item 1) and the construction of $\mathcal{H}_s$, to conclude that, each maximal solution $\phi$ to $\mathcal{H}_s$ starting from $P \backslash Q$ remains in $P \cup Q$. Hence, either $\phi$ reaches $Q$ in finite time, or $\phi$ remains in $P \backslash Q$.
Next, by Lemma \ref{lemma:H_Hs}, we conclude that $\phi$ is a maximal solution to $\mathcal{H}$.
Finally, using item 3), we conclude the existence of a solution $y$ to $\dot{y} = f_c(y)$ starting from $v(\phi(0,0))$ such that, for some $t^\star \geq 0$, \eqref{eqtstar} holds; and thus, we conclude that $Q$ is ECI with respect to $P \backslash Q$ for $\mathcal{H}_s$ via Theorem \ref{thm:ECIf}.
Therefore, via Theorem \ref{thm:stronguntil_eci}, the formula $p \,\mathcal{U}_s q$ is satisfied for $\mathcal{H}$, which completes the proof.
\end{proof}

\nDraft{\vspace{-0.15in}}
\subsection{Certifying 
$p \,\mathcal{U}_s q$ using Sufficient Conditions for 
ECI via Jumps}

The following result follows Proposition \ref{prop:pre_ECI_j}.

\begin{theorem} [$p \,\mathcal{U}_s q$ using ECI via jumps] \label{thm:strong_ECI_j}
Consider a hybrid system $\mathcal{H} = (C,F,D,G)$.
Given atomic propositions $p$ and $q$, let the sets $P$ and $Q$ be as in \eqref{eqn:K_sets} \blue{such that Assumptions \ref{assump:SA1} and \ref{assump:SA} hold}.
Let $C_s$ and $D_s$ be as in \eqref{eqn:H_m+}.
Then, the formula $p \,\mathcal{U}_s q$ is satisfied for $\mathcal{H}$ if the following hold:
\begin{itemize}
    \item[1)] The formula $p \,\mathcal{U}_w q$ is satisfied for $\mathcal{H}$.
	\item[2)] There exist a $\mathcal{C}^1$ function $w : \mathbb{R}^n \rightarrow \mathbb{R}$, $f_d : \mathbb{R} \rightarrow \mathbb{R}$ which is nondecreasing, and a constant $r_2 > 0$ such that the following hold:
	\begin{itemize}
		\item[2.1)] $\left< \nabla w(x), \eta \right> \leq 0$ for all $x \in \cl(C_s)$ and for all $\eta \in F(x) \cap T_{\cl(C_s)}(x)$;
		\item[2.2)] $w(\eta) \!\leq\! f_d(w(x))$ for all $x \!\in\! D \cap P$ and for all $\eta \!\in\! G(x)$;
	    \item[2.3)]
	    The solutions to $z^+ = f_d(z)$ starting from 
	    $w(P \backslash Q)$ converge to $(-\infty, r_2)$ in finite time, and $S_2 := \{ x \in D_s \cup \cl(C_s) : w(x) < r_2 \} \subset Q$.
	\end{itemize}
	\item[3)] For each solution $\phi \in \mathcal{S_H}(P \backslash Q)$, there exists a solution $z$ to $z^+ = f_d(z)$ starting from $v(\phi(0,0))$ such that there exists $j^\star \in \mathbb{N}$ satisfying: 
$$ \!\!j^\star \leq \sup \{ j : (t,j) \in \dom\phi \}, ~ z(j) \in (-\infty, r_2) ~\forall j \geq j^\star\mbox{.} $$
\end{itemize}
\end{theorem}
\begin{proof}
The proof follows the exact same steps used to prove Theorem \ref{thm:strong_ECI_f} while using Theorem \ref{thm:ECIj} instead of Theorem \ref{thm:ECIf}.
\end{proof}

\subsection{Certifying 
$p \,\mathcal{U}_s q$ using Sufficient Conditions for 
ECI via Approximate Flow Length}

Next, we employ the conditions for pre-ECI in Theorem~\ref{thm:pre_ECI_interval} for hybrid systems when we know the lengths of the flow interval between each successive jumps approximately.

\begin{theorem}[$p \,\mathcal{U}_s q$ using ECI and  Theorem~\ref{thm:pre_ECI_interval} $+$ Theorem~\ref{thm:eventual_ci}]
\label{thm:strong_ECI_interval}
Consider a hybrid system $\mathcal{H} = (C,F,D,G)$. Given atomic propositions $p$ and $q$, let the sets $P$ and $Q$ be as in \eqref{eqn:K_sets} \blue{such that Assumptions \ref{assump:SA1} and \ref{assump:SA} hold}. Let $K:= \text{int}(Q)$ and $\mathcal{I}_{P,K}$ be as in Definition \ref{def:flow_length}, and let $\tau_M := \sup\mathcal{I}_{P,K}$.
Then, the formula $p \,\mathcal{U}_s q$ is satisfied for $\mathcal{H}$ if the following hold:

\begin{itemize}
    \item[1)] The formula $p \,\mathcal{U}_w q$ is satisfied for $\mathcal{H}$.
    \item[2)] There exist a $\mathcal{C}^1$ function $v : \mathbb{R}^n \rightarrow \mathbb{R}$, a locally Lipschitz function $f_c : \mathbb{R} \rightarrow \mathbb{R}$, and a function $f_d : \mathbb{R} \rightarrow \mathbb{R}$ which is nondecreasing such that
	\[
	\begin{split}
	\!\!\!\!\left< \nabla v(x), \eta \right> \!\leq\! f_c(v(x)) & \:\:\forall x \!\in\! (C \backslash K) \cap P, \forall \eta \!\in\! F(x) \!\cap\! T_C(x)\mbox{,}\\ 
	\!\!\!\!v(\eta) \!\leq\! f_d(v(x))\qquad\:\: & \:\:\forall x \in (D \backslash K) \cap P, \forall \eta \in G_w(x)\mbox{.}
	\end{split}
	\]
	\item[3)] There exists a constant $r > 0$ such that 
	\[
	S := \{ x \in (C \cup D) \cap (P\cup Q) : v(x) < r \} \subset Q\mbox{.}
	\]
	\item[4)] The maximal solutions to the reduced hybrid system $\mathcal{H}_{r}$ starting from $v(P \backslash K) \times \{ 0 \}$ converge to $(-\infty, r] \times \mathbb{R}_{\geq 0}$ in finite time, where  
\begin{equation*}
\!\!\!\mathcal{H}_r:
\left\{\!\!\!
\begin{array}{cll}
	\Big[\!
	\begin{array}{c}	
 		\dot{y}\\
 		\dot{\tau}
 	\end{array}
 	\!\Big]
 	\!\!\!\!\!&=\!
 	\Big[\!
 	\begin{array}{c}
 		f_c(y)\\
 		1
	\end{array}
	\!\Big]
	&~(y, \tau) \in \mathbb{R} \!\times\! [0, \tau_M]\mbox{,}
	\\
	\Big[\!
	\begin{array}{c}
		y^+  \\
		\tau^+ 
	\end{array}
	\!\Big]
	\!\!\!\!\!&=\!
	\Big[\!
	\begin{array}{c}
		f_d(x)\\
		0
	\end{array}
	\!\Big]
	&~(y, \tau) \in \mathbb{R} \!\times\! \mathcal{I}_{P,K}\mbox{.}
\end{array}	
\right.
\end{equation*}
    \item[5)]  No maximal solution to $\mathcal{H}$ has a finite escape time in $(P\backslash Q) \cap C$.
    \item[6)] Every maximal solution to $\mathcal{H}$ from $((P\backslash Q) \cap \partial C) \backslash D$ is nontrivial.
\end{itemize}
\end{theorem}

\begin{proof}
Consider the system $\mathcal{H}_s$ introduced in \eqref{eqn:H_m+}. Using Theorem \ref{thm:pre_ECI_interval} for $\mathcal{H}_s$ under items 2), 3), and 4), we conclude that $Q$ is pre-ECI with respect to $P \backslash Q$ for $\mathcal{H}_s$. The rest of the proof follows using the same steps in the proof of Theorem \ref{thm:strong_until}.
\end{proof}

\subsection{Certifying 
$p \,\mathcal{U}_s q$ using Sufficient Conditions for FTA via Flows and Jumps}

In the following, along the lines of Remark \ref{remark:eci_fta},
we propose sufficient conditions that guarantee the satisfaction of the formula $p \,\mathcal{U}_s q$ by employing sufficient conditions for FTA in Theorem~\ref{thm:pre_FTA}, Theorem~\ref{thm:FTAf}, and Theorem~\ref{thm:FTAj}, respectively.

\begin{theorem}
[$p \,\mathcal{U}_s q$ using FTA]
\label{thm:strong_FTA}
Consider a hybrid system $\mathcal{H}=(C,F,D,G)$.
Given atomic propositions $p$ and $q$, let the sets $P$ and $Q$ as in \eqref{eqn:K_sets} be \blue{such that Assumptions \ref{assump:SA1} and \ref{assump:SA} hold}. For $\mathcal{N}$ an open neighborhood of $Q$, we suppose that there exist functions $V: \mathcal{N} \rightarrow \mathbb{R}_{\geq 0}$ and $W: \mathcal{N} \rightarrow \mathbb{R}_{\geq 0}$ that are
positive definite with respect to $Q$ and such that $P \backslash Q \subset L_V(r) \cap (C \cup D) \subset \mathcal{N}$ and $P \backslash Q \subset L_W(r) \cap (C \cup D) \subset \mathcal{N}$, for some $r > 0$. Then, the formula $p \,\mathcal{U}_s q$ is satisfied for $\mathcal{H}$ if the following hold:
\begin{itemize}
	\item[1)] The formula $p \,\mathcal{U}_w q$ is satisfied for $\mathcal{H}$.
	\item[2)] There exist constants $c_1 > 0$ and $c_2 \in [0,1)$ such that
\begin{equation}
\label{eqFTAflow}
\begin{split}
&\!\!\!\!\!\!\!\!\left< \nabla V(x), \eta \right> \leq -c_1 V^{c_2}(x)\\
&\qquad\quad\forall x \!\in\! (C \cap \mathcal{N} \cap P) \backslash Q, \forall \eta \!\in\! F(x) \cap T_C(x)\mbox{,}\\
&\!\!\!\!\!\!\!\!V(\eta) - V(x) \leq 0 \:\:\forall x \!\in\! (D \cap \mathcal{N} \cap P) \backslash Q, \forall \eta \!\in\! G(x)\mbox{.}
\end{split}
\end{equation}

\item[3)] There exists a constant $c > 0$ such that
\begin{equation}
\label{eqFTAjump}
\begin{split}
&\!\!\!\!\!\!\!\!\left< \nabla W(x), \eta \right> \!\leq\! 0\\
&\qquad\quad\forall x \!\in\! (C \cap \mathcal{N} \cap P) \backslash Q, \forall \eta \!\in\! F(x) \cap T_C(x)\mbox{,}\\
&\!\!\!\!\!\!\!\!W(\eta) - W(x) \!\leq\! -\min\{c, W(x)\}\\
&\qquad\quad\forall x \!\in\! (D \cap \mathcal{N} \cap P) \backslash Q, \forall \eta \!\in\! G(x)\mbox{.}
\end{split}
\end{equation}
    \item[4)] No maximal solution to $\mathcal{H}$ has a finite escape time in $(P\backslash Q) \cap C$.
    \item[5)] Every maximal solution to $\mathcal{H}$ from $((P\backslash Q) \cap \partial C) \backslash D$ is nontrivial.
\end{itemize}
\end{theorem}

\begin{proof}
Consider the system $\mathcal{H}_s$ introduced in \eqref{eqn:H_m+}. Using \cite[Theorem 5.1]{han.arXiv21} for $\mathcal{H}_s$ under items 2) and 3), we conclude that $Q$ is pre-FTA with respect to $P \backslash Q$ for $\mathcal{H}_s$. Next, we show that $P \cup Q$ is forward invariant for $\mathcal{H}_s$ exactly as we did in the proof of Theorem \ref{thm:strong_until}. Thus, using Theorem \cite[Theorem 5.10]{han.arXiv21}, we conclude that $Q$ is FTA with respect to $P \backslash Q$ for $\mathcal{H}_s$. Finally, the proof is completed using Theorem \ref{thm:stronguntil_fta}.
\end{proof}

\begin{example}
[Thermostat]
Consider the hybrid system $\mathcal{H} = (C,F,D,G)$ with the state $x \!:=\! (h,z) \!\in\! \{0,1\} \!\times\! \mathbb{R}$ in Example \ref{exp:thermostat}.
Following the formulation therein, a specification of interest is
that the room temperature decreases once it exceeds $z_{\max}$, which is related to the satisfaction of $p \,\mathcal{U}_s q$ for $\mathcal{H}$
with the propositions $p$ and $q$ given as in Example \ref{exp:thermostat}.
The sets $P$ and $Q$ in \eqref{eqn:K_sets} are given by $P \!=\! \{1\} \!\times\! (-\infty, z_{\max}]$ with $z_{\max} > 0$ and $Q \!=\! \{0\} \!\times\! [z_{\min}, \infty)$.
To show the satisfaction of $p \,\mathcal{U}_s q$ for $\mathcal{H}$, we apply Theorem \ref{thm:strong_FTA}.
First, consider the barrier function candidate $B(x) \!:=\! (-1)^h(z_{\max} \!-\! z)$.
Indeed, $B$ is a barrier function candidate with respect to $(P \backslash Q, (\{0,1\}\times \mathbb{R}) \backslash (P \cup Q))$ for $\mathcal{H}$ since
for all $x \!\in\! P \backslash Q \!=\! P$, $B(x) \!=\! z-z_{\max} \!\leq\! 0$; and for all $x \!\in\! (C \cup D) \backslash (P \cup Q) \!=\! (\{0\} \!\times\! (-\infty,z_{\min})) \!\cup\! (\{1\} \!\times\! (z_{\max}, \infty))$, $B(x) \!>\! 0$.
Moreover, for all $x \!\in\! C \backslash Q \!=\! \{1\} \!\times\! (-\infty, z_{\max}]$, $B(x) \!=\! z \!-\! z_{\max} \!\leq\! 0$; and thus, $(C \backslash Q) \cap (U(\partial K) \backslash K)\!=\! \emptyset$. Furthermore, for all $x \!\in\! K \cap (D \backslash Q) \!=\! \{(1, z_{\max})\}$, $G(x) \!=\! \{(0, z_{\max})\}  \!\subset\! C \cup D \cup Q$ and $B(G(x)) \!=\! 0$; hence, items 1) - 3) in Theorem \ref{thm:weak_until} hold. It follows that the formula $p \,\mathcal{U}_w q$ is satisfied for $\mathcal{H}$; and thus, item 1) in Theorem \ref{thm:strong_FTA} holds.
Next, consider the functions $V(x) \!=\! W(x) \!=\! z_{\min} - z$.
For all $x \!\in\! (C \cap P) \backslash Q \!=\! \{1\} \!\times\! (-\infty, z_{\max}]$, $\left< \nabla V(x), F(x) \right> \!=\! z \!-\! z_o \!-\! z_{\triangle} \!\leq\! z_{\max} \!-\! z_o \!-\! z_{\triangle}$.
Moreover, for all $x \!\in\! (D \cap P) \backslash Q \!=\! \{(1, z_{\max})\}$, $V(G(x)) - V(x) = 0$. 
Hence, item 2) holds for $c_1 \!=\! z_{\max} \!-\! z_o \!-\! z_{\triangle}$ and $c_2 \!=\! 0$.
On the other hand,
for all $x \!\in\! (C \cap P) \backslash Q$, $\left< \nabla W(x), F(x) \right> \!=\! z \!-\! z_o \!-\! z_{\triangle} < 0$ since $z_o \!+\! z_{\triangle} > z_{\max}$; and for all $x \!\in\! (D \cap P) \backslash Q \!=\! \{(1, z_{\max})\}$, $W(G(x)) - W(x) \!=\! 0 \!<\! -W(x)$. Hence, we conclude that item 3) holds for $c \!=\! z_{\max} - z_{\min}$.
Furthermore, since the flow map $F$ has global linear growth on $(P \backslash Q) \cap C$, the solutions to $\mathcal{H}$ do not have a finite escape time inside $(P \backslash Q) \cap C$; hence, item 4) holds.
Finally, since $(P \backslash Q) \cap \partial C = \emptyset$, item 5) holds.
Thus, Theorem \ref{thm:strong_FTA} implies that $p \,\mathcal{U}_s q$ is satisfied for $\mathcal{H}$.
\hfill $\triangle$
\end{example}

\subsection{Certifying 
$p \,\mathcal{U}_s q$ using Sufficient Conditions for FTA via Flows}

\begin{theorem}[$p \,\mathcal{U}_s q$ using FTA and Theorem~\ref{thm:FTAf}]
\label{thm:strong_FTA_f}
Consider a hybrid system $\mathcal{H}=(C,F,D,G)$.
Given atomic propositions $p$ and $q$ and let the sets $P$ and $Q$ as in \eqref{eqn:K_sets} \blue{such that Assumptions \ref{assump:SA1} and \ref{assump:SA} hold}. Let $\mathcal{N}$ be an open neighborhood around $Q$ such that there exists a $\mathcal{C}^1$ function $V: \mathcal{N} \rightarrow \mathbb{R}_{\geq 0}$ that is positive definite with respect to $Q$ and $P \backslash Q \subset L_V(r) \cap (C \cup D) \subset \mathcal{N}$, for some $r>0$. Then, the formula $p \,\mathcal{U}_s q$ is satisfied for $\mathcal{H}$ if 
\begin{itemize}
	\item[1)] The formula $p \,\mathcal{U}_w q$ is satisfied for $\mathcal{H}$.
	\item[2)] There exist constants $c_1 > 0$ and $c_2 \in [0,1)$ such that \eqref{eqFTAflow} holds. 
	\item[3)] For every $x \in P \backslash Q$, each solution $\phi \in \mathcal{S_H}(x)$ satisfies
\begin{align} \label{eqflowenough}
		\tfrac{V^{1-c_2} (x)}{c_1 (1-c_2)} \leq \sup \{ t : (t,j) \in \dom\phi \} \mbox{.}
\end{align}
\end{itemize}
\end{theorem}

\begin{proof}
Consider system $\mathcal{H}_s$ introduced in \eqref{eqn:H_m+}. Using item 1), we conclude that a maximal solution $\psi$ to $\mathcal{H}$ starting from $P \backslash Q$ either remains in $P \backslash Q$ for all hybrid time, otherwise, $\psi$ remains in $P \backslash Q$ up to when it reaches $Q$. Hence, each maximal solution $\phi$ to $\mathcal{H}_s$ starting from $P \backslash Q$ remains in $P \cup Q$. In particular, either $\phi$ reaches $Q$ in finite time, or $\phi$ remains in $P \backslash Q$. 
To exclude the latter case, we show that, when $\phi$ remains in $P \backslash Q$, $\phi$ must be a maximal solution to $\mathcal{H}$. Indeed, assume the existence of a solution $\psi$ to $\mathcal{H}$ which is a nontrivial extension of $\phi$; namely, there exists $I \subset \mathbb{R}_{\geq 0} \!\times\! \mathbb{N}$ such that 
$I \!\neq\! \emptyset$ and  
$\dom \phi \cup I \!=\! \dom \psi$. Note that $\psi (\dom \phi) \!=\! \phi (\dom \phi) \subset P \backslash Q$. Also, since $\psi$ must remain in $P \backslash Q$ up to when it reaches $Q$, we can choose $I$ such that  
$\psi (\dom \phi \cup I) \subset P \backslash Q$. Hence, $\psi$ is a solution to $\mathcal{H}_s$, which contradicts the fact that $\phi$ is a maximal solution to $\mathcal{H}_s$. 
Next, using item 3), we conclude the existence of a solution $y$ to $\dot{y} = f_c(y)$ starting from $v(\phi(0,0))$ such that, for some $t^\star \geq 0$, \eqref{eqflowenough} holds. Combining the latter fact to item 2) and using Theorem \ref{thm:FTAf} for $\mathcal{H}_s$, we conclude that $\phi$ must reach $Q$ in finite time. Hence, $Q$ is FTA with respect to $P \backslash Q$ for $\mathcal{H}_s$. 
 Finally, the proof is completed using Theorem \ref{thm:stronguntil_fta}.
\end{proof}

\subsection{Certifying 
$p \,\mathcal{U}_s q$ using Sufficient Conditions for FTA via Jumps}

\begin{theorem}[$p \,\mathcal{U}_s q$ using FTA and Theorem~\ref{thm:FTAj}]
\label{thm:strong_FTA_j}
Consider a hybrid system $\mathcal{H}=(C,F,D,G)$.
Given atomic propositions $p$ and $q$, let the sets $P$ and $Q$ as in \eqref{eqn:K_sets} such that
Suppose \blue{such that Assumptions \ref{assump:SA1} and \ref{assump:SA} hold}.
Let $\mathcal{N}$ be an open neighborhood around $Q$ and  suppose that there exists a $\mathcal{C}^1$ function $W: \mathcal{N} \rightarrow \mathbb{R}_{\geq 0}$ that is positive definite with respect to $Q$ such that $P \backslash Q \subset L_W(r) \cap (C \cup D) \subset \mathcal{N}$, for some $r >0$. Then, the formula $p \,\mathcal{U}_s q$ is satisfied for $\mathcal{H}$ if the following hold:
\begin{itemize}
	\item[1)] The formula $p \,\mathcal{U}_w q$ is satisfied for $\mathcal{H}$.
	\item[2)] There exists a constant $c > 0$ such that \eqref{eqFTAjump} holds. 

    \item[3)] For every $x \in P \backslash Q$, each solution $\phi \in \mathcal{S}_{\mathcal{H}}(x)$ satisfies
$$
    	\mbox{ceil} \big( \tfrac{W(x)}{c} \big) \leq \sup \{ j : (t,j) \in \dom\phi \} \mbox{.}
$$
\end{itemize}
\end{theorem}

\begin{proof}
The proof follows the exact same steps used to prove Theorem \ref{thm:strong_FTA_f} while using Theorem \ref{thm:FTAj} instead of Theorem \ref{thm:FTAf}.
\end{proof}

\section{Conclusion}
In this paper, tools are introduced for certifying basic LTL specifications involving until operators for hybrid systems. For such systems, sufficient and equivalence relationships are established between the satisfaction of the considered formulas and some of the existing invariance and attractivity notions studied in control literature. In particular, CI, ECI, and FTA notions are revisited in this paper in the context of hybrid systems. Furthermore, sufficient conditions certifying these invariance and attractivity properties are proposed. As a consequence, sufficient conditions (not involving the computation of the systems' solutions) guaranteeing the satisfaction of the considered formulas are proposed. \blue{Future research direction includes analyzing more complex specifications, where the until operator is involved in addition to other operators}. \red{In particular, with the proposed tools, more complex formulae can be certified through decomposition by building a finite state automaton; see \cite[Section 6.5]{han2020linear} and the references therein.}

\appendix
%\appendices
\setcounter{theorem}{0}
    \renewcommand{\thetheorem}{\Alph{section}.\arabic{theorem}}
\subsection{Auxiliary Results}

The following result is a version of the well-known comparison lemma that can be found in \cite[Lemma 3.4]{khalil}.
\begin{lemma} \label{lemcomp}
Consider the scalar system $\dot{u} = f(t, u)$, $u(t_0) = u_0$,
where for all $t \geq 0$ and all $u \in S \subset \mathbb{R}$, $f(t,u)$ is continuous in $t$ and locally Lipschitz in $u$. Furthermore, let $[t_0, T)$ be the maximal interval, $T$ can be infinity, of existence of the solution $t \mapsto u(t)$. Moreover, suppose that $u(t) \in S$ for all $t \in [t_0, T)$. 
Let $t \mapsto v(t)$ be a continuous function such that $v(t_0) \leq u_0$,
$v(t) \in S$ for all $t \in [t_0, T)$, and its upper right-hand derivative, $D^+ v(t)$, satisfies the following differential inequality, for almost all $t \in [t_0, T)$:
\begin{equation*}
    D^+ v(t) := \limsup_{s \rightarrow 0^+} \tfrac{v(t+s)-v(t)}{s} \leq f(t, v(t)).
\end{equation*}
 Then, $v(t) \leq u(t)$ for all $t \in [t_0, T)$.
\label{lem:comparison}
\end{lemma}

\begin{lemma} \label{lemaux1}
Assume that the function $t \mapsto v(t)$ in Lemma \ref{lem:comparison} satisfies $v(t) = v(x(t))$ for all $t \in [t_0, T)$ with $t \mapsto x(t)$ being a solution to the system 
$\dot{x} \in F(x)$, ~$x \in C \subset \mathbb{R}^n$,
and $v \in \mathcal{C}^1$, it follows that, 
for almost all $t \in [t_0, T)$,
$
 D^+ v(t) = \dot{v}(t) = \langle \nabla v(x(t)), \dot{x}(t) \rangle.
$
\end{lemma}

\begin{proof}
Since the solution $x$ is absolutely continuous, it follows that $\dot{x}(t)$ exists for almost all $t \in [t_0, T)$. Furthermore, since $v \in \mathcal{C}^1$. Hence, 
$\dot{v}(t)$ exists for almost all $t \in [t_0, T)$. 
Let $t \in [t_0, T)$ such that $\dot{v}(t)$ exists. Then, by definition of the time derivative, we conclude that
$\dot{v}(t) = \lim_{s \rightarrow 0} \tfrac{v(t+s)-v(t)}{s} = \limsup_{s \rightarrow 0^+} \tfrac{v(t+s)-v(t)}{s} = D^+ v(t)$.
Furthermore, using the classical chain rule for composition of differentiable functions, we conclude that 
$ \dot{v}(t) = \langle \nabla v(x(t)), \dot{x}(t) \rangle$ for almost all $t \in [t_0, T)$.
\end{proof}

\begin{lemma} \label{lemaux2}
Let $x: [t_0,T) \!\rightarrow\! \mathbb{R}^n$ be a solution to the following constrained differential inclusion 
$\dot{x} \!\in\! F(x)$ for all $x \!\in\! C \!\subset\! \mathbb{R}^n$.
Then, for almost all $t \in [t_0, T)$, $\dot{x}(t) \in T_C(x(t)).$
\end{lemma}
\begin{proof}
Let $t  \in [t_0, T)$ such that $\dot{x}(t)$ exists; thus, $\dot{x}(t) \in F(x(t,j))$. Furthermore, let a sequence $\left\{t_n\right\}_{n \in \mathbb{N}} \subset (t_0, T - t) $ such that $t_n \rightarrow 0$. 
That is, for $v_n(t) := (x(t_n) - x(t)) / t_n$, we have 
$\lim_{n \rightarrow \infty} v_n(t) = \dot{x}(t)$ and at the same time, by definition of solution to $\dot{x} \in F(x)$,
$x(t) + t_n v_n(t) = x(t_n) \in C $. Hence, using \eqref{eq.conti}, we conclude that 
$\dot{x}(t) \in T_C(x(t))$.
\end{proof}

\bibliographystyle{IEEEtran}
\bibliography{LTL_until}

\end{document}